\documentclass[a4paper,12pt,fleqn]{article}
\usepackage{amsmath,amssymb}
\usepackage{amsthm}
\usepackage{graphicx}
\usepackage{float}
\usepackage[np]{numprint}
\usepackage{calc}    
\usepackage{tikz,tkz-tab}
\usepackage{braket}
\usepackage{varioref}
\usepackage{hyperref}
\usepackage{lipsum}
\usepackage[left=3cm,right=3cm,top=2.5cm,bottom=4cm]{geometry}
\usepackage[font=footnotesize,labelfont=bf]{caption}

\newcommand\SLOCC{\mathtt{SLOCC}}
\newcommand\W{\mathtt{W}}
\newcommand\GHZ{\mathtt{GHZ}}
\newcommand\swap{\mathtt{SWAP}}
\newcommand\cnot{\mathtt{CNOT}}
\newcommand\cX{c-\mathtt{X}}

\newcommand\cZ{c-\mathtt{Z}}
\newcommand\PauliX{ Pauli-$\mathtt X$ }
\newcommand\PauliY{ Pauli-$\mathtt Y$ }
\newcommand\PauliZ{ Pauli-$\mathtt Z$ }
\newcommand\XX{ \mathtt{X} }

\newcommand\ZZ{ \mathtt{Z} }

\newcommand\HH{\mathcal{H}}
\newcommand\R{\mathbb{R}}

\newcommand\C{\mathbb{C}}
\newcommand\F{\mathbb{F}}

\newcommand\ee{\mathrm e}
\newcommand\ii{\mathrm i}

\newcommand\SYM[1][n]{\mathfrak{S}_{#1}}
\newcommand\XG[1][n]{\left<\cnot\right>_{#1}}
\newcommand\SL[1][n]{\mathrm{SL}_{#1}(\mathbb{F}_2)}
\newcommand\GL[1][n]{\mathrm{GL}_{#1}(\mathbb{F}_2)}
\newcommand\SG[1][n]{\mathcal{S}_{#1}}

\renewcommand\phi{\varphi}
\renewcommand\epsilon{\varepsilon}
\renewcommand\geq{\geqslant}
\renewcommand\leq{\leqslant}

\newtheorem{theo}{Theorem}
\newtheorem{prop}[theo]{Proposition}
\newtheorem{cor}[theo]{Corollary}
\newtheorem{lem}[theo]{Lemma}
\newtheorem{algo}[theo]{Algorithm}

\newtheorem{conj}[theo]{Conjecture}

\begin{document}
\floatstyle{boxed} 
\restylefloat{figure}
\title{Quantum circuits of $\cnot$ gates }

\author{Marc Bataille \\ marc.bataille1@univ-rouen.fr \\
  \\ LITIS laboratory, Universit\'e Rouen-Normandie \thanks{685 Avenue de l'Universit\'e, 76800 Saint-\'Etienne-du-Rouvray. France.}}

\date{}

  \maketitle
\begin{abstract}
  We study in detail the algebraic structures underlying quantum circuits generated by $\cnot$ gates. Our results allow us to propose polynomial-time heuristics to reduce the number of gates used in a given $\cnot$ circuit  and we also give algorithms to optimize this type of circuits in some particular cases. Finally we show how to create some usefull entangled states when a $\cnot$ circuit acts on a fully factorized state.
\end{abstract}

\section{Introduction}

The controlled \PauliX gate,  also called the $\cnot$ gate, is a very common and usefull gate in quantum circuits. This gate involves two qubits $i$ and $j$ in a $n$-qubit system.
One of the two qubits (say qubit $i$) is the target qubit whereas the other qubit  plays the role of control. When the control qubit $j$ is in the state
$\ket{1}$ then a \PauliX gate (\emph{i.e.} a $\mathtt{NOT}$ gate) is applied to the target qubit $i$ which  gets flipped. 
When qubit $j$ is in the state $\ket{0}$ nothing happens to qubit $i$.

Actually,  $\cnot$ gates are of crucial importance in the fields of Quantum Computation and Quantum Information. Indeed, it appears that single qubit unitary gates together with $\cnot$ gates constitute an universal set for quantum computation :  any arbitrary unitary operation on a $n$-qubit system can be implemented using only $\cnot$ gates and single qubit unitary gates (see \cite[Section 4.5.2]{NCI2011}  for a complete proof of this important result). As a consequence, many multiple-qubit gates are implemented in current experimental quantum machines using $\cnot$ gates plus other single-qubit gates. In Figure \ref{equi} we give a few classical examples of such an implementation : a $\swap$ gate can be simulated using 3 $\cnot$ gates, a controlled Pauli-Z gate can be implemented by means of 2 Hadamard single-qubit gates and one $\cnot$ gate. These implementations are used for instance in the IBM superconducting transmon device  (\url{www.ibm.com/quantum-computing/}).\medskip

From this universality result it is possible to show that any unitary operation can be approximated to arbitrary accuracy using $\cnot$ gates together with Hadamard, Phase, and  $\pi/8$  gates (see Figure \ref{univers} for a definition of these gates and \cite[Section 4.5.3]{NCI2011} for a proof of this result). This discrete set of gates is often called the standard set of universal gates. Using a discrete set of gates brings a great advantage in terms of reliability because it is possible to apply these gates in an error-resistant way
through the use of quantum error-correcting codes (again refer to \cite[Chapter 10]{NCI2011} for further information). Currently, important error rates in $\cnot$ gates implemented on experimental quantum computers are one of the main causes of their unreliability (see  \cite{2017Monroe,2019Wright} and Table \ref{errors}). So the ability to correct errors on quantum circuits is a key point for successfully building a functional and reliable quantum computer. A complementary approach to the Quantum Error-Correction is to improve reliability by minimizing the number of gates and particularly the number of two-qubit gates used in a given circuit. Our work takes place in this context : we show that a precise understanding of the algebraic structures underlying circuits of $\cnot$ gates makes it possible to reduce and, in some special cases, to minimize the number of gates in these circuits.\medskip

In this paper we also study the emergence of entanglement in  $\cnot$ circuits and we show that these circuits are a convenient tool for creating many types of entangled states. These particular states of a quantum system were first mentionned by Einstein, Podolsky and Rosen in their famous EPR article of 1935 \cite{EPR} and in Quantum Information Theory (QIT) they can be considered as a fundamental physical resource  (see e.g.  \cite{HHHH2009}). Entangled states turn out to be essential in many areas of QIT such as quantum error correcting codes \cite{2007GW}, quantum key distribution \cite{1991Ekert} or quantum secret sharing \cite{1999CGL}. Spectacular applications such as super-dense coding \cite{1992BWS} or quantum teleportation \cite{1994Vaidman,1998FSBFKP,2015PESWFB} are based on the use of classical entangled states. This significant role played by entanglement represents for us a good reason to understand how $\cnot$ circuits can be used to create entangled states.\medskip

The paper is structured  as follows. Section \ref{background} is mainly a background section where we recall some classical notions in order to guide non-specialist
readers. In Section \ref{algebra} we investigate the algebraic structure of the group generated by $\cnot$ gates.
Section \ref{general} will be dedicated to the optimization problem of $\cnot$ circuits in the general case while in Section \ref{subgroups} we deal with optimization in particular subgroups of the group generated by $\cnot$ gates.
Finally in Section \ref{entanglement} we study how to use $\cnot$ circuits to create certain usefull entangled states.

\begin{table}[h]
  \begin{center}
  $\begin{array}{|c|c|c|c|c|c|}\hline
    \cnot_{ij}&0&1&2&3&4\\\hline
    0&&1.035\times 10^{-2}&&& \\\hline
    1&1.035\times 10^{-2}&&9.658\times 10^{-3}&& \\\hline
    2&&9.658\times 10^{-3}&&9.054\times 10^{-3}& \\\hline
    3&&&9.054\times 10^{-3}&&8.838\times 10^{-3}\\\hline
    4&&&&8.838\times 10^{-3}&\\\hline\hline
     \mathtt{H}_i&2.656\times 10^{-4}&3.675\times 10^{-4}&2.581\times 10^{-4}&3.572\times 10^{-4}&2.831\times 10^{-4}\\\hline
   \end{array}$
 \end{center} 
   { \caption{ Error rates for $\cnot$ gates acting on qubits i and j and for Hadamard gates acting on qubit i, where $i,j\in\{0,1,2,3,4\}$.
     The data comes from ibmq\_rome 5-qubit quantum computer after a calibration on May 12, 2020 and is publicly avalaible at \url{www.ibm.com/quantum-computing/}.\label{errors}} }
\end{table}

\section{Quantum circuits, $\cnot$ and $\swap$ gates\label{background}}

We recall here some definitions and basic facts about quantum circuits and $\cnot$ gates. We also introduce some notations
used in this paper and we add some developments about $\swap$ gates circuits. For a comprehensive introduction to quantum circuits the reader may refer to \cite[Chapter 4]{NCI2011}
or, for a shorter but self-contained introduction), to our last article \cite{MBJGL2019}.

In QIT, a qubit is a quantum state that represents the basic information
storage unit. This state is described by a ket vector in the Dirac notation $\ket{\psi} = a_0 \ket{0} + a_1\ket{1}$ where $a_0$ and $a_1$ are complex numbers
such that $|a_0|^2 + |a_1|^2= 1$. The value of $|a_i|^2$ represents the probability that measurement produces
the value $i$. The states $\ket{0}$ and $\ket{1}$ form a basis of the Hilbert space $\HH\simeq \C^2$ where a one qubit quantum system evolves.

Operations on qubits must preserve the norm and are therefore described by unitary operators.
In Quantum Computation, these operations are represented by quantum gates and a quantum circuit is a conventional representation of the sequence  of quantum gates applied to the qubit register over time. In Figure \ref{single} we recall the definition of the Pauli gates : notice that the states $\ket{0}$
and $\ket{1}$ are eigenvectors of the \PauliZ operator respectively associated to the eigenvalues 1 and -1, \emph{i.e.} $\ZZ\ket{0}=\ket{0}$ and $\ZZ\ket{1}=-\ket{1}$.
Hence the computational standard basis $(\ket{0},\ket{1})$ is also called the Z-basis.
\begin{figure}[h]
	\begin{center}
	\begin{tikzpicture}[scale=1.500000,x=1pt,y=1pt]
\filldraw[color=white] (0.000000, -7.500000) rectangle (24.000000, 7.500000);
\draw[color=black] (0.000000,0.000000) -- (24.000000,0.000000);
\draw[color=black] (0.000000,0.000000) node[left] {\PauliX\ \ };
\draw (45.00, 0.00) node {$\left[\begin{array}{cc}0&1\\1&0\end{array}\right]$};
\begin{scope}
\draw[fill=white] (12.000000, -0.000000) +(-45.000000:8.485281pt and 8.485281pt) -- +(45.000000:8.485281pt and 8.485281pt) -- +(135.000000:8.485281pt and 8.485281pt) -- +(225.000000:8.485281pt and 8.485281pt) -- cycle;
\clip (12.000000, -0.000000) +(-45.000000:8.485281pt and 8.485281pt) -- +(45.000000:8.485281pt and 8.485281pt) -- +(135.000000:8.485281pt and 8.485281pt) -- +(225.000000:8.485281pt and 8.485281pt) -- cycle;
\draw (12.000000, -0.000000) node {$X$};
\end{scope}

\end{tikzpicture}

	\begin{tikzpicture}[scale=1.500000,x=1pt,y=1pt]
\filldraw[color=white] (0.000000, -7.500000) rectangle (24.000000, 7.500000);
\draw[color=black] (0.000000,0.000000) -- (24.000000,0.000000);
\draw[color=black] (0.000000,0.000000) node[left] {\PauliY\ \ };
\draw (45.00, 0.00) node {$\left[\begin{array}{cc}0&-i\\i&0\end{array}\right]$};
\begin{scope}
\draw[fill=white] (12.000000, -0.000000) +(-45.000000:8.485281pt and 8.485281pt) -- +(45.000000:8.485281pt and 8.485281pt) -- +(135.000000:8.485281pt and 8.485281pt) -- +(225.000000:8.485281pt and 8.485281pt) -- cycle;
\clip (12.000000, -0.000000) +(-45.000000:8.485281pt and 8.485281pt) -- +(45.000000:8.485281pt and 8.485281pt) -- +(135.000000:8.485281pt and 8.485281pt) -- +(225.000000:8.485281pt and 8.485281pt) -- cycle;
\draw (12.000000, -0.000000) node {$Y$};
\end{scope}

\end{tikzpicture}

	\begin{tikzpicture}[scale=1.500000,x=1pt,y=1pt]
\filldraw[color=white] (0.000000, -7.500000) rectangle (24.000000, 7.500000);
\draw[color=black] (0.000000,0.000000) -- (24.000000,0.000000);
\draw[color=black] (0.000000,0.000000) node[left] {\PauliZ\ \ };
\draw (45.00, 0.00) node {$\left[\begin{array}{cc}1&0\\0&-1\end{array}\right]$};
\begin{scope}
\draw[fill=white] (12.000000, -0.000000) +(-45.000000:8.485281pt and 8.485281pt) -- +(45.000000:8.485281pt and 8.485281pt) -- +(135.000000:8.485281pt and 8.485281pt) -- +(225.000000:8.485281pt and 8.485281pt) -- cycle;
\clip (12.000000, -0.000000) +(-45.000000:8.485281pt and 8.485281pt) -- +(45.000000:8.485281pt and 8.485281pt) -- +(135.000000:8.485281pt and 8.485281pt) -- +(225.000000:8.485281pt and 8.485281pt) -- cycle;
\draw (12.000000, -0.000000) node {$Z$};
\end{scope}

\end{tikzpicture}\vspace{-4mm}

{ \caption{ The Pauli gates \label{single}}}
\end{center}
\end{figure}

A quantum system of two qubits $A$ and $B$ (also called  a two-qubit register) lives in a 4-dimensional Hilbert space $\HH_A\otimes\HH_B$ and the computational basis of this space is :
$(\ket{00}=\ket{0}_A\otimes\ket{0}_B,\ket{01}=\ket{0}_A\otimes\ket{1}_B,\ket{10}=\ket{1}_A\otimes\ket{0}_B,\ket{11}=\ket{1}_A\otimes\ket{1}_B)$.
If $U$ is any unitary operator acting on one qubit, a controlled-$U$ operator acts on the Hilbert space $\HH_{A}\otimes\HH_{B}$ as follows.
One of the two qubits (say qubit $A$) is the control qubit whereas the other qubit is the target qubit. If the control qubit $A$ is in the state $\ket 1$ then $U$ is applied on the target qubit $B$ and if qubit $A$ is in the state $\ket{0}$ nothing is done on qubit $B$. If $U$ is the \PauliX operator then $\cnot$ is nothing more than the controlled-$\XX$ operator
(also denoted by $\cX$) with control on qubit $A$ and target on qubit $B$. So the action of $\cnot$ on the two-qubit register is described by :
$\cnot\ket{00}=\ket{00}, \cnot\ket{01}=\ket{01}, \cnot\ket{10}=\ket{11}, \cnot\ket{11}=\ket{10}$ (the corresponding matrix is given in Figure \ref{univers}).
Notice that the action of the $\cnot$ gate on the system can be sum up by the following simple formula :
\begin{equation}\label{cnotxy}
  \forall x,y\in\{0,1\}, \ \cnot\ket{xy}=\ket{x,x \oplus y}
\end{equation}
where $\oplus$ denotes the XOR operator between two bits which is also the addition in $\F_2$. To emphasize that qubit $A$ is the control and $B$ is the target, $\cnot$ is also denoted $\cnot_{AB}$. So interchanging the roles played by $A$ and $B$ yields another operator denoted by $\cnot_{BA}$.
\begin{equation}
\forall x,y\in\{0,1\}, \ \cnot_{BA}\ket{xy}=\ket{x\oplus y,y}\label{cnotBA}
\end{equation}

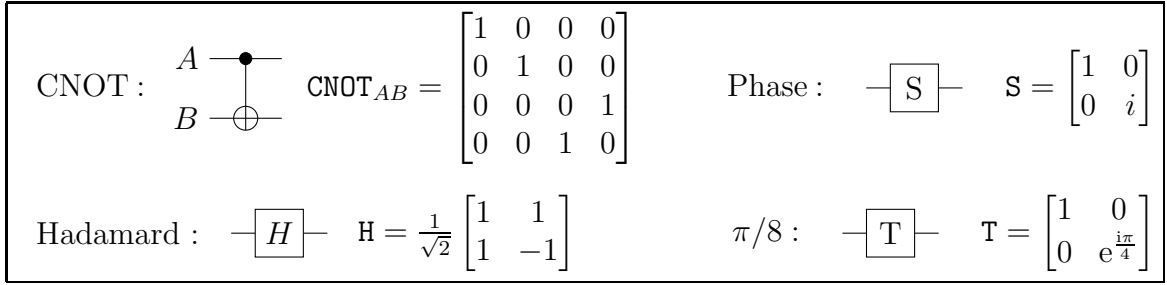
\begin{figure}[h]
\ \ \ CNOT\ :\ \raisebox{-7mm}{
 \begin{tikzpicture}[scale=1.500000,x=1pt,y=1pt]
\filldraw[color=white] (0.000000, -7.500000) rectangle (18.000000, 22.500000);
\draw[color=black] (0.000000,15.000000) -- (18.000000,15.000000);
\draw[color=black] (0.000000,15.000000) node[left] {$A$};
\draw[color=black] (0.000000,0.000000) -- (18.000000,0.000000);
\draw[color=black] (0.000000,0.000000) node[left] {$B$};
\draw (9.000000,15.000000) -- (9.000000,0.000000);
\begin{scope}
\draw[fill=white] (9.000000, 0.000000) circle(3.000000pt);
\clip (9.000000, 0.000000) circle(3.000000pt);
\draw (6.000000, 0.000000) -- (12.000000, 0.000000);
\draw (9.000000, -3.000000) -- (9.000000, 3.000000);
\end{scope}
\filldraw (9.000000, 15.000000) circle(1.500000pt);
\end{tikzpicture}
} 
\ $\cnot_{AB}=\begin{bmatrix}
  1&0&0&0\\
  0&1&0&0\\
  0&0&0&1\\
  0&0&1&0  
\end{bmatrix}$
\hspace{10mm} Phase :\quad \raisebox{-3mm}{
\begin{tikzpicture}[scale=1.500000,x=1pt,y=1pt]
\filldraw[color=white] (0.000000, -7.500000) rectangle (24.000000, 7.500000);
\draw[color=black] (0.000000,0.000000) -- (24.000000,0.000000);
\begin{scope}
\draw[fill=white] (12.000000, -0.000000) +(-45.000000:8.485281pt and 8.485281pt) -- +(45.000000:8.485281pt and 8.485281pt) -- +(135.000000:8.485281pt and 8.485281pt) -- +(225.000000:8.485281pt and 8.485281pt) -- cycle;
\clip (12.000000, -0.000000) +(-45.000000:8.485281pt and 8.485281pt) -- +(45.000000:8.485281pt and 8.485281pt) -- +(135.000000:8.485281pt and 8.485281pt) -- +(225.000000:8.485281pt and 8.485281pt) -- cycle;
\draw (12.000000, -0.000000) node {S};
\end{scope}
\end{tikzpicture}
}\quad 
$\mathtt{S}=\begin{bmatrix}1&0\\0&i\end{bmatrix}$
\vspace{3mm}

\raisebox{-3mm}{
\begin{tikzpicture}[scale=1.500000,x=1pt,y=1pt]
\filldraw[color=white] (0.000000, -7.500000) rectangle (24.000000, 7.500000);
\draw[color=black] (0.000000,0.000000) -- (24.000000,0.000000);
\draw[color=black] (0.000000,0.000000) node[left] {Hadamard\ :\ \ \ };
\begin{scope}
\draw[fill=white] (12.000000, -0.000000) +(-45.000000:8.485281pt and 8.485281pt) -- +(45.000000:8.485281pt and 8.485281pt) -- +(135.000000:8.485281pt and 8.485281pt) -- +(225.000000:8.485281pt and 8.485281pt) -- cycle;
\clip (12.000000, -0.000000) +(-45.000000:8.485281pt and 8.485281pt) -- +(45.000000:8.485281pt and 8.485281pt) -- +(135.000000:8.485281pt and 8.485281pt) -- +(225.000000:8.485281pt and 8.485281pt) -- cycle;
\draw (12.000000, -0.000000) node {$H$};
\end{scope}
\end{tikzpicture}
}\ \ 
$\mathtt{H}=\frac{1}{\sqrt2}\begin{bmatrix}1&1\\1&-1\end{bmatrix}$
\hspace{18mm}
$\pi/8 : \quad$\raisebox{-3mm}{
\begin{tikzpicture}[scale=1.500000,x=1pt,y=1pt]
\filldraw[color=white] (0.000000, -7.500000) rectangle (24.000000, 7.500000);
\draw[color=black] (0.000000,0.000000) -- (24.000000,0.000000);
\begin{scope}
\draw[fill=white] (12.000000, -0.000000) +(-45.000000:8.485281pt and 8.485281pt) -- +(45.000000:8.485281pt and 8.485281pt) -- +(135.000000:8.485281pt and 8.485281pt) -- +(225.000000:8.485281pt and 8.485281pt) -- cycle;
\clip (12.000000, -0.000000) +(-45.000000:8.485281pt and 8.485281pt) -- +(45.000000:8.485281pt and 8.485281pt) -- +(135.000000:8.485281pt and 8.485281pt) -- +(225.000000:8.485281pt and 8.485281pt) -- cycle;
\draw (12.000000, -0.000000) node {T};
\end{scope}
\end{tikzpicture}
}\quad
$\mathtt{T}=\begin{bmatrix}1&0\\0&\ee^{\frac{\ii\pi}{4}}\end{bmatrix}$

{ \caption{ The standard set of universal gates : names, circuit symbols and matrices. \label{univers}}}
\end{figure}

Another useful and common two-qubit quantum gate is the $\swap$ operator whose action on a basis state vector $\ket{xy}$ is given by : $\swap\ket{xy}=\ket{yx}$.
The interesting point is that a $\swap$ gate can be simulated using 3 $\cnot$ gates :
\begin{equation}\label{swap3cnot}
\swap=\cnot_{AB}\cnot_{BA}\cnot_{AB}=\cnot_{BA}\cnot_{AB}\cnot_{BA}
\end{equation}
This identity can be proved easily using equations \eqref{cnotxy} and \eqref{cnotBA}.
\medskip

The length of a quantum circuit is the number of gates composing this circuit and 
we say that two circuits are \emph{equivalent} if their action on the basis state vectors is the same \emph{in theory}, \textit{i.e.} the two circuits represents 
the same unitary operator.  So Identity \eqref{swap3cnot} is a proof of 
equivalence \eqref{swap} in Figure \ref{equi}. This equivalence can be used to implement a $\swap$ gate from  3 $\cnot$ gates on an actual quantum machine. 
Two equivalent circuits generaly do not have the same length. We use the word \emph{equivalent} instead of the word \emph{equal} to emphasize the fact that, due to the errors rate on gates in all current implementations (see Figure \ref{errors}), two equivalent circuits acting on two qubits registers in the same input state will probably not produce \emph{in practice} the same output state. This experimental fact is an important technical problem and a good motivation to work on quantum circuits optimization.

Equivalence \eqref{HHXHH} in Figure \ref{equi} is usefull in practice because the $\cnot_{BA}$ gate may not be native in some implementations, so we need  to simulate it : this can be done using the native gate $\cnot_{AB}$ plus 4 Hadamard gates. Finally another classical equivalence of quantum circuits involves the $\cnot$ gate with the controlled \PauliZ gate (also denoted by $\cZ$) and the Hadamard gate (equivalence \eqref{cZ} in Figure \ref{equi}). We already mentionned it in the introduction  to emphasize the ubiquity of the $\cnot$ gates in quantum circuits as well as the importance of the universality theorem.  Using the definitions of the \PauliZ gate (Figure \ref{single}) and of a controlled-U gate, it is easy to check that the action of the $\cZ$ gate on a basis state vector is defined by
\begin{equation}\label{cZdef}
\forall x,y\in\{0,1\},\   \cZ \ket{xy}=(-1)^{xy}\ket{xy}
\end{equation}
and that it is invariant by switching the control and the target qubits. Equivalence \eqref{cZ} can be proved using the definition of the Hadamard gate (Figure \ref{single}) and Relation \eqref{cZdef}. 

\begin{figure}
  \begin{equation}\label{swap}
    \swap : \ \raisebox{-5mm}{
      \begin{tikzpicture}[scale=1.200000,x=1pt,y=1pt]
\filldraw[color=white] (0.000000, -7.500000) rectangle (180.000000, 22.500000);
\draw[color=black] (0.000000,15.000000) -- (180.000000,15.000000);
\draw[color=black] (0.000000,15.000000) node[left] {$A$};
\draw[color=black] (0.000000,0.000000) -- (180.000000,0.000000);
\draw[color=black] (0.000000,0.000000) node[left] {$B$};
\draw (9.000000,15.000000) -- (9.000000,0.000000);
\begin{scope}
\draw (6.878680, 12.878680) -- (11.121320, 17.121320);
\draw (6.878680, 17.121320) -- (11.121320, 12.878680);
\end{scope}
\begin{scope}
\draw (6.878680, -2.121320) -- (11.121320, 2.121320);
\draw (6.878680, 2.121320) -- (11.121320, -2.121320);
\end{scope}
\draw[fill=white,color=white] (24.000000, -6.000000) rectangle (39.000000, 21.000000);
\draw (31.500000, 7.500000) node {$\sim$};
\draw (54.000000,15.000000) -- (54.000000,0.000000);
\begin{scope}
\draw[fill=white] (54.000000, 0.000000) circle(3.000000pt);
\clip (54.000000, 0.000000) circle(3.000000pt);
\draw (51.000000, 0.000000) -- (57.000000, 0.000000);
\draw (54.000000, -3.000000) -- (54.000000, 3.000000);
\end{scope}
\filldraw (54.000000, 15.000000) circle(1.500000pt);
\draw (72.000000,15.000000) -- (72.000000,0.000000);
\begin{scope}
\draw[fill=white] (72.000000, 15.000000) circle(3.000000pt);
\clip (72.000000, 15.000000) circle(3.000000pt);
\draw (69.000000, 15.000000) -- (75.000000, 15.000000);
\draw (72.000000, 12.000000) -- (72.000000, 18.000000);
\end{scope}
\filldraw (72.000000, 0.000000) circle(1.500000pt);
\draw (90.000000,15.000000) -- (90.000000,0.000000);
\begin{scope}
\draw[fill=white] (90.000000, 0.000000) circle(3.000000pt);
\clip (90.000000, 0.000000) circle(3.000000pt);
\draw (87.000000, 0.000000) -- (93.000000, 0.000000);
\draw (90.000000, -3.000000) -- (90.000000, 3.000000);
\end{scope}
\filldraw (90.000000, 15.000000) circle(1.500000pt);
\draw[fill=white,color=white] (105.000000, -6.000000) rectangle (120.000000, 21.000000);
\draw (112.500000, 7.500000) node {$\sim$};
\draw (135.000000,15.000000) -- (135.000000,0.000000);
\begin{scope}
\draw[fill=white] (135.000000, 15.000000) circle(3.000000pt);
\clip (135.000000, 15.000000) circle(3.000000pt);
\draw (132.000000, 15.000000) -- (138.000000, 15.000000);
\draw (135.000000, 12.000000) -- (135.000000, 18.000000);
\end{scope}
\filldraw (135.000000, 0.000000) circle(1.500000pt);
\draw (153.000000,15.000000) -- (153.000000,0.000000);
\begin{scope}
\draw[fill=white] (153.000000, 0.000000) circle(3.000000pt);
\clip (153.000000, 0.000000) circle(3.000000pt);
\draw (150.000000, 0.000000) -- (156.000000, 0.000000);
\draw (153.000000, -3.000000) -- (153.000000, 3.000000);
\end{scope}
\filldraw (153.000000, 15.000000) circle(1.500000pt);
\draw (171.000000,15.000000) -- (171.000000,0.000000);
\begin{scope}
\draw[fill=white] (171.000000, 15.000000) circle(3.000000pt);
\clip (171.000000, 15.000000) circle(3.000000pt);
\draw (168.000000, 15.000000) -- (174.000000, 15.000000);
\draw (171.000000, 12.000000) -- (171.000000, 18.000000);
\end{scope}
\filldraw (171.000000, 0.000000) circle(1.500000pt);
\end{tikzpicture}
}
\end{equation}
\begin{equation}\label{HHXHH}
    \cnot_{BA} : \ \raisebox{-5mm}{
\begin{tikzpicture}[scale=1.200000,x=1pt,y=1pt]
\filldraw[color=white] (0.000000, -7.500000) rectangle (111.000000, 22.500000);
\draw[color=black] (0.000000,15.000000) -- (111.000000,15.000000);
\draw[color=black] (0.000000,15.000000) node[left] {$A$};
\draw[color=black] (0.000000,0.000000) -- (111.000000,0.000000);
\draw[color=black] (0.000000,0.000000) node[left] {$B$};
\draw (9.000000,15.000000) -- (9.000000,0.000000);
\begin{scope}
\draw[fill=white] (9.000000, 15.000000) circle(3.000000pt);
\clip (9.000000, 15.000000) circle(3.000000pt);
\draw (6.000000, 15.000000) -- (12.000000, 15.000000);
\draw (9.000000, 12.000000) -- (9.000000, 18.000000);
\end{scope}
\filldraw (9.000000, 0.000000) circle(1.500000pt);
\draw[fill=white,color=white] (24.000000, -6.000000) rectangle (39.000000, 21.000000);
\draw (31.500000, 7.500000) node {$\sim$};
\begin{scope}
\draw[fill=white] (57.000000, 15.000000) +(-45.000000:8.485281pt and 8.485281pt) -- +(45.000000:8.485281pt and 8.485281pt) -- +(135.000000:8.485281pt and 8.485281pt) -- +(225.000000:8.485281pt and 8.485281pt) -- cycle;
\clip (57.000000, 15.000000) +(-45.000000:8.485281pt and 8.485281pt) -- +(45.000000:8.485281pt and 8.485281pt) -- +(135.000000:8.485281pt and 8.485281pt) -- +(225.000000:8.485281pt and 8.485281pt) -- cycle;
\draw (57.000000, 15.000000) node {$H$};
\end{scope}
\begin{scope}
\draw[fill=white] (57.000000, -0.000000) +(-45.000000:8.485281pt and 8.485281pt) -- +(45.000000:8.485281pt and 8.485281pt) -- +(135.000000:8.485281pt and 8.485281pt) -- +(225.000000:8.485281pt and 8.485281pt) -- cycle;
\clip (57.000000, -0.000000) +(-45.000000:8.485281pt and 8.485281pt) -- +(45.000000:8.485281pt and 8.485281pt) -- +(135.000000:8.485281pt and 8.485281pt) -- +(225.000000:8.485281pt and 8.485281pt) -- cycle;
\draw (57.000000, -0.000000) node {$H$};
\end{scope}
\draw (78.000000,15.000000) -- (78.000000,0.000000);
\begin{scope}
\draw[fill=white] (78.000000, 0.000000) circle(3.000000pt);
\clip (78.000000, 0.000000) circle(3.000000pt);
\draw (75.000000, 0.000000) -- (81.000000, 0.000000);
\draw (78.000000, -3.000000) -- (78.000000, 3.000000);
\end{scope}
\filldraw (78.000000, 15.000000) circle(1.500000pt);
\begin{scope}
\draw[fill=white] (99.000000, 15.000000) +(-45.000000:8.485281pt and 8.485281pt) -- +(45.000000:8.485281pt and 8.485281pt) -- +(135.000000:8.485281pt and 8.485281pt) -- +(225.000000:8.485281pt and 8.485281pt) -- cycle;
\clip (99.000000, 15.000000) +(-45.000000:8.485281pt and 8.485281pt) -- +(45.000000:8.485281pt and 8.485281pt) -- +(135.000000:8.485281pt and 8.485281pt) -- +(225.000000:8.485281pt and 8.485281pt) -- cycle;
\draw (99.000000, 15.000000) node {$H$};
\end{scope}
\begin{scope}
\draw[fill=white] (99.000000, -0.000000) +(-45.000000:8.485281pt and 8.485281pt) -- +(45.000000:8.485281pt and 8.485281pt) -- +(135.000000:8.485281pt and 8.485281pt) -- +(225.000000:8.485281pt and 8.485281pt) -- cycle;
\clip (99.000000, -0.000000) +(-45.000000:8.485281pt and 8.485281pt) -- +(45.000000:8.485281pt and 8.485281pt) -- +(135.000000:8.485281pt and 8.485281pt) -- +(225.000000:8.485281pt and 8.485281pt) -- cycle;
\draw (99.000000, -0.000000) node {$H$};
\end{scope}
\end{tikzpicture}
}
  \end{equation}
\begin{equation}\label{cZ}
  \cZ :\ \raisebox{-5mm}{
\begin{tikzpicture}[scale=1.200000,x=1pt,y=1pt]
\filldraw[color=white] (0.000000, -7.500000) rectangle (204.000000, 22.500000);
\draw[color=black] (0.000000,15.000000) -- (204.000000,15.000000);
\draw[color=black] (0.000000,15.000000) node[left] {$A$};
\draw[color=black] (0.000000,0.000000) -- (204.000000,0.000000);
\draw[color=black] (0.000000,0.000000) node[left] {$B$};
\draw (9.000000,15.000000) -- (9.000000,0.000000);
\filldraw (9.000000, 15.000000) circle(1.500000pt);
\filldraw (9.000000, 0.000000) circle(1.500000pt);
\draw[fill=white,color=white] (24.000000, -6.000000) rectangle (39.000000, 21.000000);
\draw (31.500000, 7.500000) node {$\sim$};
\begin{scope}
\draw[fill=white] (57.000000, 15.000000) +(-45.000000:8.485281pt and 8.485281pt) -- +(45.000000:8.485281pt and 8.485281pt) -- +(135.000000:8.485281pt and 8.485281pt) -- +(225.000000:8.485281pt and 8.485281pt) -- cycle;
\clip (57.000000, 15.000000) +(-45.000000:8.485281pt and 8.485281pt) -- +(45.000000:8.485281pt and 8.485281pt) -- +(135.000000:8.485281pt and 8.485281pt) -- +(225.000000:8.485281pt and 8.485281pt) -- cycle;
\draw (57.000000, 15.000000) node {$H$};
\end{scope}
\draw (78.000000,15.000000) -- (78.000000,0.000000);
\begin{scope}
\draw[fill=white] (78.000000, 15.000000) circle(3.000000pt);
\clip (78.000000, 15.000000) circle(3.000000pt);
\draw (75.000000, 15.000000) -- (81.000000, 15.000000);
\draw (78.000000, 12.000000) -- (78.000000, 18.000000);
\end{scope}
\filldraw (78.000000, 0.000000) circle(1.500000pt);
\begin{scope}
\draw[fill=white] (99.000000, 15.000000) +(-45.000000:8.485281pt and 8.485281pt) -- +(45.000000:8.485281pt and 8.485281pt) -- +(135.000000:8.485281pt and 8.485281pt) -- +(225.000000:8.485281pt and 8.485281pt) -- cycle;
\clip (99.000000, 15.000000) +(-45.000000:8.485281pt and 8.485281pt) -- +(45.000000:8.485281pt and 8.485281pt) -- +(135.000000:8.485281pt and 8.485281pt) -- +(225.000000:8.485281pt and 8.485281pt) -- cycle;
\draw (99.000000, 15.000000) node {$H$};
\end{scope}
\draw[fill=white,color=white] (117.000000, -6.000000) rectangle (132.000000, 21.000000);
\draw (124.500000, 7.500000) node {$\sim$};
\begin{scope}
\draw[fill=white] (150.000000, -0.000000) +(-45.000000:8.485281pt and 8.485281pt) -- +(45.000000:8.485281pt and 8.485281pt) -- +(135.000000:8.485281pt and 8.485281pt) -- +(225.000000:8.485281pt and 8.485281pt) -- cycle;
\clip (150.000000, -0.000000) +(-45.000000:8.485281pt and 8.485281pt) -- +(45.000000:8.485281pt and 8.485281pt) -- +(135.000000:8.485281pt and 8.485281pt) -- +(225.000000:8.485281pt and 8.485281pt) -- cycle;
\draw (150.000000, -0.000000) node {$H$};
\end{scope}
\draw (171.000000,15.000000) -- (171.000000,0.000000);
\begin{scope}
\draw[fill=white] (171.000000, 0.000000) circle(3.000000pt);
\clip (171.000000, 0.000000) circle(3.000000pt);
\draw (168.000000, 0.000000) -- (174.000000, 0.000000);
\draw (171.000000, -3.000000) -- (171.000000, 3.000000);
\end{scope}
\filldraw (171.000000, 15.000000) circle(1.500000pt);
\begin{scope}
\draw[fill=white] (192.000000, -0.000000) +(-45.000000:8.485281pt and 8.485281pt) -- +(45.000000:8.485281pt and 8.485281pt) -- +(135.000000:8.485281pt and 8.485281pt) -- +(225.000000:8.485281pt and 8.485281pt) -- cycle;
\clip (192.000000, -0.000000) +(-45.000000:8.485281pt and 8.485281pt) -- +(45.000000:8.485281pt and 8.485281pt) -- +(135.000000:8.485281pt and 8.485281pt) -- +(225.000000:8.485281pt and 8.485281pt) -- cycle;
\draw (192.000000, -0.000000) node {$H$};
\end{scope}
\end{tikzpicture}
    }
\end{equation}

{ \caption{ Some classical equivalences of circuits involving $\cnot$ gates.\label{equi}}}
\end{figure}
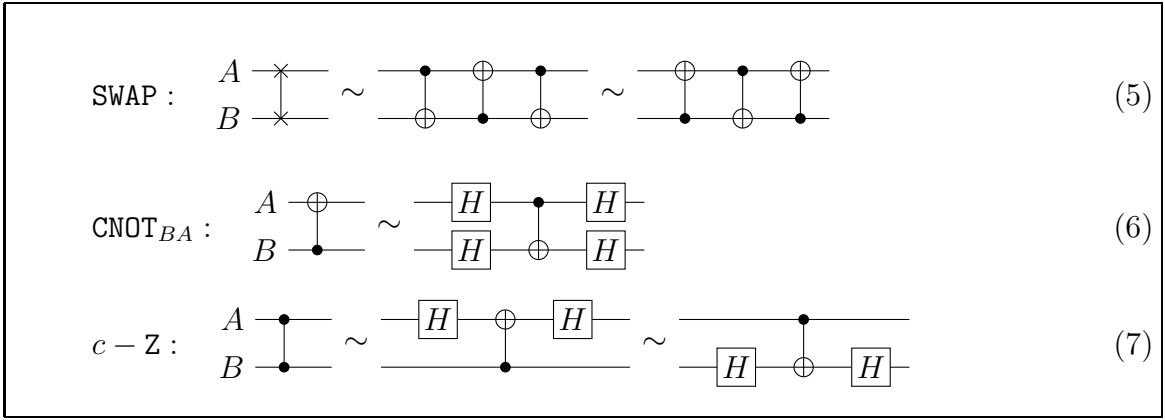
\medskip

On a system of $n$ qubits, we label each qubit from 0 to $n-1$ thus following the usual convention. For coherence we also number
the lines and columns of a matrix of dimension $n$ from 0 to $n-1$ and we consider that a permutation of the symmetric group $\SYM$ is a bijection of $\{0,\dots,n-1\}$. The $n$-qubit system evolves over time in the Hilbert space $\HH_0\otimes \HH_1\otimes\cdots\otimes \HH_{n-1}$ where $\HH_i$ is the Hilbert space of
qubit $i$. Hence the Hilbert space of an $n$-qubit system is isomorphic to $\C^{2^n}$. In this space, a state vector of the standard computational basis is denoted by $\ket{b_0b_1\cdots b_{n-1}}$ with $b_i\in\{0,1\}$. A $\cnot$ gate with target on qubit $i$ and control on qubit $j$ will be denoted $X_{ij}$. The reader will pay attention to the fact that our convention is the opposite of the one generaly used in the literature (and in the beginning of this section) where $\cnot_{ij}$ denotes a $\cnot$ gate with control on qubit $i$ and target on qubit $j$. The reason for this convention is explained  in the proof of Theorem \ref{iso}, Section \ref{algebra}. So, if $i<j$, the action of $X_{ij}$ and $X_{ji}$ on a basis state vector is given by :
\begin{align}
  X_{ij}\ket{b_0\cdots b_i\cdots b_j\cdots b_{n-1}}=\ket{b_0\cdots b_i\oplus b_j\cdots b_j\cdots b_{n-1}},\label{cnotij}\\
  X_{ji}\ket{b_0\cdots b_i\cdots b_j\cdots b_{n-1}}=\ket{b_0\cdots b_i\cdots b_j\oplus b_i\cdots b_{n-1}}.\label{cnotji}
\end{align}
Notice that the $X_{i,j}$ gates are (represented by) matrices of size $2^n\times 2^n$ in the orthogonal group $\mathcal{O}_{2^n}(\R)$  and that they are permutation matrices.

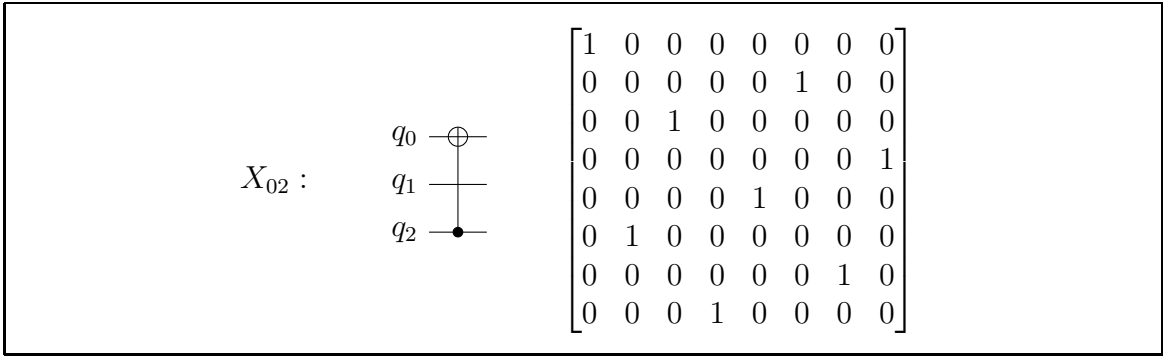
\begin{figure}[h]
  \vspace{2mm}
  
  \begin{center}
    \raisebox{9mm}{$X_{02}: \qquad$}
    \begin{tikzpicture}[scale=1.200000,x=1pt,y=1pt]
      \filldraw[color=white] (0.000000, -7.500000) rectangle (18.000000, 37.500000);
      \draw[color=black] (0.000000,30.000000) -- (18.000000,30.000000);
      \draw[color=black] (0.000000,30.000000) node[left] {$q_0$};
      \draw[color=black] (0.000000,15.000000) -- (18.000000,15.000000);
      \draw[color=black] (0.000000,15.000000) node[left] {$q_1$};
      \draw[color=black] (0.000000,0.000000) -- (18.000000,0.000000);
      \draw[color=black] (0.000000,0.000000) node[left] {$q_2$};
      \draw (9.000000,30.000000) -- (9.000000,0.000000);
      \begin{scope}
        \draw[fill=white] (9.000000, 30.000000) circle(3.000000pt);
        \clip (9.000000, 30.000000) circle(3.000000pt);
        \draw (6.000000, 30.000000) -- (12.000000, 30.000000);
        \draw (9.000000, 27.000000) -- (9.000000, 33.000000);
      \end{scope}
      \filldraw (9.000000, 0.000000) circle(1.500000pt);
    \end{tikzpicture}
    \raisebox{9mm}{$\qquad
      \begin{bmatrix}
        1&0&0&0&0&0&0&0\\
        0&0&0&0&0&1&0&0\\
        0&0&1&0&0&0&0&0\\
        0&0&0&0&0&0&0&1\\
        0&0&0&0&1&0&0&0\\
        0&1&0&0&0&0&0&0\\
        0&0&0&0&0&0&1&0\\
        0&0&0&1&0&0&0&0\\
      \end{bmatrix}$}
    { \caption{ The gate $X_{02}$ in a 3-qubit circuit and its matrix in the standard basis.\label{X02}}}
  \end{center}\vspace{-3mm}
  
\end{figure}

To represent a permutation of $\SYM$, we use the  2-line notation as well as the cycle notation and the $n\times n$ permutation matrix whose entry $(i,j)$ is 1
if $i=\sigma(j)$ and 0 otherwise. 
As an example, if $\sigma\in\SYM[6]$ is the permutation defined by $\sigma(0)=4,\sigma(1)=5,\sigma(2)=0,\sigma(3)=3,\sigma(4)=2,
\sigma(5)=1$ then the 2-line notation is $\sigma=\begin{pmatrix}0&1&2&3&4&5\\4&5&0&3&2&1\end{pmatrix}$, a cycle notation can be
$\sigma=(204)(15)=(51)(420)$ or $\sigma=(204)(15)(3)$ if we want to
write the one-cycle and the permutation matrix is
$M_{\sigma}=
\begin{bmatrix}
  0&0&1&0&0&0\\
  0&0&0&0&0&1\\
  0&0&0&0&1&0\\
  0&0&0&1&0&0\\
  1&0&0&0&0&0\\
  0&1&0&0&0&0
\end{bmatrix}$.
The \emph{cycle type} of a permutation $\sigma\in\SYM$ is the tuple $\lambda=(n_1,n_2,\dots,n_p)$ of positive integers such that $n_1\geq n_2\geq \dots\geq n_p$,
$n=\sum_{i=1}^pn_i$ and
$\sigma$ is the commutative product of cycles of length $n_i$ (including the cycles of length 1). This kind of tuple is called a decreasing \emph{partition} of $n$.
For instance $\sigma=(204)(15)(3)=(204)(15)$ has cycle type $\lambda=(3,2,1)$. Notice that the cycle type of the identity is $\lambda=(1,\dots,1)$.
Two permutations have the same cycle type iff they are in the same conjugacy class. For instance $\sigma'=(350)(24)$ has the same cycle type as $\sigma$ and $\sigma'=\gamma\sigma\gamma^{-1}$ where $\gamma=\begin{pmatrix}0&1&2&3&4&5\\5&2&3&1&0&4\end{pmatrix}$.

A $\swap$ gate between qubit $i$ and qubit $j$ is denoted by $S_{(ij)}$ where $(ij)$ is the cyclic notation for the transposition that exchange $i$ and $j$.
Since the symmetric group $\SYM$ is generated by the transpositions, it is straightforward to prove that the group generated by the $S_{(ij)}$ is isomorphic to the symmetric group  : each $\swap$ gate $S_{(ij)}$ in a $n$-qubit circuit corresponds to the transposition $(ij)$ in $\SYM$. Notice that the cyclic notation of a transposition (\textit{i.e.} $(ij)=(ji)$) implies $S_{(ij)}=S_{(ji)}$ which is coherent with the symmetry of a $\swap$ gate. More generaly any circuit  of $\swap$ gates maps to a permutation $\sigma$ (examples in Figure \ref{swapcircuit} and Figure \ref{cyclic}). If we denote by $S_{\sigma}$ the unitary operator corresponding to that circuit one has :
\begin{equation}\label{swapsigma}
  S_{\sigma}\ket{b_0b_1\cdots b_{n-1}}=\ket{b_{\sigma^{-1}(0)}b_{\sigma^{-1}(1)}\cdots b_{\sigma^{-1}(n-1)}}
\end{equation}

\begin{figure}[h]
  \begin{center}
  $\sigma=(024)(13)\qquad S_{\sigma}=S_{(02)}S_{(24)}S_{(13)}\qquad$\raisebox{-15mm}{
\begin{tikzpicture}[scale=1.200000,x=1pt,y=1pt]
\filldraw[color=white] (0.000000, -7.500000) rectangle (42.000000, 67.500000);
\draw[color=black] (0.000000,60.000000) -- (42.000000,60.000000);
\draw[color=black] (0.000000,60.000000) node[left] {$q_0$};
\draw[color=black] (0.000000,45.000000) -- (42.000000,45.000000);
\draw[color=black] (0.000000,45.000000) node[left] {$q_1$};
\draw[color=black] (0.000000,30.000000) -- (42.000000,30.000000);
\draw[color=black] (0.000000,30.000000) node[left] {$q_2$};
\draw[color=black] (0.000000,15.000000) -- (42.000000,15.000000);
\draw[color=black] (0.000000,15.000000) node[left] {$q_3$};
\draw[color=black] (0.000000,0.000000) -- (42.000000,0.000000);
\draw[color=black] (0.000000,0.000000) node[left] {$q_4$};
\draw (9.000000,45.000000) -- (9.000000,15.000000);
\begin{scope}
\draw (6.878680, 42.878680) -- (11.121320, 47.121320);
\draw (6.878680, 47.121320) -- (11.121320, 42.878680);
\end{scope}
\begin{scope}
\draw (6.878680, 12.878680) -- (11.121320, 17.121320);
\draw (6.878680, 17.121320) -- (11.121320, 12.878680);
\end{scope}
\draw (15.000000,30.000000) -- (15.000000,0.000000);
\begin{scope}
\draw (12.878680, 27.878680) -- (17.121320, 32.121320);
\draw (12.878680, 32.121320) -- (17.121320, 27.878680);
\end{scope}
\begin{scope}
\draw (12.878680, -2.121320) -- (17.121320, 2.121320);
\draw (12.878680, 2.121320) -- (17.121320, -2.121320);
\end{scope}
\draw (33.000000,60.000000) -- (33.000000,30.000000);
\begin{scope}
\draw (30.878680, 57.878680) -- (35.121320, 62.121320);
\draw (30.878680, 62.121320) -- (35.121320, 57.878680);
\end{scope}
\begin{scope}
\draw (30.878680, 27.878680) -- (35.121320, 32.121320);
\draw (30.878680, 32.121320) -- (35.121320, 27.878680);
\end{scope}
\end{tikzpicture}
}
\end{center} 
{ \caption{ A 5-qubit quantum circuit of $\swap$ gates and the corresponding permutation\label{swapcircuit}}}
  \end{figure}
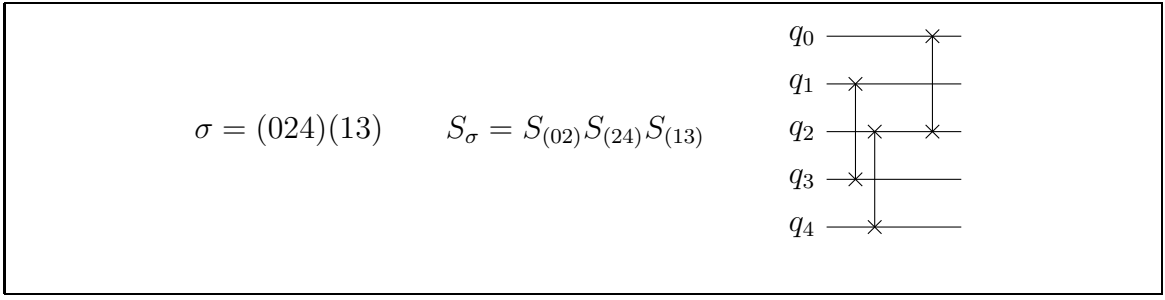

\section{The group generated by the $\cnot$ gates}\label{algebra}

We denote by $\XG$ ($n\geq  2$) the group generated by the $n(n-1)$ gates (\textit{i.e.} matrices) $X_{ij}$.
\vspace{-5mm}

\begin{equation}
\XG :=\left<X_{ij}\mid 0\leq i,j < n, i\neq j\right>
\end{equation}

\begin{figure}[h]
  \begin{center}
    \raisebox{12mm}{$X_{12}X_{10}X_ {31}X_{03}:\quad\quad$}
    \begin{tikzpicture}[scale=1.200000,x=1pt,y=1pt]
\filldraw[color=white] (0.000000, -7.500000) rectangle (72.000000, 52.500000);
\draw[color=black] (0.000000,45.000000) -- (72.000000,45.000000);
\draw[color=black] (0.000000,45.000000) node[left] {$q_0$};
\draw[color=black] (0.000000,30.000000) -- (72.000000,30.000000);
\draw[color=black] (0.000000,30.000000) node[left] {$q_1$};
\draw[color=black] (0.000000,15.000000) -- (72.000000,15.000000);
\draw[color=black] (0.000000,15.000000) node[left] {$q_2$};
\draw[color=black] (0.000000,0.000000) -- (72.000000,0.000000);
\draw[color=black] (0.000000,0.000000) node[left] {$q_3$};
\draw (9.000000,45.000000) -- (9.000000,0.000000);
\begin{scope}
\draw[fill=white] (9.000000, 45.000000) circle(3.000000pt);
\clip (9.000000, 45.000000) circle(3.000000pt);
\draw (6.000000, 45.000000) -- (12.000000, 45.000000);
\draw (9.000000, 42.000000) -- (9.000000, 48.000000);
\end{scope}
\filldraw (9.000000, 0.000000) circle(1.500000pt);
\draw (27.000000,30.000000) -- (27.000000,0.000000);
\begin{scope}
\draw[fill=white] (27.000000, 0.000000) circle(3.000000pt);
\clip (27.000000, 0.000000) circle(3.000000pt);
\draw (24.000000, 0.000000) -- (30.000000, 0.000000);
\draw (27.000000, -3.000000) -- (27.000000, 3.000000);
\end{scope}
\filldraw (27.000000, 30.000000) circle(1.500000pt);
\draw (45.000000,45.000000) -- (45.000000,30.000000);
\begin{scope}
\draw[fill=white] (45.000000, 30.000000) circle(3.000000pt);
\clip (45.000000, 30.000000) circle(3.000000pt);
\draw (42.000000, 30.000000) -- (48.000000, 30.000000);
\draw (45.000000, 27.000000) -- (45.000000, 33.000000);
\end{scope}
\filldraw (45.000000, 45.000000) circle(1.500000pt);
\draw (63.000000,30.000000) -- (63.000000,15.000000);
\begin{scope}
\draw[fill=white] (63.000000, 30.000000) circle(3.000000pt);
\clip (63.000000, 30.000000) circle(3.000000pt);
\draw (60.000000, 30.000000) -- (66.000000, 30.000000);
\draw (63.000000, 27.000000) -- (63.000000, 33.000000);
\end{scope}
\filldraw (63.000000, 15.000000) circle(1.500000pt);
\end{tikzpicture}
\end{center}
{ \caption{  A circuit of $\XG[4]$ and the corresponding operator.\label{excx4}}}
\end{figure}
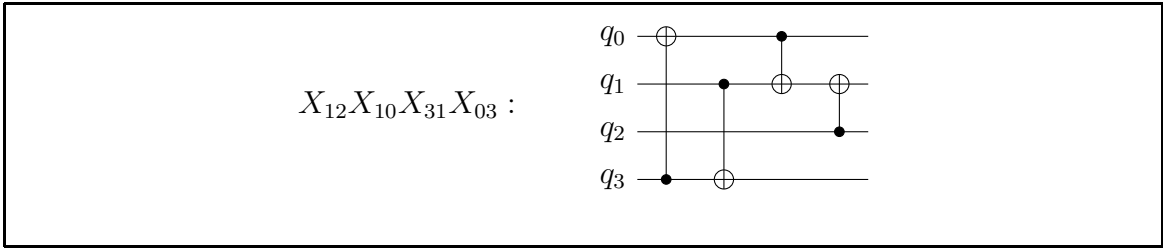

\begin{prop}\label{idx}The following identities are satisfied by the generators of $\XG$. 
  \begin{align}
    &\text{Involution : } X_{ij}^2=I\label{inv}\\
    &\text{Braid relation : }   X_{ij}X_{ji}X_{ij}=X_{ji}X_{ij}X_{ji}=S_{(ij)}\label{braid}\\
    &\text{Commutation : }  (X_{ij}X_{k\ell})^2=I \quad \text{where }i\neq \ell,j\neq k\label{commute}\\ 
     &\text{Non-commutation : } (X_{ij}X_{jk})^2=(X_{jk}X_{ij})^2=X_{ik}\label{noncom}
  \end{align}
\end{prop}

\begin{proof}
  One checks each identity $A=B$ by showing that the actions of gates $A$ and $B$ on a basis state vector $\ket{b_0\cdots b_{n-1}}$ are the same.
  Using identities  \ref{cnotij} and \ref{cnotji}, the results follow from direct computation let to the readers.
\end{proof}

\begin{figure}[h]
\begin{center}
  \begin{tikzpicture}[scale=1.200000,x=1pt,y=1pt]
\filldraw[color=white] (0.000000, -7.500000) rectangle (216.000000, 52.500000);
\draw[color=black] (0.000000,45.000000) -- (216.000000,45.000000);
\draw[color=black] (0.000000,45.000000) node[left] {$q_0$};
\draw[color=black] (0.000000,30.000000) -- (216.000000,30.000000);
\draw[color=black] (0.000000,30.000000) node[left] {$q_1$};
\draw[color=black] (0.000000,15.000000) -- (216.000000,15.000000);
\draw[color=black] (0.000000,15.000000) node[left] {$q_2$};
\draw[color=black] (0.000000,0.000000) -- (216.000000,0.000000);
\draw[color=black] (0.000000,0.000000) node[left] {$q_3$};
\draw (9.000000,45.000000) -- (9.000000,0.000000);
\begin{scope}
\draw[fill=white] (9.000000, 45.000000) circle(3.000000pt);
\clip (9.000000, 45.000000) circle(3.000000pt);
\draw (6.000000, 45.000000) -- (12.000000, 45.000000);
\draw (9.000000, 42.000000) -- (9.000000, 48.000000);
\end{scope}
\filldraw (9.000000, 0.000000) circle(1.500000pt);
\draw[fill=white,color=white] (24.000000, -6.000000) rectangle (39.000000, 51.000000);
\draw (31.500000, 22.500000) node {$\sim$};
\draw (54.000000,15.000000) -- (54.000000,0.000000);
\begin{scope}
\draw[fill=white] (54.000000, 15.000000) circle(3.000000pt);
\clip (54.000000, 15.000000) circle(3.000000pt);
\draw (51.000000, 15.000000) -- (57.000000, 15.000000);
\draw (54.000000, 12.000000) -- (54.000000, 18.000000);
\end{scope}
\filldraw (54.000000, 0.000000) circle(1.500000pt);
\draw (72.000000,45.000000) -- (72.000000,15.000000);
\begin{scope}
\draw[fill=white] (72.000000, 45.000000) circle(3.000000pt);
\clip (72.000000, 45.000000) circle(3.000000pt);
\draw (69.000000, 45.000000) -- (75.000000, 45.000000);
\draw (72.000000, 42.000000) -- (72.000000, 48.000000);
\end{scope}
\filldraw (72.000000, 15.000000) circle(1.500000pt);
\draw (90.000000,15.000000) -- (90.000000,0.000000);
\begin{scope}
\draw[fill=white] (90.000000, 15.000000) circle(3.000000pt);
\clip (90.000000, 15.000000) circle(3.000000pt);
\draw (87.000000, 15.000000) -- (93.000000, 15.000000);
\draw (90.000000, 12.000000) -- (90.000000, 18.000000);
\end{scope}
\filldraw (90.000000, 0.000000) circle(1.500000pt);
\draw (108.000000,45.000000) -- (108.000000,15.000000);
\begin{scope}
\draw[fill=white] (108.000000, 45.000000) circle(3.000000pt);
\clip (108.000000, 45.000000) circle(3.000000pt);
\draw (105.000000, 45.000000) -- (111.000000, 45.000000);
\draw (108.000000, 42.000000) -- (108.000000, 48.000000);
\end{scope}
\filldraw (108.000000, 15.000000) circle(1.500000pt);
\draw[fill=white,color=white] (123.000000, -6.000000) rectangle (138.000000, 51.000000);
\draw (130.500000, 22.500000) node {$\sim$};
\draw (153.000000,45.000000) -- (153.000000,15.000000);
\begin{scope}
\draw[fill=white] (153.000000, 45.000000) circle(3.000000pt);
\clip (153.000000, 45.000000) circle(3.000000pt);
\draw (150.000000, 45.000000) -- (156.000000, 45.000000);
\draw (153.000000, 42.000000) -- (153.000000, 48.000000);
\end{scope}
\filldraw (153.000000, 15.000000) circle(1.500000pt);
\draw (171.000000,15.000000) -- (171.000000,0.000000);
\begin{scope}
\draw[fill=white] (171.000000, 15.000000) circle(3.000000pt);
\clip (171.000000, 15.000000) circle(3.000000pt);
\draw (168.000000, 15.000000) -- (174.000000, 15.000000);
\draw (171.000000, 12.000000) -- (171.000000, 18.000000);
\end{scope}
\filldraw (171.000000, 0.000000) circle(1.500000pt);
\draw (189.000000,45.000000) -- (189.000000,15.000000);
\begin{scope}
\draw[fill=white] (189.000000, 45.000000) circle(3.000000pt);
\clip (189.000000, 45.000000) circle(3.000000pt);
\draw (186.000000, 45.000000) -- (192.000000, 45.000000);
\draw (189.000000, 42.000000) -- (189.000000, 48.000000);
\end{scope}
\filldraw (189.000000, 15.000000) circle(1.500000pt);
\draw (207.000000,15.000000) -- (207.000000,0.000000);
\begin{scope}
\draw[fill=white] (207.000000, 15.000000) circle(3.000000pt);
\clip (207.000000, 15.000000) circle(3.000000pt);
\draw (204.000000, 15.000000) -- (210.000000, 15.000000);
\draw (207.000000, 12.000000) -- (207.000000, 18.000000);
\end{scope}
\filldraw (207.000000, 0.000000) circle(1.500000pt);
\end{tikzpicture}

 $X_{03}=X_{02}X_{23}X_{02}X_{23}=X_{23}X_{02}X_{23}X_{02}$
{ \caption{ Example of Identity \eqref{noncom} in a $\XG[4]$ circuit.\label{equiX03}}}
\end{center}
\end{figure}

Using the identities from Proposition \ref{idx}, one gets easily the conjugacy relations in $\XG$ :
\vspace{-5mm}

\begin{equation}
  X_{ij}X_{jk}X_{ij}=X_{jk}X_{ik} \text{\quad and \quad  } X_{ij}X_{ki}X_{ij}=X_{ki}X_{kj}. \label{conjx}
\end{equation}

Let us denote by $\SG$ the group generated by the $S_{ij}$,
\begin{equation}\SG := \left<S_{(ij)}\mid  0\leq i,j < n\right>\end{equation}
As explained at the end of Section \ref{background}, $\SG$ is isomorphic to the symmetric group $\SYM$ and Identity \eqref{braid} implies that $\SG$ is a subgroup of $\XG$.
The action of $\SG$ on $\XG$ by conjugation is given by the following proposition :

\begin{prop} For any permutation $\sigma$ in $\SYM$ and any $i\neq j$, one has :
\begin{equation}\label{actionxsym} S_{\sigma}X_{ij}S_{\sigma^{-1}}=X_{\sigma(i)\sigma(j)}.\end{equation}
\end{prop}

\begin{proof}
  Since the transpositions generate the symmetric group, it suffices to 
prove the result when $\sigma$ is a transposition, \textit{i.e.} 
$S_\sigma=X_{k\ell}X_{\ell k}X_{k\ell}$ for some $\ell\neq k$. Hence, 
the result comes straightforwardly from equalities \eqref{inv}, \eqref{braid}, \eqref{commute} and \eqref{conjx}.
\end{proof}

\begin{figure}[h]
  \begin{center}
\begin{tikzpicture}[scale=1.200000,x=1pt,y=1pt]
\filldraw[color=white] (0.000000, -7.500000) rectangle (180.000000, 52.500000);
\draw[color=black] (0.000000,45.000000) -- (180.000000,45.000000);
\draw[color=black] (0.000000,45.000000) node[left] {$q_0$};
\draw[color=black] (0.000000,30.000000) -- (180.000000,30.000000);
\draw[color=black] (0.000000,30.000000) node[left] {$q_1$};
\draw[color=black] (0.000000,15.000000) -- (180.000000,15.000000);
\draw[color=black] (0.000000,15.000000) node[left] {$q_2$};
\draw[color=black] (0.000000,0.000000) -- (180.000000,0.000000);
\draw[color=black] (0.000000,0.000000) node[left] {$q_3$};
\draw (9.000000,45.000000) -- (9.000000,0.000000);
\begin{scope}
\draw[fill=white] (9.000000, 45.000000) circle(3.000000pt);
\clip (9.000000, 45.000000) circle(3.000000pt);
\draw (6.000000, 45.000000) -- (12.000000, 45.000000);
\draw (9.000000, 42.000000) -- (9.000000, 48.000000);
\end{scope}
\filldraw (9.000000, 0.000000) circle(1.500000pt);
\draw[fill=white,color=white] (24.000000, -6.000000) rectangle (39.000000, 51.000000);
\draw (31.500000, 22.500000) node {$\sim$};
\draw (54.000000,45.000000) -- (54.000000,15.000000);
\begin{scope}
\draw (51.878680, 42.878680) -- (56.121320, 47.121320);
\draw (51.878680, 47.121320) -- (56.121320, 42.878680);
\end{scope}
\begin{scope}
\draw (51.878680, 12.878680) -- (56.121320, 17.121320);
\draw (51.878680, 17.121320) -- (56.121320, 12.878680);
\end{scope}
\draw (72.000000,15.000000) -- (72.000000,0.000000);
\begin{scope}
\draw[fill=white] (72.000000, 15.000000) circle(3.000000pt);
\clip (72.000000, 15.000000) circle(3.000000pt);
\draw (69.000000, 15.000000) -- (75.000000, 15.000000);
\draw (72.000000, 12.000000) -- (72.000000, 18.000000);
\end{scope}
\filldraw (72.000000, 0.000000) circle(1.500000pt);
\draw (90.000000,45.000000) -- (90.000000,15.000000);
\begin{scope}
\draw (87.878680, 42.878680) -- (92.121320, 47.121320);
\draw (87.878680, 47.121320) -- (92.121320, 42.878680);
\end{scope}
\begin{scope}
\draw (87.878680, 12.878680) -- (92.121320, 17.121320);
\draw (87.878680, 17.121320) -- (92.121320, 12.878680);
\end{scope}
\draw[fill=white,color=white] (105.000000, -6.000000) rectangle (120.000000, 51.000000);
\draw (112.500000, 22.500000) node {$\sim$};
\draw (135.000000,15.000000) -- (135.000000,0.000000);
\begin{scope}
\draw (132.878680, -2.121320) -- (137.121320, 2.121320);
\draw (132.878680, 2.121320) -- (137.121320, -2.121320);
\end{scope}
\begin{scope}
\draw (132.878680, 12.878680) -- (137.121320, 17.121320);
\draw (132.878680, 17.121320) -- (137.121320, 12.878680);
\end{scope}
\draw (153.000000,45.000000) -- (153.000000,15.000000);
\begin{scope}
\draw[fill=white] (153.000000, 45.000000) circle(3.000000pt);
\clip (153.000000, 45.000000) circle(3.000000pt);
\draw (150.000000, 45.000000) -- (156.000000, 45.000000);
\draw (153.000000, 42.000000) -- (153.000000, 48.000000);
\end{scope}
\filldraw (153.000000, 15.000000) circle(1.500000pt);
\draw (171.000000,15.000000) -- (171.000000,0.000000);
\begin{scope}
\draw (168.878680, 12.878680) -- (173.121320, 17.121320);
\draw (168.878680, 17.121320) -- (173.121320, 12.878680);
\end{scope}
\begin{scope}
\draw (168.878680, -2.121320) -- (173.121320, 2.121320);
\draw (168.878680, 2.121320) -- (173.121320, -2.121320);
\end{scope}
\end{tikzpicture}

$X_{03}=S_{(02)}X_{23}S_{(02)}=S_{(23)}X_{02}S_{(23)}$
{ \caption{ Example of Identity \eqref{actionxsym} in a $\XG[4]$ circuit.\label{x03swap}}}
\end{center}
\end{figure}

We denote by $\mathrm{SL}_n(K)$ the special linear group on a field $K$ and by $T_{ij}(n)$ the matrix $I_n+n E_{ij}$  where $n\in K\setminus\{0\}$, $i\neq j$
and $E_{ij}$ is the matrix with 0 on all entries but the entry $(i,j)$ which is equal to 1.
Let $(e_k)_{0\leq k < n}$ be the canonical basis of the vector space $K^{n}$ and $(e_k^*)_{0\leq k < n}$ its  dual basis. Then the matrix $T_{ij}(n)$
represents in the canonical basis the automorphism $t_{ij}(n): u\rightarrow u + n e_j^*(u)e_i$ of $K^{n}$ which is a \emph{transvection} fixing the hyperplane $\left<e_k\mid k\neq j\right>$ and directed by the line $\left<e_{i}\right>$. So $T_{ij}(n)$ is a \emph{transvection matrix} and this is a well known fact in linear algebra that the transvection matrices generate $\mathrm{SL}_n(K)$. A simple way to find a decomposition in transvection matrices of any matrix $M$ in $\mathrm{SL}_n(K)$ is to use the Gauss-Jordan algorithm with $M$ as input.  This algorithm is generaly used to compute an inverse of a matrix but it also yields a decomposition of $M$ as a product of transvections (see Figure \ref{Gauss} in Section \ref{general} for an example).

If $K=\F_2$ then $n=1$ and the set $\{T_{ij}(n)|n \in K \}$ is reduced to the $n(n-1)$ transvection matrices $T_{ij}:=I_n+E_{ij}$.
Moreover, since any invertible matrix of $\GL$ has determinant 1, one has :
\begin{equation}
  \GL = \left<T_{ij}\mid 0\leq i,j<n, i\neq j \right>. 
\end{equation}
We recall that the Gauss-Jordan algorithm is based on the following observation :

\begin{prop}\label{GJmult}
Multiplying to the left (resp. the right) any $\F_2$-matrix $M$ by a 
transvection matrix $T_{ij}$ is equivalent to add the $j$th line (resp. 
$i$th column) to the $i$th line (resp. $j$ column) in $M$.
\end{prop}

In particular, one has
\begin{equation}\label{tiju}
T_{ij}\begin{bmatrix}b_0\\\vdots\\b_i\\\vdots\\b_j\\\vdots\\b_{n-1}\end{bmatrix}=\begin{bmatrix}b_0\\\vdots\\b_i\oplus b_j\\\vdots\\b_j\\\vdots\\b_{n-1}\end{bmatrix}.
\end{equation}

The notation $\ket{b_0b_1\cdots b_{n-1}}$ used in QIT is a shorthand for the tensor product $\ket{b_0}\otimes\ket{b_1}\otimes\dots\otimes\ket{b_{n-1}}$
and it is convenient to identify the binary  label $ b_0b_1\cdots b_{n-1}$ with the column vector $u=[b_0,b_1,\cdots,b_{n-1}]^{\mathrm{T}}$ of $\F_2^n$ since the $\oplus$ (XOR) operation between two bits corresponds to the addition in $\F_2$.  
So the computational basis of the Hilbert space $\HH_0\otimes \HH_1\otimes\cdots\otimes \HH_{n-1}$ is from now on denoted by $\left(\ket{u}\right)_{u\in\F_2^n}$
and using Relation \eqref{tiju} we can rewrite Relation \eqref{cnotij} in a much cleaner way as
\begin{equation}\label{xtrans} 
  X_{ij}\ket{u}=\ket{T_{ij}u}.
\end{equation}

The above considerations lead quiet naturally to the following theorem :

\begin{theo} \label{iso}The group $\XG$ generated by the $\cnot$ gates acting on $n$ qubits is isomorphic to $\GL$.
  The morphism $\Phi$ sending each $X_{ij}$ to $T_{ij}$ is an explicit isomorphism.
\end{theo}

\begin{proof}
  The surjectivity of $\Phi$ is due to the fact that $\GL$ is generated by the transvections $T_{ij}$.
  Furthermore, due to Relation \eqref{xtrans}, a preimage $N$ by $\Phi$ of a matrix $M$ in $\GL$ must satisfy
  the relation $N\ket{u}=\ket{Mu}$ for any basis vector $\ket{u}$ and there is only one matrix $N$ satisfying this relation.
  So $\Phi$ is also injective and the result is proved.
\end{proof}

Notice that the image by $\Phi$ of a  $\swap$ matrix $S_{(ij)}$ is a transposition matrix $P_{(ij)}:=T_{ij}T_{ji}T_{ij}=T_{ji}T_{ij}T_{ji}$ in
$\GL$ and more generaly $\Phi(S_{\sigma})$ is the permutation matrix $P_{\sigma}$ for any permutation $\sigma$ in $\SYM$.
From now on we will denote by $M^{\sigma}=P_{\sigma}MP_{\sigma}^{-1}$ the conjugate of a matrix $M$ in $\GL$ by a permutation matrix $P_{\sigma}$. Using the isomorphism $\Phi$, Relation \eqref{actionxsym}
leads to 
\begin{equation} T_{ij}^{\sigma}=P_{\sigma}T_{ij}P_{\sigma}^{-1}=T_{\sigma(i) \sigma(j)}.\label{actiontsym}\end{equation}

\medskip

The order $\GL$ is classicaly obtained by computing the number of different basis of $\F_2^{n}$ and Theorem \ref{iso} implies
\begin{cor}
\begin{equation}
  \left|\XG\right|=2^{\frac{n(n-1)}{2}}\prod_{i=1}^n(2^i-1).
\end{equation}
\end{cor}

In his "Lectures on Chevalley Group" \cite[Chapter 6]{Steinberg1967}, Steinberg gives a presentation of the special linear group on a finite field $K$ in dimension $n\geq  3$. The notation $(a,b)$ stands here for the commutator of two elements $a$ and $b$ of a group  (\textit{i.e.} $a^{-1}b^{-1}ab$) which is usualy denoted by $[a,b]$.

\begin{theo}(Steinberg)
  
  If $n\geq  3$ and $K$ is a finite field, the symbols $x_{ij} (t)$, $(1 \leq i, j \leq n,\  i\neq j,\  t\in K)$
  subject to the relations :
  \begin{description}
  \item[(A)] $\quad x_{ij} (t)x_{ij} (u) = x_{ij} (t + u)$
  \item[(B)] $\quad (x_{ij}(t), x_{jk}(u)) = x_{ik} (tu)$ if  $i,j,k$ are distinct, $(x_{ij}(t), x_{k\ell}(u)) = 1$

    if $j \neq  k,\  i \neq  \ell$
  \end{description}
  define the group $\mathrm{SL}_n(K)$.
\end{theo}\medskip    

It is straightforward to adapt this presentation to the case $K=\F_2$ and we get so a presentation for the group $\XG$ for any $n\geq  3$.

\begin{cor}\label{cXpres}
  If $n\geq  3$, a presentation of the  group $\mathrm{cX}_n$ is $<\mathcal{S}\mid \mathcal{R}>$ where $\mathcal{S}$ is the set of the $n(n-1)$ symbols $x_{ij}$  $(0\leq i,j\leq n-1,\ i\neq j)$  and $\mathcal{R}$ is the set of the relations :                        
  \vspace{-5mm}
  
  \begin{align}
    &x_{ij}^2=1, \label{invxij}\\
    &(x_{ij}x_{jk})^2=x_{ik},\\
    &(x_{ij}x_{k\ell})^2=1\quad \text{ if } i\neq \ell,j\neq k.
  \end{align}
    
\end{cor}

We remark that all the identities given by Proposition \ref{idx} appear in this presentation but Identity \eqref{braid}. Of course the braid relation \eqref{braid} can be deduced from
the presentation relations but the calculation is tricky. First of all, using the above relations we obtain  the conjugacy relation 
$x_{ij}x_{jk}x_{ij}=x_{jk}x_{ik}=x_{ik}x_{jk}$ (as we did to obtain Relation \eqref{conjx}). Then we use many times this  conjugacy  relation to obtain the relation
$x_{ij}x_{ji}x_{ij}x_{ji}x_{ij}=x_{ji}$ and finally we get the braid relation $x_{ij}x_{ji}x_{ij}=x_{ji}x_{ij}x_{ji}$ using Relation \eqref{invxij}. The calculation is detailled in Figure \ref{tricky}.
 
\begin{figure}[h]
\begin{minipage}[t]{.4\linewidth}
 \begin{equation*}
  \begin{split}
    x_{ji}x_{ij}x_{ji}& = x_{ji}(x_{ik}x_{kj})^2x_{ji}\\
    & =(x_{ji}x_{ik}x_{kj}x_{ji})^2\\
    & = (x_{ji}x_{ik}x_{ji}\ x_{ji}x_{kj}x_{ji})^2\\
    & = (x_{ik}x_{jk}x_{ki}x_{kj})^2
  \end{split}
\end{equation*}
\end{minipage}
\begin{minipage}[t]{.6\linewidth}
\begin{equation*}
  \begin{split}
    x_{ij}x_{ji}x_{ij}x_{ji}x_{ij}&=x_{ij}(x_{ik}x_{jk}x_{ki}x_{kj})^2x_{ij}\\
    &=(x_{ij}x_{ik}x_{jk}x_{ki}x_{kj}x_{ij})^2\\
    &=(x_{ik}x_{ij}x_{jk}x_{ki}x_{ij}x_{kj})^2\\
    &=(x_{ik}x_{ik}x_{jk}x_{ki}x_{kj}x_{kj})^2\\
    &=(x_{jk}x_{ki})^2\\
    &=x_{ji}
  \end{split}
\end{equation*}
\end{minipage}
{ \caption{ Getting the braid relation $x_{ij}x_{ji}x_{ij}=x_{ji}x_{ij}x_{ji}$ from Corollary \ref{cXpres}.\label{tricky}}}
\end{figure}
\medskip

A classical result in Group Theory is that the projective special linear group $\mathrm{PSL}_n(K)$ is simple when $n\geq  3$ (see \textit{e.g.} \cite[Chapter 3]{Wilson2009} for a proof). This group  is defined as $\mathrm{SL}_n(K)/\mathrm{Z}$ where $\mathrm{Z}$ is the subgroup of all scalar matrices with determinant 1,
which appears to be the center of $\mathrm{SL}_n(K)$ . If $K=\F_2$ the identity matrix is is the only scalar matrix that has determinant 1 and the center of $\SL$ is reduced to the trivial group. Hence $\SL=\mathrm{PSL}_n(\F_2)$ and the group $\XG$ is simple when $n\geq  3$.

If $n=2$, $\XG[2]$ has order 6 and is isomorphic to $\SYM[3]$, one possible isomorphism being
$X_{01}\simeq (01), X_{10}\simeq (12), X_{01}X_{10}\simeq (012), X_{10}X_{01}\simeq (021)$ and $X_{01}X_{10}X_{01}\simeq (02)$. So $\XG[2]$ is not a simple group since  the alternating group is always a normal subgroup of $\SYM$.

\medskip

Identity \eqref{noncom} involves 3 distinct integers $i,j,k$ and can be generalized by induction to an arbitrary number of distinct integers in the following way :

\begin{prop}\label{chaslesprop}
  Let $i_1,i_2,\dots,i_p$ be distinct integers $(p\geq  3)$ of $\{0,1, \dots, n-1\}$, then :
  \begin{equation}\label{chaslesform}
    X_{i_1,i_p}=\left(X_{i_1,i_2}X_{i_2,i_3}\dots X_{i_{p-2},i_{p-1}}X_{i_{p-1},i_{p}} X_{i_{p-2}, i_{p-1}}\dots X_{i_2, i_3}\right)^2
  \end{equation}
\end{prop}

Proposition \ref{chaslesprop} provides a simple way to adapt any $\cnot$ gate to the topological constraints of a given quantum computer. These constraints can be
represented by a directed graph whose vertices are labelled by the qubits. There is an arrow from qubit $i$ to qubit $j$ when these two qubits can interact to perform an  $X_{ij}$  operation on the system. In a device with the complete graph topology, the full connectivity is achieved (see \textit{e.g.} \cite{2017Monroe,2019Wright}) and any $X_{ij}$ operation can be performed directly
(\textit{i.e.} without involving other two-qubit operations on other qubits than $i$ and $j$). In contrast, the constraints in the LNN (Linear Nearest Neighbour) topology are strong as direct interaction is allowed only on
consecutive qubits, \textit{i.e.} only $X_{i,i+1}$ or $X_{i+1,i}$ \cite{2016Zajac}. Of course there are also many intermediate graph configurations as in the IBM superconducting transmon device  (\url{www.ibm.com/quantum-computing/}).
To implement a $X_{ij}$ gate when there is no arrow between $i$ and $j$, the usual method is to use $\swap$ gates together with allowed $\cnot$ gates as we did in Figure \ref{x03swap}. But, if $\swap$ gates are implemented on the device using 3 $\cnot$ gates (circuit equivalence \eqref{swap}), this implementation can be done in a more reliable fashion by using less gates : just find a shortest path $(i_1=i, \dots, i_p=j)$ between $i$ and $j$ in the undirected graph, then apply formula \eqref{chaslesform}.
Notice that it can be important to take into account the error rate associated to each arrow $(i,j)$ in the shortest path computation since it may vary from one gate to another (see Figure \ref{errors}). Case of $(i_k,i_{k+1})$ is not an arrow, use the classical equivalence \ref{HHXHH}  in Figure \ref{equi} to invert target and control. Of course this adds two single-qubit gates to the circuit but the impact on reliability is small since on current experimental quantum devices the fidelity of single qubit gates as the Hadamard gate is much higher than any two-qubit gate fidelity (\textit{e.g.} Figure \ref{errors} again). In the example below, we implement a $\cnot$ gate considering the topology of the 15-qubit ibmq\_16\_melbourne device :

  
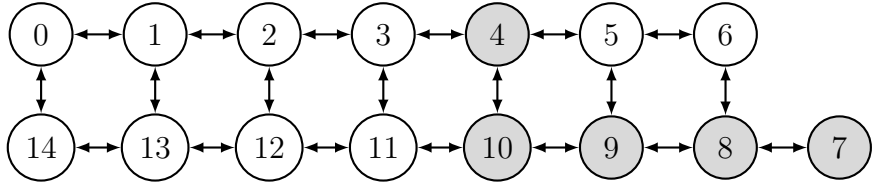
\begin{center}
    \begin{tikzpicture}[scale=1.5]
      \tikzstyle{sommet}=[circle,draw,thick,minimum width=2em]
      \tikzstyle{arrow}=[thick,<->,>=latex]
      \node[sommet] (q0) at (0,1) {0};
      \node[sommet] (q1) at (1,1) {1};
      \node[sommet] (q2) at (2,1) {2};
      \node[sommet] (q3) at (3,1) {3};
      \node[sommet,fill=gray!30] (q4) at (4,1) {4};
      \node[sommet] (q5) at (5,1) {5};
      \node[sommet] (q6) at (6,1) {6};
      \node[sommet,fill=gray!30] (q7) at (7,0) {7};
      \node[sommet,fill=gray!30] (q8) at (6,0) {8};
      \node[sommet,fill=gray!30] (q9) at (5,0) {9};
      \node[sommet,fill=gray!30] (q10) at (4,0) {10};
      \node[sommet] (q11) at (3,0) {11};
      \node[sommet] (q12) at (2,0) {12};
      \node[sommet] (q13) at (1,0) {13};
      \node[sommet] (q14) at (0,0) {14};
      \draw[arrow] (q0)--(q1);
      \draw[arrow] (q1)--(q2);
      \draw[arrow] (q2)--(q3);
      \draw[arrow] (q3)--(q4);
      \draw[arrow] (q4)--(q5);
      \draw[arrow] (q5)--(q6);
      \draw[arrow] (q6)--(q8);
      \draw[arrow] (q8)--(q9);
      \draw[arrow] (q9)--(q10);
      \draw[arrow] (q10)--(q11);
      \draw[arrow] (q11)--(q12);
      \draw[arrow] (q12)--(q13);
      \draw[arrow] (q13)--(q14);
      \draw[arrow] (q7)--(q8);
      \draw[arrow] (q0)--(q14);
      \draw[arrow] (q1)--(q13);
      \draw[arrow] (q2)--(q12);
      \draw[arrow] (q3)--(q11);
      \draw[arrow] (q4)--(q10);
      \draw[arrow] (q5)--(q9);
    \end{tikzpicture}\medskip

    Choosen path to implement $X_{4,7}$ : $(4,10,9,8,7)$\medskip
    
    $X_{4,7}=\left(X_{4,10}X_{10,9}X_{9,8}X_{8,7}X_{9,8}X_{10,9}\right)^2$    
  \end{center}

We end this section by remarking that $\GL[p]$ can be regarded as a subgroup of $\GL[n]$ if $p<n$. Indeed, let $\phi$ be the injective morphism from $\GL[p]$ into $\GL[n]$ defined by $\phi(M)=\begin{bmatrix}M&0\\0&I_{n-p}\end{bmatrix}$, one has $\GL[p]\simeq\phi(\GL[p])$, so we will consider that $\GL[p]\subset\GL[n]$ and no distinction will be made between matrices $M$ and $\phi(M)$.

\section{Optimization of $\cnot$ gates circuits}\label{general}

A $\cnot$  circuit is \emph{optimal} if there is not another equivalent $\cnot$ circuit with less gates. In the same way a decomposition of a matrix $M \in \GL$ in product of transvections $T_{i,j}$ is \emph{optimal} if $M$ cannot be written with less transvections.
If the length of a  circuit $C'$ is less than the length of an equivalent circuit $C$  we say that $C'$ is a \emph{reduction} of $C$ (or that $C$ has been \emph{reduced} to $C'$) and if $C'$ is \emph{optimal} we say that $C'$ is an \emph{optimization} of $C$ (or that $C$ has been \emph{optimized} to $C'$).

For a better readability we denote from now on the transvection matrix $T_{ij}$ by $[ij]$ and the permutation matrix $P_{\sigma}$
by $\sigma$. As a consequence,  $(ij)$ denotes the transposition which exchanges $i$ and $j$ as well as the transposition matrix $P_{(ij)}$.
Notice that $(ij)=(ji)$ but $[ij]=[ji]^{\mathrm{T}}$ ($[ij]$ is the transpose matrix of $[ji])$. With these notations the braid relation $\eqref{braid}$ becomes
$(ij)=(ji)=[ij][ji][ij]=[ji][ij][ji]$ and Identity \eqref{actiontsym} becomes $\sigma [ij]\sigma^{-1}=[ij]^{\sigma}=[\sigma(i)\sigma(j)]$. In particular, $(ij)[ij](ij)=[ij]^{(ij)}=[ji]$.
\medskip

From Section \ref{algebra}, any $\cnot$ circuit can be optimized using the following rewriting rules.
\begin{prop} \label{reduction}Let $0\leq i,j,k,l\leq n-1$ be distinct integers : 
  \begin{align}    
    &[ij]^2=1\label{invT}\\
    &[ij][jk]\mathbf{[ik]}=\mathbf{[ik]}[ij][jk]=[jk][ij] \ ;\  [ij][ki]\mathbf{[kj]}=\mathbf{[kj]}[ij][ki]=[ki][ij] \label{redcom}\\
    &\mathbf{[ij]}[jk]\mathbf{[ij]}=[jk][ik]=[ik][jk] \ ;\ \mathbf{[ij]}[ki]\mathbf{[ij]}=[ki][kj]=[kj][ki]\label{conjT}\\
    &[ij][kl]=[kl][ij]\ ;\ [ij][ik]=[ik][ij] \ ;\  [ij][kj]=[kj][ij] \label{comT}\\
    &\sigma[ij]=[ij]^{\sigma}\sigma=[\sigma(i)\sigma(j)]\sigma \ ;\ [ij]\sigma=\sigma[ij]^{\sigma^{-1}}=\sigma[\sigma^{-1}(i)\sigma^{-1}(j)]\label{permcom}
  \end{align}
\end{prop}

In practice one tries to find an \emph{ad hoc} sequence of the above rules to reduce or to optimize  a given circuit (see Figure \ref{adhoc} for an example) but the computation can be tricky even with a few qubits. So it is difficult to build a general reduction/optimization algorithm from Proposition \ref{reduction}.

\begin{figure}[h]
Circuit $C$ to reduce : \ \raisebox{-9mm}{
\begin{tikzpicture}[scale=1.200000,x=1pt,y=1pt]
\filldraw[color=white] (0.000000, -7.500000) rectangle (108.000000, 37.500000);
\draw[color=black] (0.000000,30.000000) -- (108.000000,30.000000);
\draw[color=black] (0.000000,30.000000) node[left] {$q_0$};
\draw[color=black] (0.000000,15.000000) -- (108.000000,15.000000);
\draw[color=black] (0.000000,15.000000) node[left] {$q_1$};
\draw[color=black] (0.000000,0.000000) -- (108.000000,0.000000);
\draw[color=black] (0.000000,0.000000) node[left] {$q_2$};
\draw (9.000000,30.000000) -- (9.000000,15.000000);
\begin{scope}
\draw[fill=white] (9.000000, 30.000000) circle(3.000000pt);
\clip (9.000000, 30.000000) circle(3.000000pt);
\draw (6.000000, 30.000000) -- (12.000000, 30.000000);
\draw (9.000000, 27.000000) -- (9.000000, 33.000000);
\end{scope}
\filldraw (9.000000, 15.000000) circle(1.500000pt);
\draw (27.000000,30.000000) -- (27.000000,0.000000);
\begin{scope}
\draw[fill=white] (27.000000, 0.000000) circle(3.000000pt);
\clip (27.000000, 0.000000) circle(3.000000pt);
\draw (24.000000, 0.000000) -- (30.000000, 0.000000);
\draw (27.000000, -3.000000) -- (27.000000, 3.000000);
\end{scope}
\filldraw (27.000000, 30.000000) circle(1.500000pt);
\draw (45.000000,15.000000) -- (45.000000,0.000000);
\begin{scope}
\draw[fill=white] (45.000000, 15.000000) circle(3.000000pt);
\clip (45.000000, 15.000000) circle(3.000000pt);
\draw (42.000000, 15.000000) -- (48.000000, 15.000000);
\draw (45.000000, 12.000000) -- (45.000000, 18.000000);
\end{scope}
\filldraw (45.000000, 0.000000) circle(1.500000pt);
\draw (63.000000,30.000000) -- (63.000000,15.000000);
\begin{scope}
\draw[fill=white] (63.000000, 30.000000) circle(3.000000pt);
\clip (63.000000, 30.000000) circle(3.000000pt);
\draw (60.000000, 30.000000) -- (66.000000, 30.000000);
\draw (63.000000, 27.000000) -- (63.000000, 33.000000);
\end{scope}
\filldraw (63.000000, 15.000000) circle(1.500000pt);
\draw (81.000000,30.000000) -- (81.000000,0.000000);
\begin{scope}
\draw (78.878680, 27.878680) -- (83.121320, 32.121320);
\draw (78.878680, 32.121320) -- (83.121320, 27.878680);
\end{scope}
\begin{scope}
\draw (78.878680, -2.121320) -- (83.121320, 2.121320);
\draw (78.878680, 2.121320) -- (83.121320, -2.121320);
\end{scope}
\draw (99.000000,30.000000) -- (99.000000,15.000000);
\begin{scope}
\draw[fill=white] (99.000000, 15.000000) circle(3.000000pt);
\clip (99.000000, 15.000000) circle(3.000000pt);
\draw (96.000000, 15.000000) -- (102.000000, 15.000000);
\draw (99.000000, 12.000000) -- (99.000000, 18.000000);
\end{scope}
\filldraw (99.000000, 30.000000) circle(1.500000pt);
\end{tikzpicture}
}

  \begin{align*}
    M&=[10](02)[01][12][20][01]\\
     &\stackrel{\eqref{permcom}}{=}(02)[12][01][12][20][01]\\
     &\stackrel{\eqref{conjT}}{=}(02)[01][02][20][01]\\
     &\stackrel{\eqref{permcom}}{=}[21](02)[02][20][01]\\
     &=[21][02][20][02][02][20][01]\\
     &\stackrel{\eqref{invT}}{=}[21][02][01]\\
     &\stackrel{\eqref{redcom}}{=}[02][21]\\
  \end{align*}
  \vspace{-13mm}
  
  Optimized circuit $C'\sim C$ : \ \raisebox{-9mm}{
\begin{tikzpicture}[scale=1.200000,x=1pt,y=1pt]
\filldraw[color=white] (0.000000, -7.500000) rectangle (36.000000, 37.500000);
\draw[color=black] (0.000000,30.000000) -- (36.000000,30.000000);
\draw[color=black] (0.000000,30.000000) node[left] {$q_0$};
\draw[color=black] (0.000000,15.000000) -- (36.000000,15.000000);
\draw[color=black] (0.000000,15.000000) node[left] {$q_1$};
\draw[color=black] (0.000000,0.000000) -- (36.000000,0.000000);
\draw[color=black] (0.000000,0.000000) node[left] {$q_2$};
\draw (9.000000,15.000000) -- (9.000000,0.000000);
\begin{scope}
\draw[fill=white] (9.000000, 0.000000) circle(3.000000pt);
\clip (9.000000, 0.000000) circle(3.000000pt);
\draw (6.000000, 0.000000) -- (12.000000, 0.000000);
\draw (9.000000, -3.000000) -- (9.000000, 3.000000);
\end{scope}
\filldraw (9.000000, 15.000000) circle(1.500000pt);
\draw (27.000000,30.000000) -- (27.000000,0.000000);
\begin{scope}
\draw[fill=white] (27.000000, 30.000000) circle(3.000000pt);
\clip (27.000000, 30.000000) circle(3.000000pt);
\draw (24.000000, 30.000000) -- (30.000000, 30.000000);
\draw (27.000000, 27.000000) -- (27.000000, 33.000000);
\end{scope}
\filldraw (27.000000, 0.000000) circle(1.500000pt);
\end{tikzpicture}
}
{ \caption{ Optimizing a $\cnot$ circuit using an \emph{adhoc} sequence of reduction rules.\label{adhoc}}}
\end{figure}

  \medskip
  
  To optimize a $\cnot$ circuit of a few qubits we use a C ANSI program that builds the Cayley Graph of the group $\GL$ by Breadth-first search
  and then find in the graph the group element corresponding to that circuit.  This program optimizes in a few seconds any $\cnot$ circuit up to $5$ qubits and the source code can be downloaded  at
  \href{https://github.com/marcbataille/cnot-circuits/blob/master/optimization/cnot_opt.c}{\texttt{https://github.com/marcbataille/cnot-circuits}}. In the rest of this article we refer to it as "the computer optimization program". However, due to the exponential growth of the group order (for instance $\mathrm{Card}(\XG[6])=\np{20 158 709 760}$), the computer method fails  from 6 qubits with a basic PC. 
In this context we remark that the Gauss Jordan algorithm mentionned in Section \ref{algebra} can be used
as a heuristic method to reduce a given $\cnot$ gates circuit (see Figure \ref{Gauss} for a detailled example). It works as follows.

\begin{itemize}
  
\item Each $X_{ij}$ gate of a given circuit $C$ corresponds to a transvection matrix $[ij]$.
  Multiplying these transvection matrices yields a matrix $M$ in $\GL$ that we call \emph{the matrix of the circuit in dimension $n$}.

\item The Gauss-Jordan algorithm applied to the matrix $M$ gives a decomposition of $M$ under the form $M=KU$ where  $U$ is an upper triangular matrix
  and $K$ is a matrix that is a lower triangular matrix if and only if the element of $\F_2$ that appears on the diagonal at each step of the algorithm is 1, so we can choose it as pivot.
  In this case the matrix $M$ as a $LU$ (\emph{Lower Upper}) decomposition and this decomposition is unique. If at a given step the bit on the diagonal is 0 then we have to choose another pivot and the matrix $K$ is not triangular anymore.
    Morevoer the decomposition $M=KU$ is not unique as the matrix $K$ depends on the choice of the pivot at each step, so the algorithm is not deterministic in this case.

\item A decomposition of $K$ in product of transvections is directly obtained while executing the algorithm and a simple and direct decomposition of an upper triangular matrix $U=(u_{ij})$ in $\GL$ is
  $U=\prod_{j=n-1}^{1}\prod_{i=j-1}^{0}[ij]^{u_{ij}}$. We call this decomposition the \emph{canonical} decomposition of $U$ and symmetrically the canonical
  decomposition of a lower triangular matrix is  $L=\prod_{j=0}^{n-2}\prod_{i=j+1}^{n-1}[ij]^{l_{ij}}$.
  
\item Replacing each transvection $[ij]$ in the decomposition of $M$  by the corresponding $\cnot$ gate $X_{ij}$,
  we obtain a circuit $C'$ equivalent to $C$.  If the length of $C'$ is less than the length of $C$ then we have successfully reduced the circuit $C$. If not, the heuristic failed.
  
\end{itemize}

We note that the circuit $C'$ obtained by this process is generally  not optimal even for small values of $n$.
However the Gauss-Jordan algorithm has the great
advantage of providing  an upper bound on the optimal number of $\cnot$ gates in a circuit.

\begin{prop}\label{boundGauss}
Any $n$-qubit circuit of $\cnot$ gates is equivalent to a circuit composed of less than $n^2$ gates.
\end{prop}

\begin{proof}
  We apply the Gauss-Jordan elimination algorithm to get a decomposition of $M\in\GL$ in transvections. The first part of the algorithm consists in
  multiplying $M$ to the left by
  a sequence of transvections in order to obtain an upper triangular matrix $U$. For $k=0,1,\dots,n-1$, we consider the column $k$ and the entry of the matrix on the diagonal. If this entry is 1 we choose it as a pivot but if it is 0 we first need to put a 1 on the diagonal. This can be done by swaping row $k$ with a row $\ell$ whose entry $(\ell,k)$ is 1. However, this will cost 3 transvections and it is more economical to add row $\ell$ to row $k$, \emph{i.e.} multiplying to the left by $[k,\ell]$.
  The number of left multiplications by transvections necessary to have a pivot equal to 1 on the diagonal at each step is bounded by $n-1$ and the number of left multiplications necessary to eliminate the 1's below the pivot is bounded by $n-1+\dots+1=\frac{n(n-1)}{2}$. Hence the number of transvections that appears in the first part of the algorithm is bounded
  by $\frac{n(n+1)}{2}-1$. In the second part of the algorithm we just write the canonical decomposition of $U$ which contains at most $n-1+n-2+\cdots 1=\frac{n(n-1)}{2}$ transvections. Finally we see that the number of transvections in the resulting decomposition of $M$ is at most $n^2-1$.
  \end{proof}

  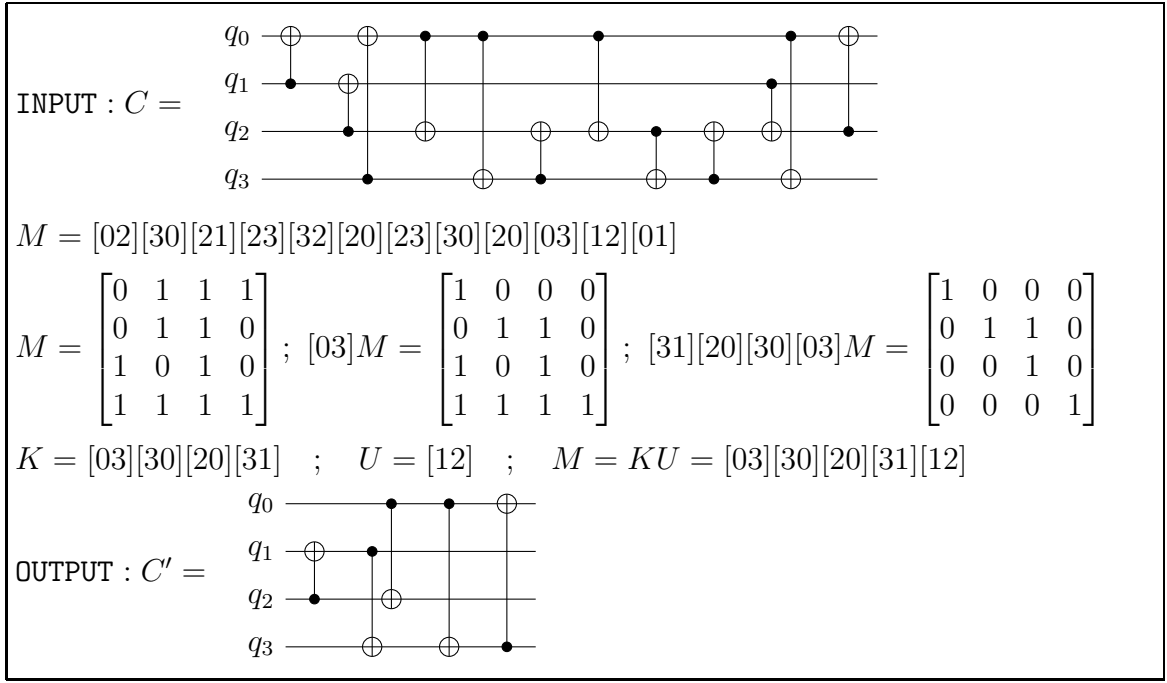
\begin{figure}[h]
    $\mathtt{INPUT} :  C=\ $ \raisebox{-12mm}{
\begin{tikzpicture}[scale=1.200000,x=1pt,y=1pt]
\filldraw[color=white] (0.000000, -7.500000) rectangle (192.000000, 52.500000);
\draw[color=black] (0.000000,45.000000) -- (192.000000,45.000000);
\draw[color=black] (0.000000,45.000000) node[left] {$q_0$};
\draw[color=black] (0.000000,30.000000) -- (192.000000,30.000000);
\draw[color=black] (0.000000,30.000000) node[left] {$q_1$};
\draw[color=black] (0.000000,15.000000) -- (192.000000,15.000000);
\draw[color=black] (0.000000,15.000000) node[left] {$q_2$};
\draw[color=black] (0.000000,0.000000) -- (192.000000,0.000000);
\draw[color=black] (0.000000,0.000000) node[left] {$q_3$};
\draw (9.000000,45.000000) -- (9.000000,30.000000);
\begin{scope}
\draw[fill=white] (9.000000, 45.000000) circle(3.000000pt);
\clip (9.000000, 45.000000) circle(3.000000pt);
\draw (6.000000, 45.000000) -- (12.000000, 45.000000);
\draw (9.000000, 42.000000) -- (9.000000, 48.000000);
\end{scope}
\filldraw (9.000000, 30.000000) circle(1.500000pt);
\draw (27.000000,30.000000) -- (27.000000,15.000000);
\begin{scope}
\draw[fill=white] (27.000000, 30.000000) circle(3.000000pt);
\clip (27.000000, 30.000000) circle(3.000000pt);
\draw (24.000000, 30.000000) -- (30.000000, 30.000000);
\draw (27.000000, 27.000000) -- (27.000000, 33.000000);
\end{scope}
\filldraw (27.000000, 15.000000) circle(1.500000pt);
\draw (33.000000,45.000000) -- (33.000000,0.000000);
\begin{scope}
\draw[fill=white] (33.000000, 45.000000) circle(3.000000pt);
\clip (33.000000, 45.000000) circle(3.000000pt);
\draw (30.000000, 45.000000) -- (36.000000, 45.000000);
\draw (33.000000, 42.000000) -- (33.000000, 48.000000);
\end{scope}
\filldraw (33.000000, 0.000000) circle(1.500000pt);
\draw (51.000000,45.000000) -- (51.000000,15.000000);
\begin{scope}
\draw[fill=white] (51.000000, 15.000000) circle(3.000000pt);
\clip (51.000000, 15.000000) circle(3.000000pt);
\draw (48.000000, 15.000000) -- (54.000000, 15.000000);
\draw (51.000000, 12.000000) -- (51.000000, 18.000000);
\end{scope}
\filldraw (51.000000, 45.000000) circle(1.500000pt);
\draw (69.000000,45.000000) -- (69.000000,0.000000);
\begin{scope}
\draw[fill=white] (69.000000, 0.000000) circle(3.000000pt);
\clip (69.000000, 0.000000) circle(3.000000pt);
\draw (66.000000, 0.000000) -- (72.000000, 0.000000);
\draw (69.000000, -3.000000) -- (69.000000, 3.000000);
\end{scope}
\filldraw (69.000000, 45.000000) circle(1.500000pt);
\draw (87.000000,15.000000) -- (87.000000,0.000000);
\begin{scope}
\draw[fill=white] (87.000000, 15.000000) circle(3.000000pt);
\clip (87.000000, 15.000000) circle(3.000000pt);
\draw (84.000000, 15.000000) -- (90.000000, 15.000000);
\draw (87.000000, 12.000000) -- (87.000000, 18.000000);
\end{scope}
\filldraw (87.000000, 0.000000) circle(1.500000pt);
\draw (105.000000,45.000000) -- (105.000000,15.000000);
\begin{scope}
\draw[fill=white] (105.000000, 15.000000) circle(3.000000pt);
\clip (105.000000, 15.000000) circle(3.000000pt);
\draw (102.000000, 15.000000) -- (108.000000, 15.000000);
\draw (105.000000, 12.000000) -- (105.000000, 18.000000);
\end{scope}
\filldraw (105.000000, 45.000000) circle(1.500000pt);
\draw (123.000000,15.000000) -- (123.000000,0.000000);
\begin{scope}
\draw[fill=white] (123.000000, 0.000000) circle(3.000000pt);
\clip (123.000000, 0.000000) circle(3.000000pt);
\draw (120.000000, 0.000000) -- (126.000000, 0.000000);
\draw (123.000000, -3.000000) -- (123.000000, 3.000000);
\end{scope}
\filldraw (123.000000, 15.000000) circle(1.500000pt);
\draw (141.000000,15.000000) -- (141.000000,0.000000);
\begin{scope}
\draw[fill=white] (141.000000, 15.000000) circle(3.000000pt);
\clip (141.000000, 15.000000) circle(3.000000pt);
\draw (138.000000, 15.000000) -- (144.000000, 15.000000);
\draw (141.000000, 12.000000) -- (141.000000, 18.000000);
\end{scope}
\filldraw (141.000000, 0.000000) circle(1.500000pt);
\draw (159.000000,30.000000) -- (159.000000,15.000000);
\begin{scope}
\draw[fill=white] (159.000000, 15.000000) circle(3.000000pt);
\clip (159.000000, 15.000000) circle(3.000000pt);
\draw (156.000000, 15.000000) -- (162.000000, 15.000000);
\draw (159.000000, 12.000000) -- (159.000000, 18.000000);
\end{scope}
\filldraw (159.000000, 30.000000) circle(1.500000pt);
\draw (165.000000,45.000000) -- (165.000000,0.000000);
\begin{scope}
\draw[fill=white] (165.000000, 0.000000) circle(3.000000pt);
\clip (165.000000, 0.000000) circle(3.000000pt);
\draw (162.000000, 0.000000) -- (168.000000, 0.000000);
\draw (165.000000, -3.000000) -- (165.000000, 3.000000);
\end{scope}
\filldraw (165.000000, 45.000000) circle(1.500000pt);
\draw (183.000000,45.000000) -- (183.000000,15.000000);
\begin{scope}
\draw[fill=white] (183.000000, 45.000000) circle(3.000000pt);
\clip (183.000000, 45.000000) circle(3.000000pt);
\draw (180.000000, 45.000000) -- (186.000000, 45.000000);
\draw (183.000000, 42.000000) -- (183.000000, 48.000000);
\end{scope}
\filldraw (183.000000, 15.000000) circle(1.500000pt);
\end{tikzpicture}
}
\medskip

$M=[02][30][21][23][32][20][23][30][20][03][12][01]$\medskip

    $M=\begin{bmatrix}
      0&1&1&1\\
      0&1&1&0\\
      1&0&1&0\\
      1&1&1&1\\
    \end{bmatrix};\
    [03]M=\begin{bmatrix}
      1&0&0&0\\
      0&1&1&0\\
      1&0&1&0\\
      1&1&1&1\\
    \end{bmatrix};\
    [31][20][30][03]M=\begin{bmatrix}
      1&0&0&0\\
      0&1&1&0\\
      0&0&1&0\\
      0&0&0&1\\
    \end{bmatrix}$\medskip
    
 $K=[03][30][20][31] \quad ;\quad U=[12] \quad ; \quad    M=KU=[03][30][20][31][12]$

    $\mathtt{OUTPUT}  : C'=\ $ \raisebox{-12mm}{
\begin{tikzpicture}[scale=1.200000,x=1pt,y=1pt]
\filldraw[color=white] (0.000000, -7.500000) rectangle (78.000000, 52.500000);
\draw[color=black] (0.000000,45.000000) -- (78.000000,45.000000);
\draw[color=black] (0.000000,45.000000) node[left] {$q_0$};
\draw[color=black] (0.000000,30.000000) -- (78.000000,30.000000);
\draw[color=black] (0.000000,30.000000) node[left] {$q_1$};
\draw[color=black] (0.000000,15.000000) -- (78.000000,15.000000);
\draw[color=black] (0.000000,15.000000) node[left] {$q_2$};
\draw[color=black] (0.000000,0.000000) -- (78.000000,0.000000);
\draw[color=black] (0.000000,0.000000) node[left] {$q_3$};
\draw (9.000000,30.000000) -- (9.000000,15.000000);
\begin{scope}
\draw[fill=white] (9.000000, 30.000000) circle(3.000000pt);
\clip (9.000000, 30.000000) circle(3.000000pt);
\draw (6.000000, 30.000000) -- (12.000000, 30.000000);
\draw (9.000000, 27.000000) -- (9.000000, 33.000000);
\end{scope}
\filldraw (9.000000, 15.000000) circle(1.500000pt);
\draw (27.000000,30.000000) -- (27.000000,0.000000);
\begin{scope}
\draw[fill=white] (27.000000, 0.000000) circle(3.000000pt);
\clip (27.000000, 0.000000) circle(3.000000pt);
\draw (24.000000, 0.000000) -- (30.000000, 0.000000);
\draw (27.000000, -3.000000) -- (27.000000, 3.000000);
\end{scope}
\filldraw (27.000000, 30.000000) circle(1.500000pt);
\draw (33.000000,45.000000) -- (33.000000,15.000000);
\begin{scope}
\draw[fill=white] (33.000000, 15.000000) circle(3.000000pt);
\clip (33.000000, 15.000000) circle(3.000000pt);
\draw (30.000000, 15.000000) -- (36.000000, 15.000000);
\draw (33.000000, 12.000000) -- (33.000000, 18.000000);
\end{scope}
\filldraw (33.000000, 45.000000) circle(1.500000pt);
\draw (51.000000,45.000000) -- (51.000000,0.000000);
\begin{scope}
\draw[fill=white] (51.000000, 0.000000) circle(3.000000pt);
\clip (51.000000, 0.000000) circle(3.000000pt);
\draw (48.000000, 0.000000) -- (54.000000, 0.000000);
\draw (51.000000, -3.000000) -- (51.000000, 3.000000);
\end{scope}
\filldraw (51.000000, 45.000000) circle(1.500000pt);
\draw (69.000000,45.000000) -- (69.000000,0.000000);
\begin{scope}
\draw[fill=white] (69.000000, 45.000000) circle(3.000000pt);
\clip (69.000000, 45.000000) circle(3.000000pt);
\draw (66.000000, 45.000000) -- (72.000000, 45.000000);
\draw (69.000000, 42.000000) -- (69.000000, 48.000000);
\end{scope}
\filldraw (69.000000, 0.000000) circle(1.500000pt);
\end{tikzpicture}
}
 \caption{ The Gauss-Jordan algorithm applied to a $\cnot$ gates circuit.\label{Gauss}}
    \end{figure}

    The canonical decomposition of an upper or lower triangular matrix is generally not optimal, so we propose in what follows another algorithm to decompose a triangular matrix in transvections. Applying this algorithm to the matrix $U$ when the decomposition of $M$ is $KU$ or to both matrices $L$ and $U$ when the decomposition of $M$ is $LU$ can often improve the decomposition given by the Gauss-Jordan heuristic, especially when the triangular matrix has a high density of 1's (See Figure \ref{LU} for an example). We describe below this algorithm for an upper triangular matrix (UTD algorithm) but it can be easily adapted to the case of a lower triangular matrix (LTD algorithm). 
    We denote by $|L_i|$ is the number of 1's in line $i$ of a matrix.

    \begin{algo} Upper Triangular Decomposition (UTD algorithm)
      
      $\mathtt{INPUT}$ : An upper triangular matrix $U$ in $\GL$.
      
      $\mathtt{OUTPUT}$ : A sequence  $S=(t_0,t_1,\dots,t_{p-1})$ of transvections such that  $U=\prod_{j=0}^{p-1}t_j$. 

      $\mathtt{1}\quad\ $ $S=()$;\  $j=0$;

      $\mathtt{2}\quad\ $ $\mathtt{while}\ U\neq I$ : 

      $\mathtt{3}\quad\ $\quad $\mathtt{for}\  i=0\  \mathtt{to}\ n-2$ :
      
      $\mathtt{4}\quad\ $\quad\quad $\mathtt{if}\ |L_i|>1$ :

      $\mathtt{5}\quad\ $\quad\quad\quad Let $E_i=\{ k \mid i<k\leq n-1 \text{ and } |L_i\oplus L_k|<|L_i|\}$;
      
      $\mathtt{6}\quad\ $\quad\quad\quad $\mathtt{if}\ E_i\neq \emptyset$ :

      $\mathtt{7}\quad\ $\quad\quad\quad\quad Choose the smallest $k\in E_i$ such that $|L_i\oplus L_k|$ is minimal;

      $\mathtt{8}\quad\ $\quad\quad\quad\quad $U=[ik]U$;

      $\mathtt{9}\quad\ $\quad\quad\quad\quad $t_j=[ik]$\ ;\  $j=j+1$;

      $\mathtt{10}\quad$ $\mathtt{return}\  S$;

    \end{algo}

    The algorithm ends because for each loop ($\mathtt{for}\  i=0\  \mathtt{to}\ n-2$) the number of 1's in the matrix decreases of at least one. Indeed, there is always at least one couple $(i,k)$ choosen during this loop, namely : $i=\mathrm{max}\{j\mid |L_j|>1\}$ and $k=\mathrm{max}\{j\mid U_{ij}\neq 0\}$.
    Furthermore, we notice that the number of transvections in the decomposition $S$ is always less than or equal to the number of transvections in the canonical decomposition of the matrix $U$ since each transvection in the canonical decomposition of $U$ contributes to a decrease of exactly one in the initial number of 1's in the matrix $U$ whereas each transvection in $S$ contributes to a decrease of at least one in this initial number of 1's.\medskip

    We observe that the decomposition in transvections obtained by applying the Gauss-Jordan algorithm followed by the UTD or LTD algorithm is generally not optimal.
To our knowledge, the best general algorithm to decompose a $\GL$-matrix in transvections is that proposed in  2004 by Patel \textit{et.al.} \cite{2004PMH}. It gives a decomposition in $O(n^2/\log n)$ transvections and shows experimentaly an improvement over Gaussian
elimination for $n$ as small as 8 according to the authors. However, the decomposition obtained using this algorithm is, again, not optimal.
     In order to better understand the optimization problem, it would be interesting to find out, for a given $n\geq 2$, what is the maximal number of transvections  required to decompose any matrix in $\GL$ (or equivalently the maximal length of an optimal circuit of $\XG$). Let us denote by $\mathrm{MaxT}(n)$ this number.
Thanks to a C ANSI program we checked that $\mathrm{MaxT}(n)=3(n-1)$ for any $n\leq 5$ (see the results summarized in Table \ref{stat}). The source code of the program can be downloaded at \href{https://github.com/marcbataille/cnot-circuits/blob/master/optimization/cnot_conj.c}{\texttt{https://github.com/marcbataille/cnot-circuits}}.
Unfortunately, this simple formula cannot be generalized to any $n$. This is a straightforward consequence of a result established by Patel \textit{et.al.} \cite[lemma 1]{2004PMH}. Indeed, the authors proved that
\begin{equation}\label{MaxT}
  \mathrm{MaxT}(n)>\dfrac{n^2-n}{\log_2(n^2-n+1)}.
\end{equation}
Their argument is as follows : there are $n(n-1)$ different $\cnot$ gates and by adding the identity we obtain a set of $n(n-1)+1$ gates. So the number of $\cnot$ circuits of length less than $B(n)$ is less than $(n(n-1)+1)^{\mathrm{MaxT}(n)}$ circuit. As this number of circuits is greater than the order of $\GL$, one has
$(n(n-1)+1)^{\mathrm{MaxT}(n)}>2^{\frac{n(n-1)}{2}}\prod_{i=1}^n(2^i-1)>2^{n(n-1)}$ and \ref{MaxT} follows.

We check that $\dfrac{n^2-n}{\log_2(n^2-n+1)}>3(n-1)$ if $n>29$. So, from $30$ qubits, there are certainly some $\cnot$ circuits that cannot be written with less than $3(n-1)$ $\cnot$ gates. However, further investigations are necessary to find out
up to how many qubits the linear value $3(n-1)$ for $\mathrm{MaxT}(n)$ remains valid.

    \begin{table}[h]
      \begin{center}
    \begin{tabular}{|c|c|c|c|c|}\hline
      optimal length&$n=2$&n=3&$n=4$&$n=5$\\\hline\hline
      $l=0$&1&1&1&1\\
      $l=1$&2&6&12&20\\
      $l=2$&2&24&96&260\\
      $l=3$&\textbf{1}&51&542&\np{2570}\\
      $l=4$&&60&\np{2058}&\np{19680}\\
      $l=5$&&24&\np{5316}&\np{117860}\\
      $l=6$&&\textbf{2}&\np{7530}&\np{540470}\\
      $l=7$&&&\np{4058}&\np{1769710}\\
      $l=8$&&&541&\np{3571175}\\
      $l=9$&&&\textbf{6}&\np{3225310}\\
      $l=10$&&&&\np{736540}\\
      $l=11$&&&&\np{15740}\\
      $l=12$&&&&\textbf{24}\\\hline\hline
      Order of $\XG$&6&168&\np{20160}&\np{9999360}\\\hline
    \end{tabular}
    \end{center}
      { \caption{ Number of elements of $\XG$ having an optimal decomposition of length $l$.
        The numbers in bold correspond to the $(n-1)!$ cyclic permutations of length $n$.\label{stat}}}
    \end{table}
    
    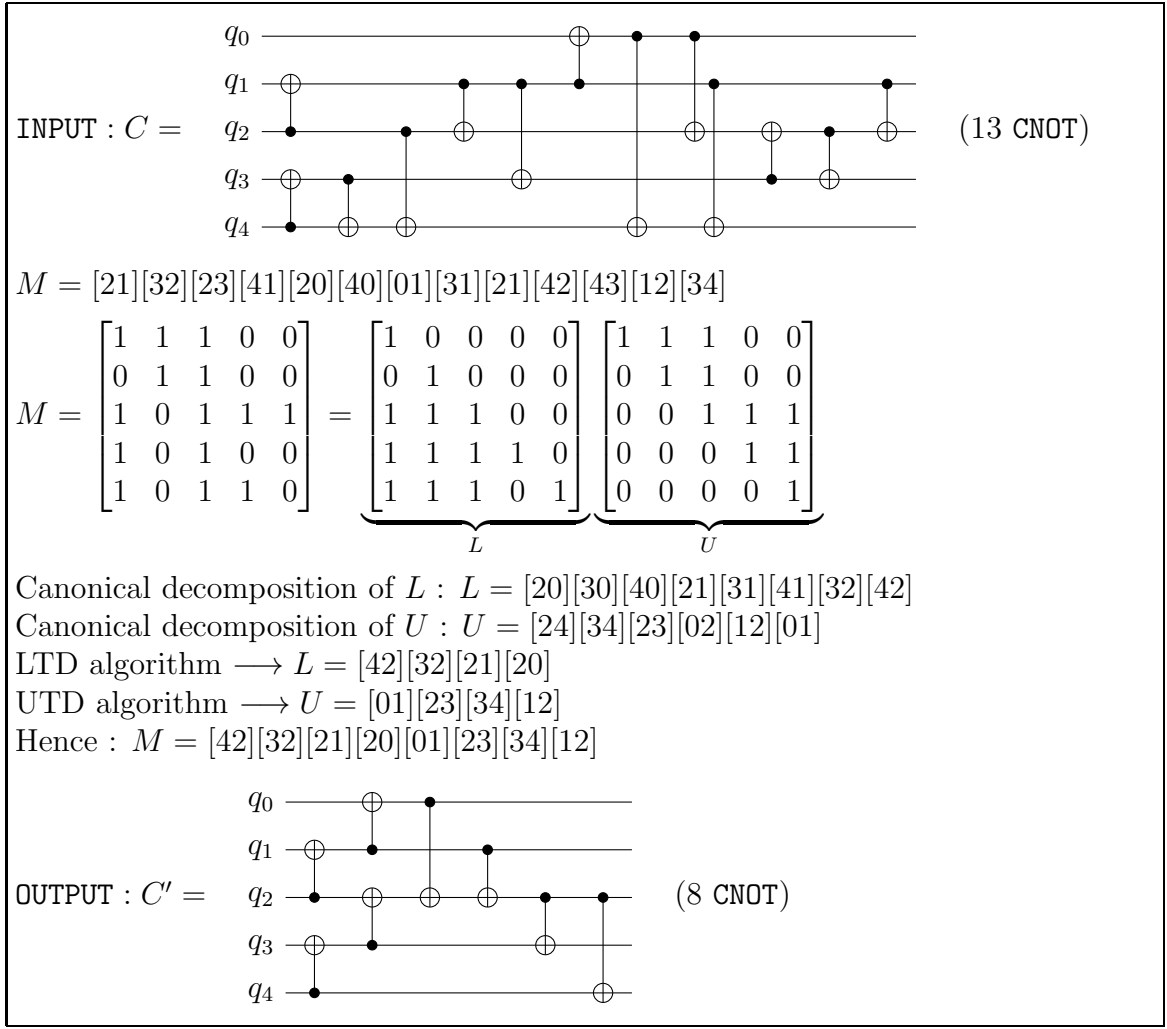
\begin{figure}[h]
      $\mathtt{INPUT} :   C=\ $ \raisebox{-15mm}{
        \begin{tikzpicture}[scale=1.200000,x=1pt,y=1pt]
          \filldraw[color=white] (0.000000, -7.500000) rectangle (204.000000, 67.500000);
          \draw[color=black] (0.000000,60.000000) -- (204.000000,60.000000);
          \draw[color=black] (0.000000,60.000000) node[left] {$q_0$};
          \draw[color=black] (0.000000,45.000000) -- (204.000000,45.000000);
          \draw[color=black] (0.000000,45.000000) node[left] {$q_1$};
          \draw[color=black] (0.000000,30.000000) -- (204.000000,30.000000);
          \draw[color=black] (0.000000,30.000000) node[left] {$q_2$};
          \draw[color=black] (0.000000,15.000000) -- (204.000000,15.000000);
          \draw[color=black] (0.000000,15.000000) node[left] {$q_3$};
          \draw[color=black] (0.000000,0.000000) -- (204.000000,0.000000);
          \draw[color=black] (0.000000,0.000000) node[left] {$q_4$};
          \draw (9.000000,45.000000) -- (9.000000,30.000000);
          \begin{scope}
            \draw[fill=white] (9.000000, 45.000000) circle(3.000000pt);
            \clip (9.000000, 45.000000) circle(3.000000pt);
            \draw (6.000000, 45.000000) -- (12.000000, 45.000000);
            \draw (9.000000, 42.000000) -- (9.000000, 48.000000);
          \end{scope}
          \filldraw (9.000000, 30.000000) circle(1.500000pt);
          \draw (9.000000,15.000000) -- (9.000000,0.000000);
          \begin{scope}
            \draw[fill=white] (9.000000, 15.000000) circle(3.000000pt);
            \clip (9.000000, 15.000000) circle(3.000000pt);
            \draw (6.000000, 15.000000) -- (12.000000, 15.000000);
            \draw (9.000000, 12.000000) -- (9.000000, 18.000000);
          \end{scope}
          \filldraw (9.000000, 0.000000) circle(1.500000pt);
          \draw (27.000000,15.000000) -- (27.000000,0.000000);
          \begin{scope}
            \draw[fill=white] (27.000000, 0.000000) circle(3.000000pt);
            \clip (27.000000, 0.000000) circle(3.000000pt);
            \draw (24.000000, 0.000000) -- (30.000000, 0.000000);
            \draw (27.000000, -3.000000) -- (27.000000, 3.000000);
          \end{scope}
          \filldraw (27.000000, 15.000000) circle(1.500000pt);
          \draw (45.000000,30.000000) -- (45.000000,0.000000);
          \begin{scope}
            \draw[fill=white] (45.000000, 0.000000) circle(3.000000pt);
            \clip (45.000000, 0.000000) circle(3.000000pt);
            \draw (42.000000, 0.000000) -- (48.000000, 0.000000);
            \draw (45.000000, -3.000000) -- (45.000000, 3.000000);
          \end{scope}
          \filldraw (45.000000, 30.000000) circle(1.500000pt);
          \draw (63.000000,45.000000) -- (63.000000,30.000000);
          \begin{scope}
            \draw[fill=white] (63.000000, 30.000000) circle(3.000000pt);
            \clip (63.000000, 30.000000) circle(3.000000pt);
            \draw (60.000000, 30.000000) -- (66.000000, 30.000000);
            \draw (63.000000, 27.000000) -- (63.000000, 33.000000);
          \end{scope}
          \filldraw (63.000000, 45.000000) circle(1.500000pt);
          \draw (81.000000,45.000000) -- (81.000000,15.000000);
          \begin{scope}
            \draw[fill=white] (81.000000, 15.000000) circle(3.000000pt);
            \clip (81.000000, 15.000000) circle(3.000000pt);
            \draw (78.000000, 15.000000) -- (84.000000, 15.000000);
            \draw (81.000000, 12.000000) -- (81.000000, 18.000000);
          \end{scope}
          \filldraw (81.000000, 45.000000) circle(1.500000pt);
          \draw (99.000000,60.000000) -- (99.000000,45.000000);
          \begin{scope}
            \draw[fill=white] (99.000000, 60.000000) circle(3.000000pt);
            \clip (99.000000, 60.000000) circle(3.000000pt);
            \draw (96.000000, 60.000000) -- (102.000000, 60.000000);
            \draw (99.000000, 57.000000) -- (99.000000, 63.000000);
          \end{scope}
          \filldraw (99.000000, 45.000000) circle(1.500000pt);
          \draw (117.000000,60.000000) -- (117.000000,0.000000);
          \begin{scope}
            \draw[fill=white] (117.000000, 0.000000) circle(3.000000pt);
            \clip (117.000000, 0.000000) circle(3.000000pt);
            \draw (114.000000, 0.000000) -- (120.000000, 0.000000);
            \draw (117.000000, -3.000000) -- (117.000000, 3.000000);
          \end{scope}
          \filldraw (117.000000, 60.000000) circle(1.500000pt);
          \draw (135.000000,60.000000) -- (135.000000,30.000000);
          \begin{scope}
            \draw[fill=white] (135.000000, 30.000000) circle(3.000000pt);
            \clip (135.000000, 30.000000) circle(3.000000pt);
            \draw (132.000000, 30.000000) -- (138.000000, 30.000000);
            \draw (135.000000, 27.000000) -- (135.000000, 33.000000);
          \end{scope}
          \filldraw (135.000000, 60.000000) circle(1.500000pt);
          \draw (141.000000,45.000000) -- (141.000000,0.000000);
          \begin{scope}
            \draw[fill=white] (141.000000, 0.000000) circle(3.000000pt);
            \clip (141.000000, 0.000000) circle(3.000000pt);
            \draw (138.000000, 0.000000) -- (144.000000, 0.000000);
            \draw (141.000000, -3.000000) -- (141.000000, 3.000000);
          \end{scope}
          \filldraw (141.000000, 45.000000) circle(1.500000pt);
          \draw (159.000000,30.000000) -- (159.000000,15.000000);
          \begin{scope}
            \draw[fill=white] (159.000000, 30.000000) circle(3.000000pt);
            \clip (159.000000, 30.000000) circle(3.000000pt);
            \draw (156.000000, 30.000000) -- (162.000000, 30.000000);
            \draw (159.000000, 27.000000) -- (159.000000, 33.000000);
          \end{scope}
          \filldraw (159.000000, 15.000000) circle(1.500000pt);
          \draw (177.000000,30.000000) -- (177.000000,15.000000);
          \begin{scope}
            \draw[fill=white] (177.000000, 15.000000) circle(3.000000pt);
            \clip (177.000000, 15.000000) circle(3.000000pt);
            \draw (174.000000, 15.000000) -- (180.000000, 15.000000);
            \draw (177.000000, 12.000000) -- (177.000000, 18.000000);
          \end{scope}
          \filldraw (177.000000, 30.000000) circle(1.500000pt);
          \draw (195.000000,45.000000) -- (195.000000,30.000000);
          \begin{scope}
            \draw[fill=white] (195.000000, 30.000000) circle(3.000000pt);
            \clip (195.000000, 30.000000) circle(3.000000pt);
            \draw (192.000000, 30.000000) -- (198.000000, 30.000000);
            \draw (195.000000, 27.000000) -- (195.000000, 33.000000);
          \end{scope}
          \filldraw (195.000000, 45.000000) circle(1.500000pt);
        \end{tikzpicture}
      }$\quad(13\ \cnot)$ \medskip

      $M=[21][32][23][41][20][40][01][31][21][42][43][12][34]$\medskip

      $M=\begin{bmatrix}
        1&1&1&0&0\\
        0&1&1&0&0\\
        1&0&1&1&1\\
        1&0&1&0&0\\
        1&0&1&1&0
      \end{bmatrix}=
      \underbrace{\begin{bmatrix}
        1&0&0&0&0\\
        0&1&0&0&0\\
        1&1&1&0&0\\
        1&1&1&1&0\\
        1&1&1&0&1
      \end{bmatrix}}_{L}
      \underbrace{\begin{bmatrix}
        1&1&1&0&0\\
        0&1&1&0&0\\
        0&0&1&1&1\\
        0&0&0&1&1\\
        0&0&0&0&1
      \end{bmatrix}}_{U}$\medskip

    Canonical decomposition of $L$ : $L=[20][30][40][21][31][41][32][42]$

    Canonical decomposition of $U$ : $U=[24][34][23][02][12][01]$

    LTD algorithm $\longrightarrow L=[42][32][21][20]$

    UTD algorithm $\longrightarrow U=[01][23][34][12]$

    Hence : $M=[42][32][21][20][01][23][34][12]$\medskip
      
      $\mathtt{OUTPUT} : C'=\ $ \raisebox{-15mm}{
\begin{tikzpicture}[scale=1.200000,x=1pt,y=1pt]
\filldraw[color=white] (0.000000, -7.500000) rectangle (108.000000, 67.500000);
\draw[color=black] (0.000000,60.000000) -- (108.000000,60.000000);
\draw[color=black] (0.000000,60.000000) node[left] {$q_0$};
\draw[color=black] (0.000000,45.000000) -- (108.000000,45.000000);
\draw[color=black] (0.000000,45.000000) node[left] {$q_1$};
\draw[color=black] (0.000000,30.000000) -- (108.000000,30.000000);
\draw[color=black] (0.000000,30.000000) node[left] {$q_2$};
\draw[color=black] (0.000000,15.000000) -- (108.000000,15.000000);
\draw[color=black] (0.000000,15.000000) node[left] {$q_3$};
\draw[color=black] (0.000000,0.000000) -- (108.000000,0.000000);
\draw[color=black] (0.000000,0.000000) node[left] {$q_4$};
\draw (9.000000,45.000000) -- (9.000000,30.000000);
\begin{scope}
\draw[fill=white] (9.000000, 45.000000) circle(3.000000pt);
\clip (9.000000, 45.000000) circle(3.000000pt);
\draw (6.000000, 45.000000) -- (12.000000, 45.000000);
\draw (9.000000, 42.000000) -- (9.000000, 48.000000);
\end{scope}
\filldraw (9.000000, 30.000000) circle(1.500000pt);
\draw (9.000000,15.000000) -- (9.000000,0.000000);
\begin{scope}
\draw[fill=white] (9.000000, 15.000000) circle(3.000000pt);
\clip (9.000000, 15.000000) circle(3.000000pt);
\draw (6.000000, 15.000000) -- (12.000000, 15.000000);
\draw (9.000000, 12.000000) -- (9.000000, 18.000000);
\end{scope}
\filldraw (9.000000, 0.000000) circle(1.500000pt);
\draw (27.000000,30.000000) -- (27.000000,15.000000);
\begin{scope}
\draw[fill=white] (27.000000, 30.000000) circle(3.000000pt);
\clip (27.000000, 30.000000) circle(3.000000pt);
\draw (24.000000, 30.000000) -- (30.000000, 30.000000);
\draw (27.000000, 27.000000) -- (27.000000, 33.000000);
\end{scope}
\filldraw (27.000000, 15.000000) circle(1.500000pt);
\draw (27.000000,60.000000) -- (27.000000,45.000000);
\begin{scope}
\draw[fill=white] (27.000000, 60.000000) circle(3.000000pt);
\clip (27.000000, 60.000000) circle(3.000000pt);
\draw (24.000000, 60.000000) -- (30.000000, 60.000000);
\draw (27.000000, 57.000000) -- (27.000000, 63.000000);
\end{scope}
\filldraw (27.000000, 45.000000) circle(1.500000pt);
\draw (45.000000,60.000000) -- (45.000000,30.000000);
\begin{scope}
\draw[fill=white] (45.000000, 30.000000) circle(3.000000pt);
\clip (45.000000, 30.000000) circle(3.000000pt);
\draw (42.000000, 30.000000) -- (48.000000, 30.000000);
\draw (45.000000, 27.000000) -- (45.000000, 33.000000);
\end{scope}
\filldraw (45.000000, 60.000000) circle(1.500000pt);
\draw (63.000000,45.000000) -- (63.000000,30.000000);
\begin{scope}
\draw[fill=white] (63.000000, 30.000000) circle(3.000000pt);
\clip (63.000000, 30.000000) circle(3.000000pt);
\draw (60.000000, 30.000000) -- (66.000000, 30.000000);
\draw (63.000000, 27.000000) -- (63.000000, 33.000000);
\end{scope}
\filldraw (63.000000, 45.000000) circle(1.500000pt);
\draw (81.000000,30.000000) -- (81.000000,15.000000);
\begin{scope}
\draw[fill=white] (81.000000, 15.000000) circle(3.000000pt);
\clip (81.000000, 15.000000) circle(3.000000pt);
\draw (78.000000, 15.000000) -- (84.000000, 15.000000);
\draw (81.000000, 12.000000) -- (81.000000, 18.000000);
\end{scope}
\filldraw (81.000000, 30.000000) circle(1.500000pt);
\draw (99.000000,30.000000) -- (99.000000,0.000000);
\begin{scope}
\draw[fill=white] (99.000000, 0.000000) circle(3.000000pt);
\clip (99.000000, 0.000000) circle(3.000000pt);
\draw (96.000000, 0.000000) -- (102.000000, 0.000000);
\draw (99.000000, -3.000000) -- (99.000000, 3.000000);
\end{scope}
\filldraw (99.000000, 30.000000) circle(1.500000pt);
\end{tikzpicture}
}$\quad(8\ \cnot)$
{ \caption{ The UTD/LTD algorithms applied to a $LU$ circuit of $\cnot$ gates \label{LU}}}
\end{figure}

      
\section{Optimization in subgroups of $\XG$}\label{subgroups}
In this section we describe the structure of some particular subgroups of $\XG\simeq \GL$ and we propose algorithms to optimize $\cnot$ circuits of these subgroups. Starting from an input circuit $C$, we first compute the matrix $M$ of the circuit in dimension $n$. In some specific cases the matrix $M$ has a particular shape so we can propose an optimal decomposition  in transvections for $M$ and output an optimization $C'$ of the circuit $C$. We describe  3 subgroups corresponding to 3 types of matrices with the objective of building step by step a kind of atlas of $\cnot$ circuits with specific optimization algorithms.
\subsection{Permutation matrices}
The first case considered is when the matrix $M\in\GL$ of the circuit is a permutation  matrix $P_{\sigma}$. In this case the circuit is equivalent to a $\swap$ gates circuit and we work in  $\SG$, the subgroup of $\XG$ generated by the $S_{(ij)}$ gates which is isomorphic to $\SYM$ (see example in Figure \ref{cyclic}).

\begin{prop}\label{sigmadecompo}
Any permutation matrix $P_{\sigma}$ can be decomposed as a product of $3(n-p)$ transvections where $p$ is the number of cycles of the permutation $\sigma$.
  \end{prop}
  \begin{proof}Let $\lambda=(n_1,\dots,n_p)$ be the cycle type of $\sigma$. 
     Any cycle of length $n_i$ can be decomposed in the product of $n_{i}-1$ transpositions and each transposition can be decomposed in a product of 3 transvections, so the total number of transvections used in the decomposition of $P_{\sigma}$ is $3\sum_{i=1}^{p}(n_i-1)=3(n-p)$. 
  \end{proof}

  The following conjecture has been checked up to $5$ qubits using the computer optimization program.

  \begin{conj}\label{sigmaconj}
The decomposition in transvections of a permutation matrix $P_{\sigma}$ given by Proposition \ref{sigmadecompo} is optimal.
    \end{conj}

  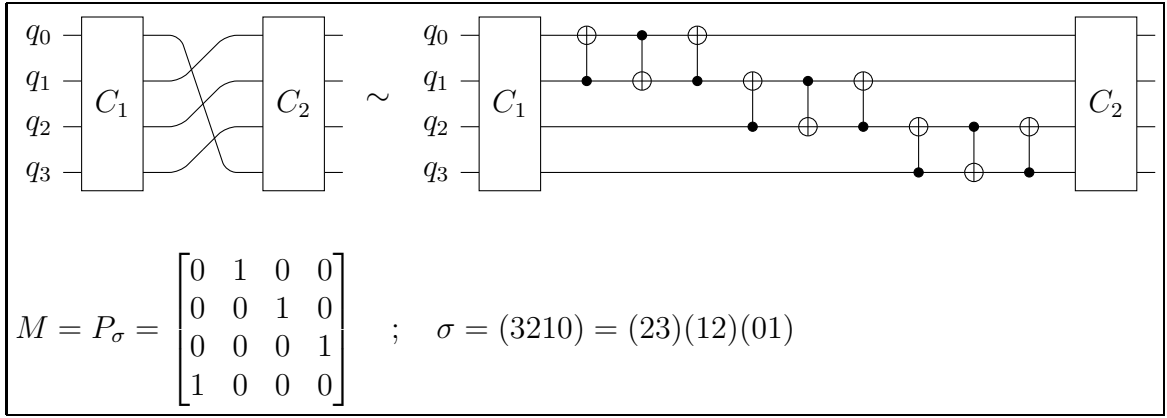
\begin{figure}
    \begin{tikzpicture}[scale=1.1500000,x=1pt,y=1pt]
\filldraw[color=white] (0.000000, -7.500000) rectangle (91.000000, 52.500000);
\draw[color=black,rounded corners=4.000000pt] (0.000000,45.000000) -- (38.000000,45.000000) -- (45.500000,22.500000);
\draw[color=black,rounded corners=4.000000pt] (45.500000,22.500000) -- (53.000000,0.000000) -- (91.000000,0.000000);
\draw[color=black] (0.000000,45.000000) node[left] {$q_0$};
\draw[color=black,rounded corners=4.000000pt] (0.000000,30.000000) -- (38.000000,30.000000) -- (45.500000,37.500000);
\draw[color=black,rounded corners=4.000000pt] (45.500000,37.500000) -- (53.000000,45.000000) -- (91.000000,45.000000);
\draw[color=black] (0.000000,30.000000) node[left] {$q_1$};
\draw[color=black,rounded corners=4.000000pt] (0.000000,15.000000) -- (38.000000,15.000000) -- (45.500000,22.500000);
\draw[color=black,rounded corners=4.000000pt] (45.500000,22.500000) -- (53.000000,30.000000) -- (91.000000,30.000000);
\draw[color=black] (0.000000,15.000000) node[left] {$q_2$};
\draw[color=black,rounded corners=4.000000pt] (0.000000,0.000000) -- (38.000000,0.000000) -- (45.500000,7.500000);
\draw[color=black,rounded corners=4.000000pt] (45.500000,7.500000) -- (53.000000,15.000000) -- (91.000000,15.000000);
\draw[color=black] (0.000000,0.000000) node[left] {$q_3$};
\draw (16.000000,45.000000) -- (16.000000,0.000000);
\begin{scope}
\draw[fill=white] (16.000000, 22.500000) +(-45.000000:14.142136pt and 40.305087pt) -- +(45.000000:14.142136pt and 40.305087pt) -- +(135.000000:14.142136pt and 40.305087pt) -- +(225.000000:14.142136pt and 40.305087pt) -- cycle;
\clip (16.000000, 22.500000) +(-45.000000:14.142136pt and 40.305087pt) -- +(45.000000:14.142136pt and 40.305087pt) -- +(135.000000:14.142136pt and 40.305087pt) -- +(225.000000:14.142136pt and 40.305087pt) -- cycle;
\draw (16.000000, 22.500000) node {{$C_1$}};
\end{scope}
\draw (75.000000,45.000000) -- (75.000000,0.000000);
\begin{scope}
\draw[fill=white] (75.000000, 22.500000) +(-45.000000:14.142136pt and 40.305087pt) -- +(45.000000:14.142136pt and 40.305087pt) -- +(135.000000:14.142136pt and 40.305087pt) -- +(225.000000:14.142136pt and 40.305087pt) -- cycle;
\clip (75.000000, 22.500000) +(-45.000000:14.142136pt and 40.305087pt) -- +(45.000000:14.142136pt and 40.305087pt) -- +(135.000000:14.142136pt and 40.305087pt) -- +(225.000000:14.142136pt and 40.305087pt) -- cycle;
\draw (75.000000, 22.500000) node {{$C_2$}};
\end{scope}
\end{tikzpicture}
\raisebox{12mm}{$\ \sim\ $}
\begin{tikzpicture}[scale=1.1500000,x=1pt,y=1pt]
\filldraw[color=white] (0.000000, -7.500000) rectangle (226.000000, 52.500000);
\draw[color=black] (0.000000,45.000000) -- (226.000000,45.000000);
\draw[color=black] (0.000000,45.000000) node[left] {$q_0$};
\draw[color=black] (0.000000,30.000000) -- (226.000000,30.000000);
\draw[color=black] (0.000000,30.000000) node[left] {$q_1$};
\draw[color=black] (0.000000,15.000000) -- (226.000000,15.000000);
\draw[color=black] (0.000000,15.000000) node[left] {$q_2$};
\draw[color=black] (0.000000,0.000000) -- (226.000000,0.000000);
\draw[color=black] (0.000000,0.000000) node[left] {$q_3$};
\draw (16.000000,45.000000) -- (16.000000,0.000000);
\begin{scope}
\draw[fill=white] (16.000000, 22.500000) +(-45.000000:14.142136pt and 40.305087pt) -- +(45.000000:14.142136pt and 40.305087pt) -- +(135.000000:14.142136pt and 40.305087pt) -- +(225.000000:14.142136pt and 40.305087pt) -- cycle;
\clip (16.000000, 22.500000) +(-45.000000:14.142136pt and 40.305087pt) -- +(45.000000:14.142136pt and 40.305087pt) -- +(135.000000:14.142136pt and 40.305087pt) -- +(225.000000:14.142136pt and 40.305087pt) -- cycle;
\draw (16.000000, 22.500000) node {{$C_1$}};
\end{scope}
\draw (41.000000,45.000000) -- (41.000000,30.000000);
\begin{scope}
\draw[fill=white] (41.000000, 45.000000) circle(3.000000pt);
\clip (41.000000, 45.000000) circle(3.000000pt);
\draw (38.000000, 45.000000) -- (44.000000, 45.000000);
\draw (41.000000, 42.000000) -- (41.000000, 48.000000);
\end{scope}
\filldraw (41.000000, 30.000000) circle(1.500000pt);
\draw (59.000000,45.000000) -- (59.000000,30.000000);
\begin{scope}
\draw[fill=white] (59.000000, 30.000000) circle(3.000000pt);
\clip (59.000000, 30.000000) circle(3.000000pt);
\draw (56.000000, 30.000000) -- (62.000000, 30.000000);
\draw (59.000000, 27.000000) -- (59.000000, 33.000000);
\end{scope}
\filldraw (59.000000, 45.000000) circle(1.500000pt);
\draw (77.000000,45.000000) -- (77.000000,30.000000);
\begin{scope}
\draw[fill=white] (77.000000, 45.000000) circle(3.000000pt);
\clip (77.000000, 45.000000) circle(3.000000pt);
\draw (74.000000, 45.000000) -- (80.000000, 45.000000);
\draw (77.000000, 42.000000) -- (77.000000, 48.000000);
\end{scope}
\filldraw (77.000000, 30.000000) circle(1.500000pt);
\draw (95.000000,30.000000) -- (95.000000,15.000000);
\begin{scope}
\draw[fill=white] (95.000000, 30.000000) circle(3.000000pt);
\clip (95.000000, 30.000000) circle(3.000000pt);
\draw (92.000000, 30.000000) -- (98.000000, 30.000000);
\draw (95.000000, 27.000000) -- (95.000000, 33.000000);
\end{scope}
\filldraw (95.000000, 15.000000) circle(1.500000pt);
\draw (113.000000,30.000000) -- (113.000000,15.000000);
\begin{scope}
\draw[fill=white] (113.000000, 15.000000) circle(3.000000pt);
\clip (113.000000, 15.000000) circle(3.000000pt);
\draw (110.000000, 15.000000) -- (116.000000, 15.000000);
\draw (113.000000, 12.000000) -- (113.000000, 18.000000);
\end{scope}
\filldraw (113.000000, 30.000000) circle(1.500000pt);
\draw (131.000000,30.000000) -- (131.000000,15.000000);
\begin{scope}
\draw[fill=white] (131.000000, 30.000000) circle(3.000000pt);
\clip (131.000000, 30.000000) circle(3.000000pt);
\draw (128.000000, 30.000000) -- (134.000000, 30.000000);
\draw (131.000000, 27.000000) -- (131.000000, 33.000000);
\end{scope}
\filldraw (131.000000, 15.000000) circle(1.500000pt);
\draw (149.000000,15.000000) -- (149.000000,0.000000);
\begin{scope}
\draw[fill=white] (149.000000, 15.000000) circle(3.000000pt);
\clip (149.000000, 15.000000) circle(3.000000pt);
\draw (146.000000, 15.000000) -- (152.000000, 15.000000);
\draw (149.000000, 12.000000) -- (149.000000, 18.000000);
\end{scope}
\filldraw (149.000000, 0.000000) circle(1.500000pt);
\draw (167.000000,15.000000) -- (167.000000,0.000000);
\begin{scope}
\draw[fill=white] (167.000000, 0.000000) circle(3.000000pt);
\clip (167.000000, 0.000000) circle(3.000000pt);
\draw (164.000000, 0.000000) -- (170.000000, 0.000000);
\draw (167.000000, -3.000000) -- (167.000000, 3.000000);
\end{scope}
\filldraw (167.000000, 15.000000) circle(1.500000pt);
\draw (185.000000,15.000000) -- (185.000000,0.000000);
\begin{scope}
\draw[fill=white] (185.000000, 15.000000) circle(3.000000pt);
\clip (185.000000, 15.000000) circle(3.000000pt);
\draw (182.000000, 15.000000) -- (188.000000, 15.000000);
\draw (185.000000, 12.000000) -- (185.000000, 18.000000);
\end{scope}
\filldraw (185.000000, 0.000000) circle(1.500000pt);
\draw (210.000000,45.000000) -- (210.000000,0.000000);
\begin{scope}
\draw[fill=white] (210.000000, 22.500000) +(-45.000000:14.142136pt and 40.305087pt) -- +(45.000000:14.142136pt and 40.305087pt) -- +(135.000000:14.142136pt and 40.305087pt) -- +(225.000000:14.142136pt and 40.305087pt) -- cycle;
\clip (210.000000, 22.500000) +(-45.000000:14.142136pt and 40.305087pt) -- +(45.000000:14.142136pt and 40.305087pt) -- +(135.000000:14.142136pt and 40.305087pt) -- +(225.000000:14.142136pt and 40.305087pt) -- cycle;
\draw (210.000000, 22.500000) node {{$C_2$}};
\end{scope}
\end{tikzpicture}
\medskip
    
    $M=P_{\sigma}=\begin{bmatrix}0&1&0&0\\
      0&0&1&0\\
      0&0&0&1\\
      1&0&0&0\end{bmatrix}\quad;\quad \sigma=(3210)=(23)(12)(01)$
{ \caption{ A cyclic permutation of the qubits between two circuits $C_1$ and $C_2$  \label{cyclic}}}
  \end{figure}

\subsection{Block diagonal matrices}
  In the second case, we consider that the matrix $M$ of the circuit in dimension $n$ is a block diagonal matrix or the conjugate of a diagonal block matrix by a permutation matrix.
  Let $(n_1,\dots, n_p)$ be a tuple of $p$
  positive integers such that $\sum_{i=1}^pn_i=n$ and consider the matrix
  $M_S=\begin{bmatrix}A_1&0&\dots&0\\
    0&A_2&\ddots&\vdots\\
    \vdots&\ddots&\ddots&0\\
    0&\dots&0&A_p
  \end{bmatrix}$ where $S=(A_1,\dots,A_p)$ is a tuple of matrices $A_i$ in $\GL[n_i]$. Here $M_S$ is a block diagonal matrix and we shall see right after the case of the conjugate of a block diagonal matrix by a permutation matrix. Clearly $\{M_S\mid S\in\GL[n_1]\times\dots\times\GL[n_p]\}$ is a group isomorphic to the direct group product $\GL[n_1]\times\dots\times\GL[n_p]$ since

  $M_S=\begin{bmatrix}A_{1}&0&\dots&0\\
    0&I_{n_2}&\ddots&\vdots\\
    \vdots&\ddots&\ddots&0\\
    0&\dots&0&I_{n_p}
  \end{bmatrix}
  \times\dots\times
  \begin{bmatrix}I_{n_1}&0&\dots&0\\
    0&\ddots&\ddots&\vdots\\
    \vdots&\ddots&I_{n_{p-1}}&0\\
    0&\dots&0&A_p
  \end{bmatrix}$ where $I_{\alpha}$ is the identity matrix in dimension $\alpha$. This equality can be written as $M_S=\prod_{i=1}^{p}M_{i}$ where $M_i=\begin{bmatrix}I_{N_{i-1}}&0&0\\0&A_i&0\\0&0&I_{n-N_i}\end{bmatrix}$ and $N_i=\sum_{k=1}^in_k$.
   We recall that the conjugation $M^{\sigma}=P_{\sigma}MP_{\sigma}^{-1}$ of a square matrix $M$ by a permutation matrix  $P_{\sigma}$ has the following effect on $M$ :  if we denote by $L_i$ (resp. $L_i'$) and $C_i$ (resp. $C_i'$) the lines and columns of matrix $M$
   (resp. $M^{\sigma}$) then $L_i'=L_{\sigma^{-1}(i)}$ and $C_i'=C_{\sigma^{-1}(i)}$, hence if $M=(m_{ij})$ and $M'=(m'_{ij})$ then $m'_{ij}=m_{\sigma^{-1}(i)\sigma^{-1}(j)}$.
   So, if we consider now a tuple of $p$ permutations $(\sigma_1,\dots,\sigma_p)$ verifying the conditions
  $\sigma_i(n_1+n_2+\dots+n_{i-1})=0,\sigma_i(n_1+n_2+\dots+n_{i-1}+1)=1,\dots,\sigma_i(n_1+n_2+\dots +n_{i-1}+n_i-1)=n_{i}-1$ then 
$M_S=\prod_{i=1}^{p}{\left(M_{i}^{\sigma_i}\right)}^{\sigma_i^{-1}}=\prod_{i=1}^{p}\left(M_i'\right)^{\sigma_i^{-1}}$
where $M_i'=\begin{bmatrix}A_i&0\\0&I_{n-n_i}\end{bmatrix}\simeq A_i$. Moreover, for Relation \eqref{actiontsym}, the number of transvections in an optimal decomposition of any matrix $M$ in $\GL$ is the same as the number of transvections in an optimal decomposition of $M^{\sigma}$ for any permutation $\sigma$. 
 As a consequence, the product of the conjugates by $\sigma_i^{-1}$ of any optimal decomposition of the matrix $A_i$  yields an optimal decomposition of the matrix $M_S$   (see Figure \ref{cart1} for an example).

\begin{figure}[h]
Let $M_{(A_1,A_2)}=\begin{bmatrix}
   1&0&1&1&0&0&0\\
   1&0&1&0&0&0&0\\
   0&1&1&1&0&0&0\\
   1&1&1&1&0&0&0\\
   0&0&0&0&0&0&1\\
   0&0&0&0&1&1&1\\
   0&0&0&0&0&1&1\\
 \end{bmatrix}$
 where $A_1=\begin{bmatrix}
   1&0&1&1\\
   1&0&1&0\\
   0&1&1&1\\
   1&1&1&1
 \end{bmatrix}$ and
 
 $A_2=\begin{bmatrix}
   0&0&1\\
   1&1&1\\
   0&1&1
 \end{bmatrix}$.
  Let $\sigma_1=I_7$ and
 $\sigma_2=\begin{pmatrix}0&1&2&3&4&5&6\\3&4&5&6&0&1&2\end{pmatrix}$,
 then :

 $M_{(A_1,A_2)}=\begin{bmatrix}
   1&0&1&1&0&0&0\\
   1&0&1&0&0&0&0\\
   0&1&1&1&0&0&0\\
   1&1&1&1&0&0&0\\
   0&0&0&0&1&0&0\\
   0&0&0&0&0&1&0\\
   0&0&0&0&0&0&1\\
 \end{bmatrix}
 \begin{bmatrix}
   0&0&1&0&0&0&0\\
   1&1&1&0&0&0&0\\
   0&1&1&0&0&0&0\\
   0&0&0&1&0&0&0\\
   0&0&0&0&1&0&0\\
   0&0&0&0&0&1&0\\
   0&0&0&0&0&0&1\\
 \end{bmatrix}^{\sigma_2^{-1}}={M_1'}^{\sigma_1^{-1}}{M_2'}^{\sigma_2^{-1}}$.

Using the computer optimization program, one finds 

 $A_1=[13][01][30][21][13][02][01]\ ,\ A_2=[01][12][10][21][01]$.

 Hence an optimal decomposition for $M$ is 

$\begin{array}{ll}
M_{(A_1,A_2)}&=[13][01][30][21][13][02][01]([01][12][10][21][01])^{\sigma_2^{-1}}\\
             &=[13][01][30][21][13][02][01][45][56][54][65][45].
\end{array}$
 
{ \caption{ Decomposition of a matrix using cartesian product of subgroups.\label{cart1}}}
  \end{figure}
  \medskip
  
  The generalization of the previous situation to the case of  a matrix $M=M_S^{\sigma}$ where $M_S$ is a block diagonal matrix and $\sigma$ a permutation matrix is straightforward. Indeed, if we have an optimal decomposition of $M_S$ in transvections, then we deduce an optimal decomposition for $M_S^{\sigma}$ by conjugating by $\sigma$ each transvection in the decomposition of $M_S$. This is due to the fact that the number of transvections in an optimal decomposition of any matrix $M$ in $\GL$ is the same as the number of transvections in an optimal decomposition of $M^{\sigma}$ (see Relation \eqref{actiontsym}). In fact the problem we adress here is rather the following : how can we efficiently recognize that the matrix $M$ is of type $M=M_S^{\sigma}$ ? To do that we interpret the matrix $M$ of the circuit in dimension $n$ as the adjacency matrix of the directed  graph $\mathcal{G}(M)$ whose set of vertices is $V=\{0,1,\dots,n-1\}$ and whose set of edges is $E=\{(i,j)\mid m_{ij}=1\}$. We call $\mathcal{G}(M)$ \emph{the graph of the circuit}. The graph $(\mathcal{G}(M))^{\sigma}$ defined by the set of vertices $\sigma(V)=V$ and by the set of edges $E_n^{\sigma}=\{(\sigma(i),\sigma(j))\mid m_{ij}=1\}$
  is isomorphic to $\mathcal{G}(M)$. Besides, $E_n^{\sigma}=\{(i,j)\mid m_{\sigma^{-1}(i)\sigma^{-1}(j)}=1\}$ and $M^{\sigma}=\left(m_{\sigma^{-1}(i)\sigma^{-1}(j)}\right)$,
  so $(\mathcal{G}(M))^{\sigma}=\mathcal{G}(M^{\sigma})$.
    The graph $\mathcal{G}(M_S)$ is a disconnected graph whose $p$ connected components are
    $V_1=\{0,\dots,n_1-1\},V_2=\{n_1,\dots,n_1+n_2-1\},\dots,V_p=\{n_1+\dots + n_{p-1},\dots,n_1+\dots +n_{p-1}+n_p-1\}$
    , so the graph $\mathcal{G}(M_S^{\sigma})=(\mathcal{G}(M_S))^{\sigma}$ is also
    a disconnected graph whose $p$ connected components are $\sigma(V_1),\dots,\sigma(V_p)$.

    By applying a Breadth-first search or a Deepth-first search to the graph $\mathcal{G}(M)$ of a given circuit, one can get the connected components of $\mathcal{G}(M)$ in linear time (see \cite{HT1973} for a detailled description of the algorithm). Suppose that we find more than one connected component, say $C_1,\dots C_p$ and let $n_i$ be the cardinal of $C_i$ for $1\leq i\leq p$. Let $\sigma$ be any permutation verifying $\sigma(C_1)=\{0,\dots,n_1-1\}, \sigma(C_2)=\{n_1,\dots,n_1+n_2-1\},\dots, \sigma(C_p)=\{n_1+\dots + n_{p-1},\dots,n_1+\dots +n_{p-1}+n_p-1\}$. Then the matrix $M^{\sigma}$ is a block diagonal matrix $M_S$ and we obtain an optimal decomposition of $M^{\sigma}$ as explained before. Finally, we deduce  an optimal decomposition of $M=(M^{\sigma})^{\sigma^{-1}}$ by conjugating by $\sigma^{-1}$ this decomposition (see Figure \ref{cart2} for an example). Since we can find an optimal decomposition of any matrix of $\GL$ for $n\leq5$ (using the computer optimization program), then we can optimize any $\cnot$ circuit whose graph has connected components of size up to 5.

    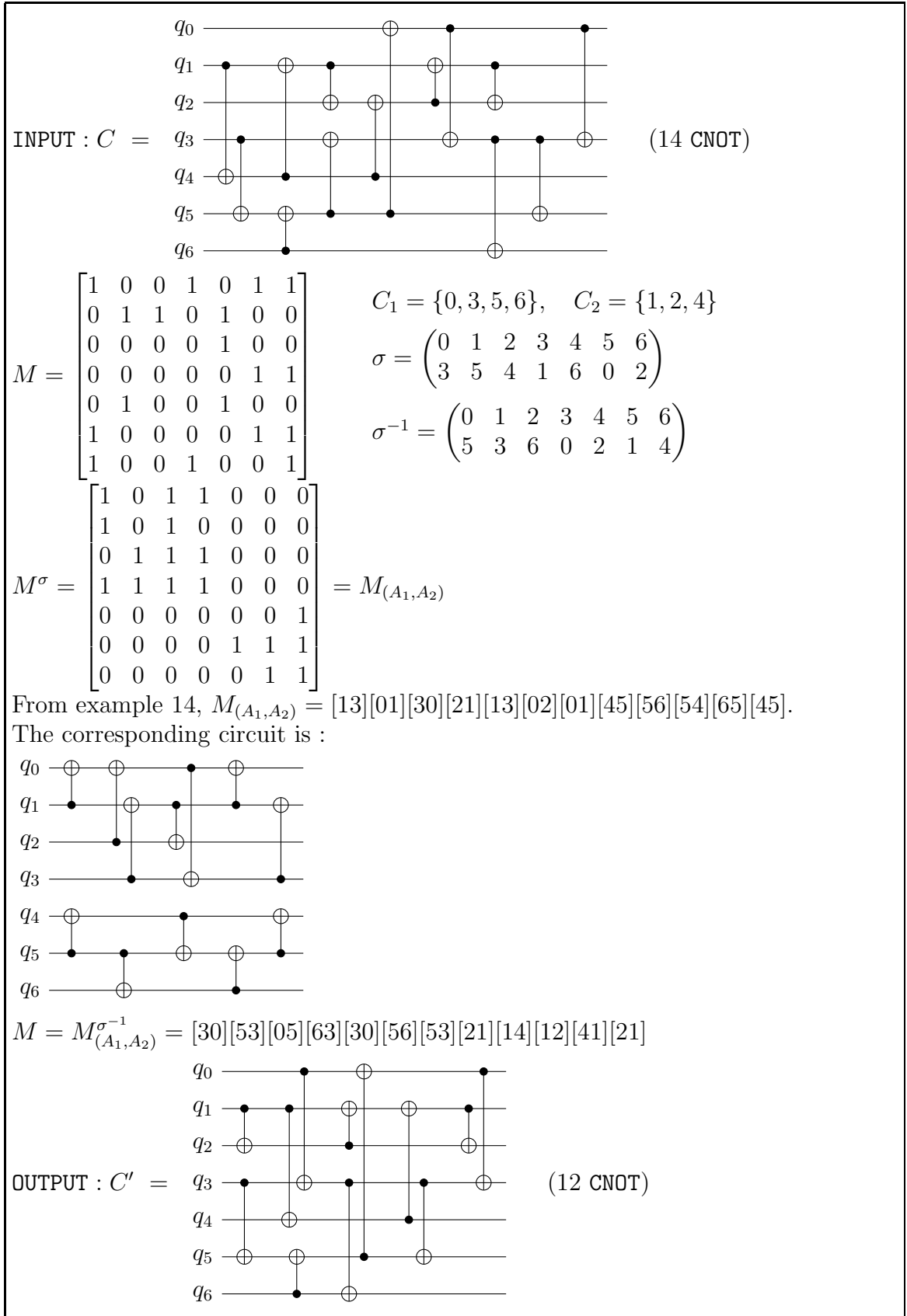
\begin{figure}
      $\mathtt{INPUT}: C\ =\ $\raisebox{-21mm}{
        \begin{tikzpicture}[scale=1.200000,x=1pt,y=1pt]
\filldraw[color=white] (0.000000, -7.500000) rectangle (162.000000, 97.500000);
\draw[color=black] (0.000000,90.000000) -- (162.000000,90.000000);
\draw[color=black] (0.000000,90.000000) node[left] {$q_0$};
\draw[color=black] (0.000000,75.000000) -- (162.000000,75.000000);
\draw[color=black] (0.000000,75.000000) node[left] {$q_1$};
\draw[color=black] (0.000000,60.000000) -- (162.000000,60.000000);
\draw[color=black] (0.000000,60.000000) node[left] {$q_2$};
\draw[color=black] (0.000000,45.000000) -- (162.000000,45.000000);
\draw[color=black] (0.000000,45.000000) node[left] {$q_3$};
\draw[color=black] (0.000000,30.000000) -- (162.000000,30.000000);
\draw[color=black] (0.000000,30.000000) node[left] {$q_4$};
\draw[color=black] (0.000000,15.000000) -- (162.000000,15.000000);
\draw[color=black] (0.000000,15.000000) node[left] {$q_5$};
\draw[color=black] (0.000000,0.000000) -- (162.000000,0.000000);
\draw[color=black] (0.000000,0.000000) node[left] {$q_6$};
\draw (9.000000,75.000000) -- (9.000000,30.000000);
\begin{scope}
\draw[fill=white] (9.000000, 30.000000) circle(3.000000pt);
\clip (9.000000, 30.000000) circle(3.000000pt);
\draw (6.000000, 30.000000) -- (12.000000, 30.000000);
\draw (9.000000, 27.000000) -- (9.000000, 33.000000);
\end{scope}
\filldraw (9.000000, 75.000000) circle(1.500000pt);
\draw (15.000000,45.000000) -- (15.000000,15.000000);
\begin{scope}
\draw[fill=white] (15.000000, 15.000000) circle(3.000000pt);
\clip (15.000000, 15.000000) circle(3.000000pt);
\draw (12.000000, 15.000000) -- (18.000000, 15.000000);
\draw (15.000000, 12.000000) -- (15.000000, 18.000000);
\end{scope}
\filldraw (15.000000, 45.000000) circle(1.500000pt);
\draw (33.000000,75.000000) -- (33.000000,30.000000);
\begin{scope}
\draw[fill=white] (33.000000, 75.000000) circle(3.000000pt);
\clip (33.000000, 75.000000) circle(3.000000pt);
\draw (30.000000, 75.000000) -- (36.000000, 75.000000);
\draw (33.000000, 72.000000) -- (33.000000, 78.000000);
\end{scope}
\filldraw (33.000000, 30.000000) circle(1.500000pt);
\draw (33.000000,15.000000) -- (33.000000,0.000000);
\begin{scope}
\draw[fill=white] (33.000000, 15.000000) circle(3.000000pt);
\clip (33.000000, 15.000000) circle(3.000000pt);
\draw (30.000000, 15.000000) -- (36.000000, 15.000000);
\draw (33.000000, 12.000000) -- (33.000000, 18.000000);
\end{scope}
\filldraw (33.000000, 0.000000) circle(1.500000pt);
\draw (51.000000,75.000000) -- (51.000000,60.000000);
\begin{scope}
\draw[fill=white] (51.000000, 60.000000) circle(3.000000pt);
\clip (51.000000, 60.000000) circle(3.000000pt);
\draw (48.000000, 60.000000) -- (54.000000, 60.000000);
\draw (51.000000, 57.000000) -- (51.000000, 63.000000);
\end{scope}
\filldraw (51.000000, 75.000000) circle(1.500000pt);
\draw (51.000000,45.000000) -- (51.000000,15.000000);
\begin{scope}
\draw[fill=white] (51.000000, 45.000000) circle(3.000000pt);
\clip (51.000000, 45.000000) circle(3.000000pt);
\draw (48.000000, 45.000000) -- (54.000000, 45.000000);
\draw (51.000000, 42.000000) -- (51.000000, 48.000000);
\end{scope}
\filldraw (51.000000, 15.000000) circle(1.500000pt);
\draw (69.000000,60.000000) -- (69.000000,30.000000);
\begin{scope}
\draw[fill=white] (69.000000, 60.000000) circle(3.000000pt);
\clip (69.000000, 60.000000) circle(3.000000pt);
\draw (66.000000, 60.000000) -- (72.000000, 60.000000);
\draw (69.000000, 57.000000) -- (69.000000, 63.000000);
\end{scope}
\filldraw (69.000000, 30.000000) circle(1.500000pt);
\draw (75.000000,90.000000) -- (75.000000,15.000000);
\begin{scope}
\draw[fill=white] (75.000000, 90.000000) circle(3.000000pt);
\clip (75.000000, 90.000000) circle(3.000000pt);
\draw (72.000000, 90.000000) -- (78.000000, 90.000000);
\draw (75.000000, 87.000000) -- (75.000000, 93.000000);
\end{scope}
\filldraw (75.000000, 15.000000) circle(1.500000pt);
\draw (93.000000,75.000000) -- (93.000000,60.000000);
\begin{scope}
\draw[fill=white] (93.000000, 75.000000) circle(3.000000pt);
\clip (93.000000, 75.000000) circle(3.000000pt);
\draw (90.000000, 75.000000) -- (96.000000, 75.000000);
\draw (93.000000, 72.000000) -- (93.000000, 78.000000);
\end{scope}
\filldraw (93.000000, 60.000000) circle(1.500000pt);
\draw (99.000000,90.000000) -- (99.000000,45.000000);
\begin{scope}
\draw[fill=white] (99.000000, 45.000000) circle(3.000000pt);
\clip (99.000000, 45.000000) circle(3.000000pt);
\draw (96.000000, 45.000000) -- (102.000000, 45.000000);
\draw (99.000000, 42.000000) -- (99.000000, 48.000000);
\end{scope}
\filldraw (99.000000, 90.000000) circle(1.500000pt);
\draw (117.000000,45.000000) -- (117.000000,0.000000);
\begin{scope}
\draw[fill=white] (117.000000, 0.000000) circle(3.000000pt);
\clip (117.000000, 0.000000) circle(3.000000pt);
\draw (114.000000, 0.000000) -- (120.000000, 0.000000);
\draw (117.000000, -3.000000) -- (117.000000, 3.000000);
\end{scope}
\filldraw (117.000000, 45.000000) circle(1.500000pt);
\draw (117.000000,75.000000) -- (117.000000,60.000000);
\begin{scope}
\draw[fill=white] (117.000000, 60.000000) circle(3.000000pt);
\clip (117.000000, 60.000000) circle(3.000000pt);
\draw (114.000000, 60.000000) -- (120.000000, 60.000000);
\draw (117.000000, 57.000000) -- (117.000000, 63.000000);
\end{scope}
\filldraw (117.000000, 75.000000) circle(1.500000pt);
\draw (135.000000,45.000000) -- (135.000000,15.000000);
\begin{scope}
\draw[fill=white] (135.000000, 15.000000) circle(3.000000pt);
\clip (135.000000, 15.000000) circle(3.000000pt);
\draw (132.000000, 15.000000) -- (138.000000, 15.000000);
\draw (135.000000, 12.000000) -- (135.000000, 18.000000);
\end{scope}
\filldraw (135.000000, 45.000000) circle(1.500000pt);
\draw (153.000000,90.000000) -- (153.000000,45.000000);
\begin{scope}
\draw[fill=white] (153.000000, 45.000000) circle(3.000000pt);
\clip (153.000000, 45.000000) circle(3.000000pt);
\draw (150.000000, 45.000000) -- (156.000000, 45.000000);
\draw (153.000000, 42.000000) -- (153.000000, 48.000000);
\end{scope}
\filldraw (153.000000, 90.000000) circle(1.500000pt);
\end{tikzpicture}
}
$\quad(14\ \cnot)$

\begin{minipage}[t]{.4\linewidth}
$M=\begin{bmatrix}
1&0&0&1&0&1&1\\
0&1&1&0&1&0&0\\
0&0&0&0&1&0&0\\
0&0&0&0&0&1&1\\
0&1&0&0&1&0&0\\
1&0&0&0&0&1&1\\
1&0&0&1&0&0&1
\end{bmatrix}$
\end{minipage}
\begin{minipage}[c]{.6\linewidth}
  $C_1=\{0,3,5,6\},\quad C_2=\{1,2,4\}$\medskip
  
  $ \sigma=\begin{pmatrix}0&1&2&3&4&5&6\\3&5&4&1&6&0&2\end{pmatrix}$\medskip

  $ \sigma^{-1}=\begin{pmatrix}0&1&2&3&4&5&6\\5&3&6&0&2&1&4\end{pmatrix}$
\end{minipage}
$M^{\sigma}=\begin{bmatrix}
   1&0&1&1&0&0&0\\
   1&0&1&0&0&0&0\\
   0&1&1&1&0&0&0\\
   1&1&1&1&0&0&0\\
   0&0&0&0&0&0&1\\
   0&0&0&0&1&1&1\\
   0&0&0&0&0&1&1\\
 \end{bmatrix}=M_{(A_1,A_2)}$

 From example \ref{cart1}, $M_{(A_1,A_2)}=[13][01][30][21][13][02][01][45][56][54][65][45]$.

 The corresponding circuit is :
 
\begin{tikzpicture}[scale=1.200000,x=1pt,y=1pt]
\filldraw[color=white] (0.000000, -7.500000) rectangle (102.000000, 97.500000);
\draw[color=black] (0.000000,90.000000) -- (102.000000,90.000000);
\draw[color=black] (0.000000,90.000000) node[left] {$q_0$};
\draw[color=black] (0.000000,75.000000) -- (102.000000,75.000000);
\draw[color=black] (0.000000,75.000000) node[left] {$q_1$};
\draw[color=black] (0.000000,60.000000) -- (102.000000,60.000000);
\draw[color=black] (0.000000,60.000000) node[left] {$q_2$};
\draw[color=black] (0.000000,45.000000) -- (102.000000,45.000000);
\draw[color=black] (0.000000,45.000000) node[left] {$q_3$};
\draw[color=black] (0.000000,30.000000) -- (102.000000,30.000000);
\draw[color=black] (0.000000,30.000000) node[left] {$q_4$};
\draw[color=black] (0.000000,15.000000) -- (102.000000,15.000000);
\draw[color=black] (0.000000,15.000000) node[left] {$q_5$};
\draw[color=black] (0.000000,0.000000) -- (102.000000,0.000000);
\draw[color=black] (0.000000,0.000000) node[left] {$q_6$};
\draw (9.000000,30.000000) -- (9.000000,15.000000);
\begin{scope}
\draw[fill=white] (9.000000, 30.000000) circle(3.000000pt);
\clip (9.000000, 30.000000) circle(3.000000pt);
\draw (6.000000, 30.000000) -- (12.000000, 30.000000);
\draw (9.000000, 27.000000) -- (9.000000, 33.000000);
\end{scope}
\filldraw (9.000000, 15.000000) circle(1.500000pt);
\draw (9.000000,90.000000) -- (9.000000,75.000000);
\begin{scope}
\draw[fill=white] (9.000000, 90.000000) circle(3.000000pt);
\clip (9.000000, 90.000000) circle(3.000000pt);
\draw (6.000000, 90.000000) -- (12.000000, 90.000000);
\draw (9.000000, 87.000000) -- (9.000000, 93.000000);
\end{scope}
\filldraw (9.000000, 75.000000) circle(1.500000pt);
\draw (30.000000,15.000000) -- (30.000000,0.000000);
\begin{scope}
\draw[fill=white] (30.000000, 0.000000) circle(3.000000pt);
\clip (30.000000, 0.000000) circle(3.000000pt);
\draw (27.000000, 0.000000) -- (33.000000, 0.000000);
\draw (30.000000, -3.000000) -- (30.000000, 3.000000);
\end{scope}
\filldraw (30.000000, 15.000000) circle(1.500000pt);
\draw (27.000000,90.000000) -- (27.000000,60.000000);
\begin{scope}
\draw[fill=white] (27.000000, 90.000000) circle(3.000000pt);
\clip (27.000000, 90.000000) circle(3.000000pt);
\draw (24.000000, 90.000000) -- (30.000000, 90.000000);
\draw (27.000000, 87.000000) -- (27.000000, 93.000000);
\end{scope}
\filldraw (27.000000, 60.000000) circle(1.500000pt);
\draw (33.000000,75.000000) -- (33.000000,45.000000);
\begin{scope}
\draw[fill=white] (33.000000, 75.000000) circle(3.000000pt);
\clip (33.000000, 75.000000) circle(3.000000pt);
\draw (30.000000, 75.000000) -- (36.000000, 75.000000);
\draw (33.000000, 72.000000) -- (33.000000, 78.000000);
\end{scope}
\filldraw (33.000000, 45.000000) circle(1.500000pt);
\draw (54.000000,30.000000) -- (54.000000,15.000000);
\begin{scope}
\draw[fill=white] (54.000000, 15.000000) circle(3.000000pt);
\clip (54.000000, 15.000000) circle(3.000000pt);
\draw (51.000000, 15.000000) -- (57.000000, 15.000000);
\draw (54.000000, 12.000000) -- (54.000000, 18.000000);
\end{scope}
\filldraw (54.000000, 30.000000) circle(1.500000pt);
\draw (51.000000,75.000000) -- (51.000000,60.000000);
\begin{scope}
\draw[fill=white] (51.000000, 60.000000) circle(3.000000pt);
\clip (51.000000, 60.000000) circle(3.000000pt);
\draw (48.000000, 60.000000) -- (54.000000, 60.000000);
\draw (51.000000, 57.000000) -- (51.000000, 63.000000);
\end{scope}
\filldraw (51.000000, 75.000000) circle(1.500000pt);
\draw (57.000000,90.000000) -- (57.000000,45.000000);
\begin{scope}
\draw[fill=white] (57.000000, 45.000000) circle(3.000000pt);
\clip (57.000000, 45.000000) circle(3.000000pt);
\draw (54.000000, 45.000000) -- (60.000000, 45.000000);
\draw (57.000000, 42.000000) -- (57.000000, 48.000000);
\end{scope}
\filldraw (57.000000, 90.000000) circle(1.500000pt);
\draw (75.000000,15.000000) -- (75.000000,0.000000);
\begin{scope}
\draw[fill=white] (75.000000, 15.000000) circle(3.000000pt);
\clip (75.000000, 15.000000) circle(3.000000pt);
\draw (72.000000, 15.000000) -- (78.000000, 15.000000);
\draw (75.000000, 12.000000) -- (75.000000, 18.000000);
\end{scope}
\filldraw (75.000000, 0.000000) circle(1.500000pt);
\draw (75.000000,90.000000) -- (75.000000,75.000000);
\begin{scope}
\draw[fill=white] (75.000000, 90.000000) circle(3.000000pt);
\clip (75.000000, 90.000000) circle(3.000000pt);
\draw (72.000000, 90.000000) -- (78.000000, 90.000000);
\draw (75.000000, 87.000000) -- (75.000000, 93.000000);
\end{scope}
\filldraw (75.000000, 75.000000) circle(1.500000pt);
\draw (93.000000,30.000000) -- (93.000000,15.000000);
\begin{scope}
\draw[fill=white] (93.000000, 30.000000) circle(3.000000pt);
\clip (93.000000, 30.000000) circle(3.000000pt);
\draw (90.000000, 30.000000) -- (96.000000, 30.000000);
\draw (93.000000, 27.000000) -- (93.000000, 33.000000);
\end{scope}
\filldraw (93.000000, 15.000000) circle(1.500000pt);
\draw (93.000000,75.000000) -- (93.000000,45.000000);
\begin{scope}
\draw[fill=white] (93.000000, 75.000000) circle(3.000000pt);
\clip (93.000000, 75.000000) circle(3.000000pt);
\draw (90.000000, 75.000000) -- (96.000000, 75.000000);
\draw (93.000000, 72.000000) -- (93.000000, 78.000000);
\end{scope}
\filldraw (93.000000, 45.000000) circle(1.500000pt);
\end{tikzpicture}

$M=M_{(A_1,A_2)}^{\sigma^{-1}}=[30][53][05][63][30][56][53][21][14][12][41][21]$

 $\mathtt{OUTPUT}: C'\ =\ $\raisebox{-21mm}{
 \begin{tikzpicture}[scale=1.200000,x=1pt,y=1pt]
\filldraw[color=white] (0.000000, -7.500000) rectangle (114.000000, 97.500000);
\draw[color=black] (0.000000,90.000000) -- (114.000000,90.000000);
\draw[color=black] (0.000000,90.000000) node[left] {$q_0$};
\draw[color=black] (0.000000,75.000000) -- (114.000000,75.000000);
\draw[color=black] (0.000000,75.000000) node[left] {$q_1$};
\draw[color=black] (0.000000,60.000000) -- (114.000000,60.000000);
\draw[color=black] (0.000000,60.000000) node[left] {$q_2$};
\draw[color=black] (0.000000,45.000000) -- (114.000000,45.000000);
\draw[color=black] (0.000000,45.000000) node[left] {$q_3$};
\draw[color=black] (0.000000,30.000000) -- (114.000000,30.000000);
\draw[color=black] (0.000000,30.000000) node[left] {$q_4$};
\draw[color=black] (0.000000,15.000000) -- (114.000000,15.000000);
\draw[color=black] (0.000000,15.000000) node[left] {$q_5$};
\draw[color=black] (0.000000,0.000000) -- (114.000000,0.000000);
\draw[color=black] (0.000000,0.000000) node[left] {$q_6$};
\draw (9.000000,75.000000) -- (9.000000,60.000000);
\begin{scope}
\draw[fill=white] (9.000000, 60.000000) circle(3.000000pt);
\clip (9.000000, 60.000000) circle(3.000000pt);
\draw (6.000000, 60.000000) -- (12.000000, 60.000000);
\draw (9.000000, 57.000000) -- (9.000000, 63.000000);
\end{scope}
\filldraw (9.000000, 75.000000) circle(1.500000pt);
\draw (9.000000,45.000000) -- (9.000000,15.000000);
\begin{scope}
\draw[fill=white] (9.000000, 15.000000) circle(3.000000pt);
\clip (9.000000, 15.000000) circle(3.000000pt);
\draw (6.000000, 15.000000) -- (12.000000, 15.000000);
\draw (9.000000, 12.000000) -- (9.000000, 18.000000);
\end{scope}
\filldraw (9.000000, 45.000000) circle(1.500000pt);
\draw (27.000000,75.000000) -- (27.000000,30.000000);
\begin{scope}
\draw[fill=white] (27.000000, 30.000000) circle(3.000000pt);
\clip (27.000000, 30.000000) circle(3.000000pt);
\draw (24.000000, 30.000000) -- (30.000000, 30.000000);
\draw (27.000000, 27.000000) -- (27.000000, 33.000000);
\end{scope}
\filldraw (27.000000, 75.000000) circle(1.500000pt);
\draw (30.000000,15.000000) -- (30.000000,0.000000);
\begin{scope}
\draw[fill=white] (30.000000, 15.000000) circle(3.000000pt);
\clip (30.000000, 15.000000) circle(3.000000pt);
\draw (27.000000, 15.000000) -- (33.000000, 15.000000);
\draw (30.000000, 12.000000) -- (30.000000, 18.000000);
\end{scope}
\filldraw (30.000000, 0.000000) circle(1.500000pt);
\draw (33.000000,90.000000) -- (33.000000,45.000000);
\begin{scope}
\draw[fill=white] (33.000000, 45.000000) circle(3.000000pt);
\clip (33.000000, 45.000000) circle(3.000000pt);
\draw (30.000000, 45.000000) -- (36.000000, 45.000000);
\draw (33.000000, 42.000000) -- (33.000000, 48.000000);
\end{scope}
\filldraw (33.000000, 90.000000) circle(1.500000pt);
\draw (51.000000,75.000000) -- (51.000000,60.000000);
\begin{scope}
\draw[fill=white] (51.000000, 75.000000) circle(3.000000pt);
\clip (51.000000, 75.000000) circle(3.000000pt);
\draw (48.000000, 75.000000) -- (54.000000, 75.000000);
\draw (51.000000, 72.000000) -- (51.000000, 78.000000);
\end{scope}
\filldraw (51.000000, 60.000000) circle(1.500000pt);
\draw (51.000000,45.000000) -- (51.000000,0.000000);
\begin{scope}
\draw[fill=white] (51.000000, 0.000000) circle(3.000000pt);
\clip (51.000000, 0.000000) circle(3.000000pt);
\draw (48.000000, 0.000000) -- (54.000000, 0.000000);
\draw (51.000000, -3.000000) -- (51.000000, 3.000000);
\end{scope}
\filldraw (51.000000, 45.000000) circle(1.500000pt);
\draw (57.000000,90.000000) -- (57.000000,15.000000);
\begin{scope}
\draw[fill=white] (57.000000, 90.000000) circle(3.000000pt);
\clip (57.000000, 90.000000) circle(3.000000pt);
\draw (54.000000, 90.000000) -- (60.000000, 90.000000);
\draw (57.000000, 87.000000) -- (57.000000, 93.000000);
\end{scope}
\filldraw (57.000000, 15.000000) circle(1.500000pt);
\draw (75.000000,75.000000) -- (75.000000,30.000000);
\begin{scope}
\draw[fill=white] (75.000000, 75.000000) circle(3.000000pt);
\clip (75.000000, 75.000000) circle(3.000000pt);
\draw (72.000000, 75.000000) -- (78.000000, 75.000000);
\draw (75.000000, 72.000000) -- (75.000000, 78.000000);
\end{scope}
\filldraw (75.000000, 30.000000) circle(1.500000pt);
\draw (81.000000,45.000000) -- (81.000000,15.000000);
\begin{scope}
\draw[fill=white] (81.000000, 15.000000) circle(3.000000pt);
\clip (81.000000, 15.000000) circle(3.000000pt);
\draw (78.000000, 15.000000) -- (84.000000, 15.000000);
\draw (81.000000, 12.000000) -- (81.000000, 18.000000);
\end{scope}
\filldraw (81.000000, 45.000000) circle(1.500000pt);
\draw (99.000000,75.000000) -- (99.000000,60.000000);
\begin{scope}
\draw[fill=white] (99.000000, 60.000000) circle(3.000000pt);
\clip (99.000000, 60.000000) circle(3.000000pt);
\draw (96.000000, 60.000000) -- (102.000000, 60.000000);
\draw (99.000000, 57.000000) -- (99.000000, 63.000000);
\end{scope}
\filldraw (99.000000, 75.000000) circle(1.500000pt);
\draw (105.000000,90.000000) -- (105.000000,45.000000);
\begin{scope}
\draw[fill=white] (105.000000, 45.000000) circle(3.000000pt);
\clip (105.000000, 45.000000) circle(3.000000pt);
\draw (102.000000, 45.000000) -- (108.000000, 45.000000);
\draw (105.000000, 42.000000) -- (105.000000, 48.000000);
\end{scope}
\filldraw (105.000000, 90.000000) circle(1.500000pt);
\end{tikzpicture}
}
$\quad(12\ \cnot)$
{ \caption{ Circuit optimization using a cartesian product of subgroups. \label{cart2}}}
\end{figure}
\subsection{Bit reverse of permutation matrices\label{bitrev}}
The third and last case we consider in this section is the case of a matrix $M$  of type $\overline{\sigma}$ where $\overline{\sigma}$  denotes a permutation matrix in even dimension in which all the bits are reversed. For instance $\overline{(02)(13)}=\begin{bmatrix}1&1&0&1\\1&1&1&0\\0&1&1&1\\1&0&1&1\end{bmatrix}$ and
$\overline{I_4}=\begin{bmatrix}0&1&1&1\\1&0&1&1\\1&1&0&1\\1&1&1&0\end{bmatrix}$.

We check that $\overline{\sigma}=\sigma\overline{I}$. For instance $\overline{(02)(13)}=\begin{bmatrix}0&0&1&0\\0&0&0&1\\1&0&0&0\\0&1&0&0\end{bmatrix}\begin{bmatrix}0&1&1&1\\1&0&1&1\\1&1&0&1\\1&1&1&0\end{bmatrix}$.

\begin{lem}\label{ibar1}Let $n\geq  2$ be an even integer.  
  \begin{align}
    &\overline{I_n}\in\GL \text{ and } \overline{I_n}^2=I_n.\label{ibarinv}\\
    &\forall\sigma\in\SYM,\ \sigma\overline{I_n}\sigma^{-1}=\overline{I_n}.\label{ibarconj}\\
    &\forall\sigma,\gamma\in\SYM,\ \overline{\sigma}\;\overline{\gamma}=\sigma\gamma.\label{ibarprod}
  \end{align}
\end{lem}

\begin{proof}Let denote by $1_n$ the matrix of dimension $n$ whose elements are all equal to 1. One has
  $\overline{I_n}^2=(I_n\oplus 1_n)^2=I_n^2\oplus 1_n^2=I_n\oplus 1_n^2$. If $n$ is even then $1_n^2=0$, hence
  $\overline{I_n}^2=I_n.$
  Besides $\sigma\overline{I_n}\sigma^{-1}=\sigma(I_n\oplus 1_n)\sigma^{-1}=\sigma I_n\sigma^{-1}\oplus \sigma 1_n \sigma^{-1}=I_n\oplus 1_n=\overline{I_n}$.
  Finally $\overline{\sigma}\;\overline{\gamma}=\sigma\overline{I_n}\gamma\overline{I_n}=\sigma\gamma\gamma^{-1}\overline{I_n}\gamma\overline{I_n}=\sigma\gamma\overline{I_n}^2=\sigma\gamma$.
\end{proof}

From Lemma \ref{ibar1} we directly deduce the structure of the group $<\sigma,\overline{\sigma}>$ :

\begin{prop} The group generated by the permutation matrices and the matrix $\overline{I_n}$ is isomorphic to the cartesian
  product $\SYM\times (\F_2,\oplus)$.
\end{prop}

The following lemma  gives a simple decomposition for the product of $[ij][ki][jk]$ by a transposition matrix $(ij)$, $(ki)$ or $(jk)$. We use it in the proof of
Proposition \ref{ibarsigmaprop}.


  \begin{lem} \label{ijk} Let $0\leq i,j,k\leq n-1$ be distinct integers : 
    \begin{align}
      &\text{R1 : }[ij][ki][jk]\mathbf{(jk)} = [ki][ij][jk]\label{R1}\\
      &\text{L1 : } \mathbf{(ij)}[ij][ki][jk] = [ij][jk][ki]\label{L1}\\
      &\text{R2 : }[ij][ki][jk]\mathbf{(ki)} = [jk][ij][ki][jk]\label{R2}\\
      &\text{L2 : }\mathbf{(ki)}[ij][ki][jk] = [ij][ki][jk][ij]\label{L2}\\
      &\text{R3 : }[ij][ki][jk]\mathbf{(ij)} = [ji][ik][kj]\label{R3}\\
      &\text{L3 : }\mathbf{(jk)}[ij][ki][jk] = [ji][ik][kj]\label{L3}
    \end{align}
\end{lem}

\begin{proof}We prove only \eqref{L3}, the other proofs being similar :
  \begin{align*}
    \mathbf{(jk)}[ij][ki][jk]&=[ij]^{(jk)}[ki]^{(jk)}\mathbf{(jk)}[jk] \\
                             &=[ik][ji]\mathbf{(jk)}[jk] \text{ (using \eqref{actiontsym})}\\
                             &=[ik][ji][jk][kj] \\
                             &= [ji][ik][kj] \text{ (using \eqref{redcom})}
  \end{align*}
  \end{proof}

  One obtains similar rules for matrices of type $[ij][jk][ki]$ by taking the inverse of each member of the identities of Lemma \ref{ijk}.
  For convenience we denote the product $[ij][ki][jk]$ by $[ijk]$. With this notation one has $[ij][jk][ki]=[kij]^{-1}$.

\begin{prop}\label{ibardecprop}
  Let $n\geq  2$ be even and let $n=2q$.

  Let  $B_q=\prod\limits_{i=0}^{q-2}[2i+1\quad 2i+2\quad 2i+3]$ if $q>1$ and $B_1=I_2$. Then
  \begin{equation}\overline{I_{2q}}=B_q^{-1}(01)B_q.\label{ibardec}\end{equation}
  \end{prop}

\begin{proof}
  We proove the result by induction on $q\geq  1$.
  
    The initial  case follows from $\overline{I_2}=(01)$.

  Induction step : suppose that $q\geq  1$ and $\overline{I_{2q}}=B_q^{-1}(01)B_q$.

  We use the morphism $\phi$ already mentionned in the final remark of Section \ref{algebra} : $\phi$ is the injective morphism from $\GL[2q]$ into $\GL[2(q+1)]$
  defined by
$\phi(M)=\begin{bmatrix}M&0\\0&I_2\end{bmatrix}$ and we make no distinction between matrices $M$ and $\phi(M)$.

So, one has :  
  $\phi(\overline{I_{2q}})=\begin{bmatrix}\overline{I_{2q}}&0\\0&I_2\end{bmatrix}=\begin{bmatrix}
    0&1&\dots&1&0&0\\
    1&0&\ddots&\vdots&\vdots&\vdots\\
    \vdots&\ddots&\ddots&1&\vdots&\vdots\\
    1&\dots&1&0&0&\vdots\\
    0&\dots&\dots&0&1&0\\
    0&\dots&\dots&\dots&0&1\\
  \end{bmatrix}$ and
  $\phi(\overline{I_{2q}})=B_q^{-1}(01)B_q$.

  Due to the effect on lines and columns of a mutliplication by a transvection matrix (Proposition \ref{GJmult}), we check that :

  $[2q\quad 2q+1][2q+1\quad 2q-1][2q-1\quad 2q]\phi(\overline{I_{2q}})[2q-1\quad 2q][2q+1\quad 2q-1][2q\quad 2q+1]=\overline{I_{2q+2}}$.

  Since $[2q-1\quad 2q][2q+1\quad 2q-1][2q\quad 2q+1]=[2q-1\quad 2q\quad 2q+1]$ one has :

  $\overline{I_{2q+2}}=[2q-1\quad 2q\quad 2q+1]^{-1}\phi(\overline{I_{2q}})[2q-1\quad 2q\quad 2q+1]$.

  $\overline{I_{2q+2}}=[2q-1\quad 2q\quad 2q+1]^{-1}B_q^{-1}(01)B_q[2q-1\quad 2q\quad 2q+1]$.

  Hence : $\overline{I_{2q+2}}=B_{q+1}^{-1}(01)B_{q+1}$.
\end{proof}

Formula \eqref{ibardec} gives a decomposition of $\overline{I_{2q}}$ in $3(q-1)+3+3(q-1)=3(n-1)$ transvections (where $n=2q$). We remark that this decompostion is not optimal if $q>1$.
Indeed, it can be reduced as follows.

Firstly 
$B_q=\prod\limits_{i=0}^{q-2}[2i+1\quad 2i+2\quad 2i+3]=[12][31][23]\prod\limits_{i=1}^{q-2}[2i+1\quad 2i+2\quad 2i+3]$.

Thus $B_q=[12][31][23]B'_q$, where $B_q'=\prod\limits_{i=1}^{q-2}[2i+1\quad 2i+2\quad 2i+3]$. Hence
\begin{equation}
\overline{I_{2q}}=B_q^{-1}(01)B_q=(B_q')^{-1}[23][31][12](01)[12][31][23]B'_q.\label{redBq}
\end{equation}

Then, one has :

$\begin{array}{ll}
[23][31][12](01)[12][31][23]&=[23][31][12][10][01][10][12][31][23]\\

                            &\stackrel{\ref{comT}}{=}[23][31][10][12][01][12][10][31][23]\\

                            &\stackrel{\ref{conjT}}{=}[23][31][10][01][02][10][31][23]. \text{ (8 transvections)}
 \end{array}$

 So, from Equation \eqref{redBq}, we now have a decomposition of $\overline{I_{2q}}$ in $3(q-2)+8+3(q-2)=3(n-1)-1$ transvections. We conjecture that this decomposition of $\overline{I_n}$ in $3(n-1)-1$ transvections is optimal and we checked this conjecture for $n=4$ using our computer optimization program.\medskip 

Actually Proposition \ref{ibardecprop} is a special case of the following proposition.

\begin{prop}\label{ibarsigmaprop} Let $n\geq  2$ be even and let $n=2q$. For any permutation matrix $\sigma$ of $\GL[n]$ different from $I_n$ there exists $q-1$ matrices $M_1,\dots,M_{q-1}$ of type $[xyz]$ , $q-1$ integers $\epsilon_i\in\{1,-1\}$, $q-1$ matrices $M_1',\dots,M_{q-1}'$ of type $[xyz]$ and $q-1$ integers $\epsilon_i'\in\{1,-1\}$ such that :  
  \begin{equation}\label{ibarsigma}
    \overline{ I_n}=\prod\limits_{i=1}^{q-1}M_i^{\varepsilon_i}\ \sigma\ \prod\limits_{i=1}^{q-1}M_i'^{\varepsilon_i'}
\end{equation}
    \end{prop}

    \begin{proof} Let $n\geq  2$ be an even integer. We denote by $\mathcal{P}_n^{*}$ the set of decreasing partitions $\lambda$ of $n$  such that
      $\lambda\neq (1,\dots,1)$. Any $\lambda\in\mathcal{P}_n^{*}$ is the cycle type of a permutation different from the identity (see Section \ref{background}).
      If $\lambda\in\mathcal{P}_n^{*}$, we denote by  $\alpha_{\lambda}$ the permutation of $\SYM$ (or equivalently the permutation matrix of $\GL$)
      defined by :      $\alpha_{\lambda}=\underbrace{(0\dots n_1-1)}_{\text{cycle of length } n_1}
    \underbrace{(n_1\dots n_1+n_2-1)}_{\text{cycle of length } n_2}\dots
    \underbrace{(\sum\limits_{i=1}^{p-1} n_i\dots\sum\limits_{i=1}^{p} n_i-1)}_{\text{cycle of length } n_p}$.
    
    Notice that $\alpha_{\lambda}$ has cycle type $\lambda$ and is therefore different from the identity.
    
    If $n\geq  4$ and $\lambda\in\mathcal{P}_n^{*}$, we denote by $\lambda_{-2}$ the element of $\mathcal{P}_{n-2}^{*}$ obtained from $\lambda$ as follows. To describe
    the operation we distinguish the general case ($n_p>2$) and three specific cases and we write each time the relation between
    $\alpha_{\lambda}$ and $\alpha_{\lambda_{-2}}$.
    Notice that $\alpha_{\lambda_{-2}}$ is a permutation in $\SYM[n-2]$ or a permutation matrix in $\GL[n-2]$ since $\lambda\in\mathcal{P}_{n-2}^{*}$ but we can consider
    $\alpha_{\lambda_{-2}}$ as a permutation in $\SYM[n]$ (by setting $\alpha_{\lambda_{-2}}(n-2)=n-2$ and $\alpha_{\lambda_{-2}}(n-1)=n-1$) or as a permutation matrix of $\GL[n]$    (using the injective morphism $\phi$ from $\GL[n-2]$ into $\GL$).

    In the general case, $\lambda_{-2}:=(n_1,\dots,n_{p}-2)$, so $\alpha_{\lambda_{-2}}=\alpha_{\lambda}(n-2\quad n-1)(n-3\quad n-2)$.
    If $n_p=n_{p-1}=1$ (first specific case), $\lambda_{-2}:=(n_1,\dots,n_{p-2})$, hence $\alpha_{\lambda_{-2}}=\alpha_{\lambda}$.
    If $n_p=1$ and $n_{p-1}>1$ (second specific case), $\lambda_{-2}:=(n_1,\dots,n_{p-2},n_{p-1}-1)$, hence $\alpha_{\lambda_{-2}}=\alpha_{\lambda}(n-3\quad n-2)$.
      If $n_p=2$ (third specific case), $\lambda_{-2}:=(n_1,\dots,n_{p-1})$,  hence  $\alpha_{\lambda_{-2}}=\alpha_{\lambda}(n-2\quad n-1)$.

     Without loss of generality we can prove Proposition \ref{ibarsigmaprop} in the case of a permutation $\sigma=\alpha_{\lambda}$
    where $\lambda\in\mathcal{P}_n^{*}$. Indeed, if $\sigma$ is any permutation (different from the identity) that has cycle type $\lambda$, then $\sigma$ is in the same conjugacy
    class as $\alpha_{\lambda}$ since they have the same cycle type and we build a permutation $\gamma$ such that $\sigma=\gamma\alpha_{\lambda}\gamma^{-1}$. Since 
$\overline{ I_n}=\prod\limits_{i=1}^{q-1}M_i^{\varepsilon_i}\ \alpha_{\lambda}\ \prod\limits_{i=1}^{q-1}M_i'^{\varepsilon_i'}$ and
$\overline{ I_n}\stackrel{\ref{ibarconj}}{=}\left(\overline{ I_n}\right)^{\gamma}$, it follows that
\begin{equation}\label{ibarproof1}
\overline{ I_n}=\prod\limits_{i=1}^{q-1}\left(M_i^{\varepsilon_i}\right)^{\gamma}\ \sigma\
\prod\limits_{i=1}^{q-1}\left(M_i'^{\varepsilon_i'}\right)^{\gamma}.
\end{equation}
To conclude it is sufficient to notice that $\left([ijk]^{\epsilon}\right)^{\gamma}=[\gamma(i) \gamma(j) \gamma(k)]^{\epsilon}$ for any permutation $\gamma$ and for $\epsilon\in\{-1,1\}$.
So Equation \eqref{ibarproof1} can easily be
writen under the form of Equation \eqref{ibarsigma} and this proves Proposition \ref{ibarsigmaprop} for the permutation $\sigma$.

We prove now  Proposition \ref{ibarsigmaprop} for a permutation $\alpha_{\lambda}$, by induction on $n\geq 2$ even.

The intial case is clear since there is only
one partition of $2$ different from $(1,1)$ namely $\lambda=(2)$. In this case $\alpha_{\lambda}=(01)=\overline{I_2}$.

      Induction step : let $n\geq  2$ be an even integer,  let $\lambda$ in $\mathcal{P}_{n+2}^{*}$ and $\alpha_{\lambda}$ in $\SYM[n+2]$.
      From Proposition \ref{ibardecprop}, one has $\overline{I_{n+2}}=[n-1\quad n\quad n+1]^{-1}\overline{I_n}[n-1\quad n\quad n+1]$.
      We use the induction hypothesis on $\alpha_{\lambda_{-2}}\in\SYM$ :
      $\overline{I_n}=\prod\limits_{i=1}^{q-1}M_i^{\varepsilon_i}\ \alpha_{\lambda_{-2}}\ \prod\limits_{i=1}^{q-1}M_i'^{\varepsilon_i'}$, hence
      \begin{equation}\label{ibarproof2}
        \overline{I_{n+2}}=[n-1\quad n\quad n+1]^{-1}\prod\limits_{i=1}^{q-1}M_i^{\varepsilon_i}\ \alpha_{\lambda_{-2}}\ \prod\limits_{i=1}^{q-1}M_i'^{\varepsilon_i'}[n-1\quad n\quad n+1].
      \end{equation}
      We consider now the different possible relations between $\alpha_{\lambda_{-2}}$ and $\alpha_{\lambda}$, starting by the three specific cases
      and ending by the general case which is more technical.
      
      In the first case $\alpha_{\lambda_{-2}}=\alpha_{\lambda}$, so Equation \eqref{ibarproof2} becomes
      
      $\overline{I_{n+2}}=[n-1\quad n\quad n+1]^{-1}\prod\limits_{i=1}^{q-1}M_i^{\varepsilon_i}\ \alpha_{\lambda}\ \prod\limits_{i=1}^{q-1}M_i'^{\varepsilon_i'}[n-1\quad n\quad n+1]$

      and we are done with the induction step.

      In the second case  $\alpha_{\lambda_{-2}}=\alpha_{\lambda}(n-1\quad n)$, so Equation \eqref{ibarproof2} becomes
      \begin{equation}\label{ibarproof3}
      \overline{I_{n+2}}=[n-1\quad n\quad n+1]^{-1}\prod\limits_{i=1}^{q-1}M_i^{\varepsilon_i}\ \alpha_{\lambda}(n-1\quad n)\ \prod\limits_{i=1}^{q-1}M_i'^{\varepsilon_i'}[n-1\quad n\quad n+1].
      \end{equation}
      Let $M=(n-1\quad n)\ \prod\limits_{i=1}^{q-1}M_i'^{\varepsilon_i'}[n-1\quad n\quad n+1]$. One has :

      $M=\prod\limits_{i=1}^{q-1}\left(M_i'^{\varepsilon_i'}\right)^{(n-1\ n)}(n-1\quad n)[n-1\quad n][n+1\quad n-1][n\quad n+1]$
                        
      $\phantom{M}\stackrel{\ref{L1}}{=}\prod\limits_{i=1}^{q-1}\left(M_i'^{\varepsilon_i'}\right)^{(n-1\ n)}[n-1\quad n][n\quad n+1][n+1\quad n-1]$

      $\phantom{M}=\prod\limits_{i=1}^{q-1}\left(M_i'^{\varepsilon_i'}\right)^{(n-1\  n)}[n+1\quad n-1\quad n]^{-1}$.

      Hence Equation \eqref{ibarproof3} becomes
\begin{equation*}
    \overline{I_{n+2}}=[n-1\quad n\quad n+1]^{-1}\prod\limits_{i=1}^{q-1}M_i^{\varepsilon_i}\ \alpha_{\lambda}\prod\limits_{i=1}^{q-1}\left(M_i'^{\varepsilon_i'}\right)^{(n-1\  n)}[n+1\quad n-1\quad n]^{-1}
\end{equation*}
    and we are done with the induction step since $([xyz]^{\epsilon})^{\sigma}=[\sigma(x)\sigma(y)\sigma(z)]^{\epsilon}$ where $\sigma=(n-1\quad n)$.

       In the third case  $\alpha_{\lambda_{-2}}=\alpha_{\lambda}(n\quad n+1)$, so Equation \eqref{ibarproof2} becomes
       \begin{equation}\label{ibarproof4}
        \overline{I_{n+2}}=[n-1\quad n\quad n+1]^{-1}\prod\limits_{i=1}^{q-1}M_i^{\varepsilon_i}\ \alpha_{\lambda}(n\quad n+1)\
        \prod\limits_{i=1}^{q-1}M_i'^{\varepsilon_i'}[n-1\quad n\quad n+1].
      \end{equation}
      Let $M=(n\quad n+1)\prod\limits_{i=1}^{q-1}M_i'^{\varepsilon_i'}[n-1\quad n\quad n+1]$. One has :
      
        $M=\prod\limits_{i=1}^{q-1}\left(M_i'^{\varepsilon_i'}\right)^{(n\ n+1)}(n\quad n+1)[n-1\quad n][n+1\quad n-1][n\quad n+1]$

        $\phantom{M}\stackrel{\ref{L3}}{=}\prod\limits_{i=1}^{q-1}\left(M_i'^{\varepsilon_i'}\right)^{(n\ n+1)}[n\quad n-1][n-1\quad n+1][n+1\quad n]$

        $\phantom{M}=\prod\limits_{i=1}^{q-1}\left(M_i'^{\varepsilon_i'}\right)^{(n\ n+1)}[n+1\quad n\quad n-1]^{-1}$.

        Hence Equation \eqref{ibarproof4} becomes
        \begin{equation*}
    \overline{I_{n+2}}=[n-1\quad n\quad n+1]^{-1}\prod\limits_{i=1}^{q-1}M_i^{\varepsilon_i}\ \alpha_{\lambda}\prod\limits_{i=1}^{q-1}\left(M_i'^{\varepsilon_i'}\right)^{(n\ n+1)}[n+1\quad n\quad n-1]^{-1}
  \end{equation*}
    and we are done with the induction step since $([xyz]^{\epsilon})^{\sigma}=[\sigma(x)\sigma(y)\sigma(z)]^{\epsilon}$ where $\sigma=(n\quad n+1)$.

    In the general case, we start by conjugating each member of Equation \eqref{ibarproof2} by $(n\quad n+1)$. Using the invariance of $\overline{I_n}$ by conjugation (Identity \eqref{ibarconj}), we get :

        \begin{equation}\label{ibarproof5}
\overline{I_{n+2}}=(n\quad n+1)[n-1\quad n\quad n+1]^{-1}\prod\limits_{i=1}^{q-1}M_i^{\varepsilon_i}\ \alpha_{\lambda_{-2}}
    \prod\limits_{i=1}^{q-1}M_i'^{\varepsilon_i'}[n-1\quad n\quad n+1](n\quad n+1).
    \end{equation}
    On one hand, we have :
    
    $(n\quad n+1)[n-1\quad n\quad n+1]^{-1}=(n\quad n+1)[n\quad n+1][n+1\quad n-1][[n-1\quad n]$
    
    $\phantom{(n\quad n+1)[n-1\quad n\quad n+1]^{-1}}\stackrel{\ref{L1}}{=}[n\quad n+1][[n-1\quad n][n+1\quad n-1]$

    $\phantom{(n\quad n+1)[n-1\quad n\quad n+1]^{-1}}=[n\quad n+1\quad n-1]$

    On the other hand we have :
    
    $[n-1\quad n\quad n+1](n\quad n+1)=[n-1\quad n][n+1\quad n-1][n\quad n+1](n\quad n+1)$.

    $\phantom{[n-1\quad n\quad n+1](n\quad n+1)}\stackrel{\ref{R1}}{=}[n+1\quad n-1][n-1\quad n][n\quad n+1]$.

    $\phantom{[n-1\quad n\quad n+1](n\quad n+1)}=[n\quad n+1\quad n-1]^{-1}$

    So Equation \eqref{ibarproof5} becomes :

$\overline{I_{n+2}}=[n\quad n+1\quad n-1]\prod\limits_{i=1}^{q-1}M_i^{\varepsilon_i}\ \alpha_{\lambda_{-2}}
    \prod\limits_{i=1}^{q-1}M_i'^{\varepsilon_i'}[n\quad n+1\quad n-1]^{-1}$

    In the general case, one has :

    $\alpha_{\lambda_{-2}}=\alpha_{\lambda}(n\quad n+1)(n-1\quad n)=\alpha_{\lambda}(n-1\quad n)(n+1\quad n-1)$, hence
\begin{equation}\label{ibarproof6}
\overline{I_{n+2}}=[n\quad n+1\quad n-1]\prod\limits_{i=1}^{q-1}M_i^{\varepsilon_i}\alpha_{\lambda}(n-1\quad n)(n+1\quad n-1)
\prod\limits_{i=1}^{q-1}M_i'^{\varepsilon_i'}[n\quad n+1\quad n-1]^{-1}.
\end{equation}
Let $M=(n-1\quad n)(n+1\quad n-1)
\prod\limits_{i=1}^{q-1}M_i'^{\varepsilon_i'}[n\quad n+1\quad n-1]^{-1}$. One has :

$M=\prod\limits_{i=1}^{q-1}\left(M_i'^{\varepsilon_i'}\right)^{(n-1\ n)(n+1\ n-1)}(n-1\quad n)(n+1\quad n-1)[n\quad n+1\quad n-1]^{-1}$

        $\phantom{M}\stackrel{\ref{L1}}{=}\prod\limits_{i=1}^{q-1}\left(M_i'^{\varepsilon_i'}\right)^{(n-1\ n)(n+1\ n-1)}(n-1\quad n)[n+1\quad n-1][n\quad n+1][n-1\quad n]$

        $\phantom{M}\stackrel{\ref{L3}}{=}\prod\limits_{i=1}^{q-1}\left(M_i'^{\varepsilon_i'}\right)^{(n-1\ n)(n+1\ n-1)}[n-1\quad n+1][n+1\quad n][n\quad n-1]$

        $\phantom{M}=\prod\limits_{i=1}^{q-1}\left(M_i'^{\varepsilon_i'}\right)^{(n-1\  n)(n+1\ n-1)}[n\quad n-1\quad n+1]^{-1}$.

        Hence Equation \eqref{ibarproof6} becomes :
        \begin{equation*}
        \overline{I_{n+2}}=[n\quad n+1\quad n-1]\prod\limits_{i=1}^{q-1}M_i^{\varepsilon_i}\alpha_{\lambda}
        \prod\limits_{i=1}^{q-1}\left(M_i'^{\varepsilon_i'}\right)^{(n-1\  n)(n+1\ n-1)}[n\quad n-1\quad n+1]^{-1}
      \end{equation*}
      and we are done with the induction step since $([xyz]^{\epsilon})^{\sigma}=[\sigma(x)\sigma(y)\sigma(z)]^{\epsilon}$ where $\sigma=(n-1\quad n)(n+1\quad n-1)$.

\end{proof}

Actually the main interest of Proposition \ref{ibarsigmaprop} is to give a way to compute easily a decomposition of $\overline{\sigma}$ in transvections,
as described in
the following proposition.

\begin{prop}\label{sigmabarprop}
Let $n\geq  2$ be an even integer. For any permutation matrix $\sigma$ of $\GL[n]$ different from $I_n$ there exists $n-2$ matrices $M_1,\dots,M_{n-2}$ of type $[xyz]$ and $n-2$ integers $\epsilon_i\in\{1,-1\}$ such that :
  \begin{equation}\label{sigmabar}
    \overline{\sigma}=\prod\limits_{i=1}^{n-2}M_i^{\varepsilon_i}
\end{equation}
\end{prop}

\begin{proof} Let $\sigma$ be a  permutation matrix of $\GL[n]$ different from $I_n$. 
  Applying Proposition \ref{ibarsigmaprop} to $\sigma^{-1}$ one has 
  $\overline{ I_n}=\prod\limits_{i=1}^{q-1}M_i^{\varepsilon_i}\ \sigma^{-1}\ \prod\limits_{i=1}^{q-1}M_i'^{\varepsilon_i'}$. Hence 
  $\overline{\sigma}=\sigma\overline{I_n}=\prod\limits_{i=1}^{q-1}\left(M_i^{\varepsilon_i}\right)^{\sigma}\prod\limits_{i=1}^{q-1}M_i'^{\varepsilon_i'}$.
  \end{proof}

  Since the proof of Proposition $\ref{ibarsigmaprop}$ is constructive, the method used in the proof of Proposition \ref{sigmabarprop} gives an algorithm to decompose any matrix $\overline{\sigma}$ different from $\overline{I_n}$ in $3(n-2)$ transvections (see Figure \ref{sigmabarex} for an example). We conjecture that this decomposition is optimal and we checked it for $n=4$.

  \begin{figure}
 Let $n=6$ and
 $\overline{\sigma}=\overline{(503)(142)}
 =\begin{bmatrix}
      1&1&1&1&1&0\\
      1&1&0&1&1&1\\
      1&1&1&1&0&1\\      
      0&1&1&1&1&1\\
      1&0&1&1&1&1\\
      1&1&1&0&1&1\\
    \end{bmatrix}$.
    
    $\sigma^{-1}=(305)(241)$. Let $\lambda_6=(3,3)$ be the cycle type of $\sigma^{-1}$, then $\alpha_{\lambda_6}=(012)(345)$.
    
    Let $\lambda_4=(3,1)$, then $\alpha_{\lambda_4}=(012)(3)=\alpha_{\lambda_6}(34)(53)$.

    Let $\lambda_2=(2)$, then $\alpha_{\lambda_2}=(01)=\alpha_{\lambda_4}(12)$.

    $\overline{I_2}=\alpha_{\lambda_2}$

    $\overline{I_4}=[123]^{-1}\overline{I_2}[123]=[123]^{-1}\alpha_{\lambda_2}[123]=[123]^{-1}\alpha_{\lambda_4}(12)[123]$

    $\overline{I_4}=[123]^{-1}\alpha_{\lambda_4}(12)[12][31][23]=[123]^{-1}\alpha_{\lambda_4}[12][23][31]=[123]^{-1}\alpha_{\lambda_4}[312]^{-1}$

    $\overline{I_6}=[345]^{-1}\overline{I_4}[345]=[345]^{-1}[123]^{-1}\alpha_{\lambda_4}[312]^{-1}[345]$

    $\overline{I_6}=[45][53][34][123]^{-1}\alpha_{\lambda_6}(34)(53)[312]^{-1}[34][53][45]$
    
    $\overline{I_6}=(45)[45][53][34][123]^{-1}\alpha_{\lambda_6}(34)(53)[312]^{-1}[34][53][45](45)$

    $\overline{I_6}=[45][34][53][123]^{-1}\alpha_{\lambda_6}(34)(53)[312]^{-1}[53][34][45]$

    $\overline{I_6}=[453][123]^{-1}\alpha_{\lambda_6}\left([312]^{-1}\right)^{(34)(53)}(34)(53)[53][34][45]$

    $\overline{I_6}=[453][123]^{-1}\alpha_{\lambda_6}[512]^{-1}(34)[53][45][34]$

    $\overline{I_6}=[453][123]^{-1}\alpha_{\lambda_6}[512]^{-1}[35][54][43]$

    $\overline{I_6}=[453][123]^{-1}\alpha_{\lambda_6}[512]^{-1}[435]^{-1}$

    $\sigma^{-1}=\alpha_{\lambda_6}^{\gamma}$
    where $\gamma=\begin{pmatrix}
      0&1&2&3&4&5\\
      3&0&5&2&4&1\end{pmatrix}$, hence :

    $\overline{I_6}=(\overline{I_6})^{\gamma}=([453][123]^{-1})^{\gamma}\sigma^{-1}([512]^{-1}[435]^{-1})^{\gamma}=
    [412][052]^{-1}\sigma^{-1}[105]^{-1}[421]^{-1}$

    $\overline{\sigma}=\sigma\overline{I_6}=([412][052]^{-1})^{\sigma}[105]^{-1}[421]^{-1}=[241][301]^{-1}[105]^{-1}[421]^{-1}$

     $\overline{\sigma}=[24][12][41][01][13][30][05][51][10][21][14][42]$
    
    { \caption{ Decomposition in transvections of a matrix of type $\overline{\sigma}$.\label{sigmabarex}}}
    \end{figure}
    \clearpage

    \section{Entanglement in $\cnot$ gates circuits}\label{entanglement}
    The notion of entanglement is usually defined through the group of Stochastic Local Operations assisted
    by Classical Communication denoted by $\SLOCC$ which is assimilated to the cartesian group product $\mathrm{SL}_2(\C)^{\times n}$ \cite{2004Miyake,2000DVC}. 
    Mathematically two states $\ket{\psi}$ and $\ket{\phi}$ are $\SLOCC$-equivalent if there exist $n$ operators $ A_{1}, \dots, A_{n}$ in $\mathrm{SL}_2(\C)$
    such that
    $A_{1}\otimes\cdots\otimes A_{n}\ket{\psi}=\lambda \ket{\phi}$ for some complex number $\lambda$. In other words, $\ket{\psi}$ and $\ket{\phi}$ are $\SLOCC$-equivalent if they are in the same orbit of $\mathrm{GL}_2(\C)^{\times n}$ acting on the Hilbert space $\C^{2}\otimes\cdots\otimes\C^{2}$. The $\SLOCC$-equivalence of two states $\ket{\phi}$ and $\ket{\psi}$ has a physical interpretation as explained in \cite{2000DVC} : $\ket{\phi}$ and $\ket{\psi}$ can be interconverted into each other with non zero probability by $n$ parties being able to coordinate their action by classical communication,  each party acting separately  on one of the qubits.
    The $\SLOCC$-equivalence of two states $\ket{\phi}$ and $\ket{\psi}$ implies that both states can achiveve the same tasks (for instance in a communication protocol) but with a probability of a success that may differ. In this sense the $\SLOCC$-equivalence can be considered as a qualitative
    way of separating non equivalent quantum states.    
\medskip    

    Some states have a particular interest like the Greenberger-Horne-Zeilinger state \cite{1990GHZ}
\begin{equation}
  \ket{\GHZ_{n}}=\frac1{\sqrt2}\left(\ket{\overbrace{0\cdots0}^{\times n}}+\ \ket{\overbrace{1\cdots1}^{\times n}}\right)=
  \frac1{\sqrt2}\left(\ket{0}^{\otimes n}+\ket{1}^{\otimes n}\right)\qquad(n\geq  3)
\end{equation}
and the $\W$-state \cite{2000DVC}
\begin{equation}
	\ket{\W_{n}}=\frac1{\sqrt n}\left(\ket{10\cdots0}+\ket{010\cdots0}+\cdots+\ket{0\cdots01}\right) \qquad(n\geq  3). 
\end{equation}
These two states represent two non-equivalent kind of entanglements : they belong to distinct $\SLOCC$ orbits and thus cannot be interconverted into each other by local operations
(even probalistically).
These states are particularly usefull because they are required as a physical ressource to realize many specific tasks.
For instance, in  an anonymous network where the processors
share the $\W$-state, the leader election problem can be solved by a simple protocol 
whereas the $\GHZ$-state is the only shared state
that allows solution of distributed consensus \cite{2004DP}.
The $\GHZ$-state is also used in many protocols in quantum cryptography, \emph{e.g.} in secret sharing \cite{1999Hillery}.
\medskip

We denote by $H_i$ the Hadamard gate (see Figure \ref{univers}) applied on qubit $i$. For instance in a 3-qubit system, $H_1=\mathtt{I}\otimes \mathtt{H}\otimes\mathtt{I}$ where $\mathtt{I}=\begin{bmatrix}1&0\\0&1\end{bmatrix}$ and $\mathtt{H}=\frac{1}{\sqrt{2}}\begin{bmatrix}1&1\\1&-1\end{bmatrix}$.
In the following we study the emergence of entanglement when a $\cnot$ circuit acts on a fully factorized state.

\subsection{Creating a $\GHZ$-state}
The construction of the $\GHZ$-state for $n$ qubit is straightforward :
\begin{equation}
\ket{\GHZ_n}=X_{n-1\ n-2}\dots X_{21}X_{10}H_0\ket{0}^{\otimes n}.\label{GHZ}
\end{equation}

A simple computation proves Equation \eqref{GHZ}. Indeed one has 

$H_0\ket{0\dots 0}=\frac{1}{\sqrt2}(\ket{0\dots 0}+\ket{10\dots 0})$ and
$X_{n-1\ n-2}\dots X_{21}X_{10}\ket{10\dots 0}=\ket{1\dots 1}$.\medskip

If $n=4$, Equation \eqref{GHZ} corresponds to the circuit in Figure \ref{GHZ4}.

\begin{figure}[h]
\begin{center}
\begin{tikzpicture}[scale=1.200000,x=1pt,y=1pt]
\filldraw[color=white] (0.000000, -7.500000) rectangle (78.000000, 52.500000);
\draw[color=black] (0.000000,45.000000) -- (78.000000,45.000000);
\draw[color=black] (0.000000,45.000000) node[left] {$\ket{0}$};
\draw[color=black] (0.000000,30.000000) -- (78.000000,30.000000);
\draw[color=black] (0.000000,30.000000) node[left] {$\ket{0}$};
\draw[color=black] (0.000000,15.000000) -- (78.000000,15.000000);
\draw[color=black] (0.000000,15.000000) node[left] {$\ket{0}$};
\draw[color=black] (0.000000,0.000000) -- (78.000000,0.000000);
\draw[color=black] (0.000000,0.000000) node[left] {$\ket{0}$};
\begin{scope}
\draw[fill=white] (12.000000, 45.000000) +(-45.000000:8.485281pt and 8.485281pt) -- +(45.000000:8.485281pt and 8.485281pt) -- +(135.000000:8.485281pt and 8.485281pt) -- +(225.000000:8.485281pt and 8.485281pt) -- cycle;
\clip (12.000000, 45.000000) +(-45.000000:8.485281pt and 8.485281pt) -- +(45.000000:8.485281pt and 8.485281pt) -- +(135.000000:8.485281pt and 8.485281pt) -- +(225.000000:8.485281pt and 8.485281pt) -- cycle;
\draw (12.000000, 45.000000) node {$H$};
\end{scope}
\draw (33.000000,45.000000) -- (33.000000,30.000000);
\begin{scope}
\draw[fill=white] (33.000000, 30.000000) circle(3.000000pt);
\clip (33.000000, 30.000000) circle(3.000000pt);
\draw (30.000000, 30.000000) -- (36.000000, 30.000000);
\draw (33.000000, 27.000000) -- (33.000000, 33.000000);
\end{scope}
\filldraw (33.000000, 45.000000) circle(1.500000pt);
\draw (51.000000,30.000000) -- (51.000000,15.000000);
\begin{scope}
\draw[fill=white] (51.000000, 15.000000) circle(3.000000pt);
\clip (51.000000, 15.000000) circle(3.000000pt);
\draw (48.000000, 15.000000) -- (54.000000, 15.000000);
\draw (51.000000, 12.000000) -- (51.000000, 18.000000);
\end{scope}
\filldraw (51.000000, 30.000000) circle(1.500000pt);
\draw (69.000000,15.000000) -- (69.000000,0.000000);
\begin{scope}
\draw[fill=white] (69.000000, 0.000000) circle(3.000000pt);
\clip (69.000000, 0.000000) circle(3.000000pt);
\draw (66.000000, 0.000000) -- (72.000000, 0.000000);
\draw (69.000000, -3.000000) -- (69.000000, 3.000000);
\end{scope}
\filldraw (69.000000, 15.000000) circle(1.500000pt);
\end{tikzpicture}
\raisebox{12mm}{\Large\ $\ket{\GHZ_4}$}
\end{center}
{ \caption{Using a circuit of $\XG[4]$ to obtain $\ket{\GHZ_4}$.\label{GHZ4}}}
\end{figure}
\subsection{The group $\XG[3]$ and entanglement of a 3-qubit system}
We prove that the  group $\XG[3]$  is powerful enough to generate any entanglement type from a completely factorized state. We use the method pioneered by Klyachko  in \cite{2002Klyachko}  wherein he promoted the use of Algebraic Theory of Invariant. The states of a $\SLOCC$-orbit are characterized by their values on covariant polynomials.
Let $\ket{\psi}=\sum\limits_{i,j,k\in\{0,1\}}\alpha_{ijk}\ket{ijk}$ be a 3-qubit state.
The simplest covariant associated to $\ket{\psi}$ is the trilinear form :
\begin{equation}
  A=\sum\limits_{i,j,k\in\{0,1\}}\alpha_{ijk}x_{i}y_{j}z_{k}.
\end{equation}
From $A$ one computes three quadratic forms :
\begin{equation}
	B_x(x_0,x_1)=\left|\begin{array}{cc}
  \frac{\partial^2A}{\partial y_0\partial z_0}&\frac{\partial^2A}{\partial y_0\partial z_1}\\
  \frac{\partial^2A}{\partial y_1\partial z_0}&\frac{\partial^2A}{\partial y_1\partial z_1}\\
\end{array}\right|,
\end{equation}
\begin{equation}
	B_y(y_0,y_1)=\left|\begin{array}{cc}
  \frac{\partial^2A}{\partial x_0\partial z_0}&\frac{\partial^2A}{\partial x_0\partial z_1}\\
  \frac{\partial^2A}{\partial x_1\partial z_0}&\frac{\partial^2A}{\partial x_1\partial z_1}\\
\end{array}\right|,
\end{equation}
\begin{equation}
	B_z(z_0,z_1)=\left|\begin{array}{cc}
  \frac{\partial^2A}{\partial x_0\partial y_0}&\frac{\partial^2A}{\partial x_0\partial y_1}\\
  \frac{\partial^2A}{\partial x_1\partial y_0}&\frac{\partial^2A}{\partial x_1\partial y_1}\\
\end{array}\right|.
\end{equation}
The catalectican is a trilinear form obtained by computing any of the three Jacobians
of A with one of the quadratic forms, which turns out to be the same, 
\begin{equation}
	C(x_0,x_1,y_0,y_1,z_0,z_1)=\left|\begin{array}{cc}
  \frac{\partial A}{\partial x_0}&\frac{\partial A}{\partial x_1}\\
  \frac{\partial B_x}{\partial x_0}&\frac{\partial B_x}{\partial x_1}\\
\end{array}\right|.
\label{cata}\end{equation}
The three quadratic forms $B_x$, $B_y$ and $B_z$ have the same discriminant $\Delta$ which is the last covariant polynomial we need
to characterize the $\SLOCC$-orbits :
\begin{equation}\begin{array}{rcl}
	\Delta(\ket{\psi})&=& \left(\alpha_{000} \alpha_{111} - \alpha_{001} \alpha_{110} - \alpha_{010} \alpha_{101} + \alpha_{011} \alpha_{100}\right)^{2}\\&&-4\left( \alpha_{000} \alpha_{011} -  \alpha_{001} \alpha_{010}\right) \left(\alpha_{100} \alpha_{111} - \alpha_{101} \alpha_{110}\right).\end{array}
    \label{delta3}\end{equation}

  The polynomial $\Delta$ is the generator of the algebra of invariant polynomials under the action of $\SLOCC$ (\textit{i.e} $\Delta(\ket{\psi})=\Delta(M\ket{\psi})$
  for any $M\in\mathrm{SL}_2(\C)^{\times 3}$). It is also the Cayley hyperdeterminant of the trilinear binary form $A$ \cite{1846Cayley}.
  
     Let $V:=[B_x, B_y,B_z,C,\Delta]$ be a vector of covariants. We associate the binary vector
      $V[\ket{\psi}] := [[B_x(\ket{\psi}], [B_y(\ket{\psi})], [B_z(\ket{\psi})], [C(\ket{\psi})], [\Delta(\ket{\psi})]]$ to any state $\ket{\psi}$,
    where $[P(\ket{\psi})] = 0$ if $P(\ket{\psi})= 0$ and $[P(\ket{\psi})] = 1$ if $P(\ket{\psi}) \neq 0$. The value of $V[\ket{\psi}]$ is sufficient to distinguish between the different
orbits (see \cite{2012HLT}). Results are summarized in Table \ref{orbits3q}.

\begin{table}[h]
  \begin{center}
    $\begin{array}{|c|c|c|}\hline
      \text{Orbit Symbols}&\text{Representatives } \ket{\psi}&V[\ket{\psi}]\\\hline
      \mathcal{O}_{VI}&\ket{\GHZ_3}&[1,1,1,1,1]\\
      \mathcal{O}_{V}&\ket{\W_3}&[1,1,1,1,0]\\
      \mathcal{O}_{IV}&\frac{1}{\sqrt2}(\ket{000}+\ket{110})&[0,0,1,0,0]\\
      \mathcal{O}_{III}&\frac{1}{\sqrt2}(\ket{000}+\ket{101})&[0,1,0,0,0]\\
      \mathcal{O}_{II}&\frac{1}{\sqrt2}(\ket{000}+\ket{011})&[1,0,0,0,0]\\
      \mathcal{O}_{I}&\ket{000}&[0,0,0,0,0]\\\hline
      \end{array}$
    \end{center}
    { \caption{The $\SLOCC$ orbits in a 3-qubit system.\label{orbits3q}}}
  \end{table}

  For Equation \eqref{GHZ} one has $\ket{\GHZ_3}=X_{21}X_{10}H_0\ket{000}$ (orbit $\mathcal{O}_{V}$) and it is easy to check that
  $\frac{1}{\sqrt2}(\ket{000}+\ket{011})=X_{21}H_1\ket{000}$ (orbit $\mathcal{O}_{II}$),
  $\frac{1}{\sqrt2}(\ket{000}+\ket{101})=X_{20}H_0\ket{000}$ (orbit $\mathcal{O}_{III}$), and
  $\frac{1}{\sqrt2}(\ket{000}+\ket{110})=X_{10}H_0\ket{000}$ (orbit $\mathcal{O}_{IV}$).
  
  In order to obtain a $\SLOCC$-equivalent to $\ket{\W_3}$ we see from Table \ref{orbits3q} that one has to construct a state $\psi$
  such that $\Delta(\ket{\psi})=0$ and $C(x_0,x_1,y_0,y_1,z_0,z_1)\neq 0$. This can be done using only gates of the standard set of universal gates (Figure \ref{univers})
as explained in Proposition \ref{W3} (see Figure \ref{W3circ} for an example of circuit).
  Note that the construction involves $\cnot$ circuits that have matrices of type $[ijk]$ (described in Subsection \ref{bitrev}). 

  \begin{prop}\label{W3}
    Let $X_{[ijk]}=X_{ij}X_{ki}X_{jk}$ where $i,j,k$ are distinct integers in $\{0,1,2\}$ and $k\neq 2$. The state
    \begin{equation}
X_{[ijk]}(\mathtt{T}\otimes\mathtt{T}\otimes\mathtt{S})\mathtt{H}^{\otimes 3}\ket{000}
      \end{equation}
is $\SLOCC$-equivalent to $\ket{\W_3}$.
  \end{prop}

  \begin{proof}
    Let $q=\ee^{\frac{\ii\pi}{4}}$ and let
    $\ket{\psi_{(k_0,k_1,\dots,k_7)}}=\frac{1}{\sqrt8}\left(q^{k_0}\ket{000}+q^{k_1}\ket{001}+\dots+q^{k_7}\ket{111}\right)$
    where $k_i$ is an integer. A simple calculation shows that $(\mathtt{T}\otimes\mathtt{T}\otimes\mathtt{S})\mathtt{H}^{\otimes 3}\ket{000}=\ket{\psi_{(0,2,1,3,1,3,2,4)}}$.
    Then we compute the values of polynomials $\Delta$ (Formula \eqref{delta3}) and $C$ (Formula \eqref{cata}) on the state $\ket{\psi'}=X_{[ijk]}\ket{\psi_{(0,2,1,3,1,3,2,4)}}$
    for $i,j,k\in\{0,1,2\}$. When $k\neq 2$ we find that $\Delta=0$ and $C\neq 0$ (see results in the table below). \medskip
    
  {\footnotesize
  $\begin{array}{|c|c|c|c|}\hline
     [ijk]&X_{[xyz]}\ket{\psi_{(0,2,1,3,1,3,2,4)}}&C(x_0,x_1,y_0,y_1,z_0,z_1)\\\hline
     [021]&\ket{\psi_{(0,2,3,3,4,2,1,1)}}&\frac18\left((1+\ii)x_0y_0z_0+(1+\ii) x_0y_0z_1+(1-\ii) x_1y_0z_0+(1-\ii) x_1y_0z_1\right)\\\hline
     [120]&\ket{\psi_{(0,2,1,3,1,3,2,4)}}&\frac18\left((1+\ii)x_0y_0z_0+(1+\ii) x_0y_0z_1+(1-\ii) x_0y_1z_0+(1-\ii) x_0y_1z_1\right)\\\hline
     [201]&\ket{\psi_{(0,4,2,2,3,1,3,1)}}&\frac18\left((1+\ii)x_0y_0z_0+(1-\ii) x_0y_0z_1+(1+\ii) x_0y_1z_0+(1-\ii) x_0y_1z_1\right)\\\hline
     [210]&\ket{\psi_{(0,4,3,1,2,2,3,1)}}&\frac18\left((1+\ii)x_0y_0z_0+(1-\ii) x_0y_0z_1+(1+\ii) x_1y_0z_0+(1-\ii) x_1y_0z_1\right)\\\hline
   \end{array}$}
 
\end{proof}

Computing $\Delta$ and $C$ for the state $X_{[ijk]}(\mathtt{T}\otimes\mathtt{T}\otimes\mathtt{S})\mathtt{H}^{\otimes 3}\ket{000}$ when $k=2$, one finds $\Delta\neq 0$ and $C\neq 0$, thus the resulting state is $\SLOCC$-equivalent to $\ket{\GHZ_3}$.\medskip

As the dimension of the Hilbert space $\HH^{\otimes3}$ is small, namely $2^3$, one can compute three 2-dimensional matrices $A,B,C$ of determinant 1  and a complex number $k$ such that $A\otimes B\otimes C\ket{\psi_{(0,2,3,3,4,2,1,1)}}=k\ket{\W_3}$. This can be done by solving an eight algebraic equations system. We check that the values
$A=\begin{bmatrix}3\ii&1\\\frac{1}{2}&-\frac{1}{2}\ii\end{bmatrix}$, $B=2^{\frac14}\begin{bmatrix}-\ii\sqrt{2}&2\\0&\frac{1}{2}\ii\end{bmatrix}$,
$C=\begin{bmatrix}\ii&\ii\\\frac{1}{2}\ii&-\frac{1}{2}\ii\end{bmatrix}$ and $k=2^{\frac14}\frac{\sqrt3}{\sqrt2}\ee^{\frac{\ii\pi}{4}}$ are solutions.

\begin{figure}[h]
  \begin{center}
\begin{tikzpicture}[scale=1.200000,x=1pt,y=1pt]
\filldraw[color=white] (0.000000, -7.500000) rectangle (102.000000, 37.500000);
\draw[color=black] (0.000000,30.000000) -- (102.000000,30.000000);
\draw[color=black] (0.000000,30.000000) node[left] {$\ket{0}$};
\draw[color=black] (0.000000,15.000000) -- (102.000000,15.000000);
\draw[color=black] (0.000000,15.000000) node[left] {$\ket{0}$};
\draw[color=black] (0.000000,0.000000) -- (102.000000,0.000000);
\draw[color=black] (0.000000,0.000000) node[left] {$\ket{0}$};
\begin{scope}
\draw[fill=white] (12.000000, 30.000000) +(-45.000000:8.485281pt and 8.485281pt) -- +(45.000000:8.485281pt and 8.485281pt) -- +(135.000000:8.485281pt and 8.485281pt) -- +(225.000000:8.485281pt and 8.485281pt) -- cycle;
\clip (12.000000, 30.000000) +(-45.000000:8.485281pt and 8.485281pt) -- +(45.000000:8.485281pt and 8.485281pt) -- +(135.000000:8.485281pt and 8.485281pt) -- +(225.000000:8.485281pt and 8.485281pt) -- cycle;
\draw (12.000000, 30.000000) node {$H$};
\end{scope}
\begin{scope}
\draw[fill=white] (12.000000, 15.000000) +(-45.000000:8.485281pt and 8.485281pt) -- +(45.000000:8.485281pt and 8.485281pt) -- +(135.000000:8.485281pt and 8.485281pt) -- +(225.000000:8.485281pt and 8.485281pt) -- cycle;
\clip (12.000000, 15.000000) +(-45.000000:8.485281pt and 8.485281pt) -- +(45.000000:8.485281pt and 8.485281pt) -- +(135.000000:8.485281pt and 8.485281pt) -- +(225.000000:8.485281pt and 8.485281pt) -- cycle;
\draw (12.000000, 15.000000) node {$H$};
\end{scope}
\begin{scope}
\draw[fill=white] (12.000000, -0.000000) +(-45.000000:8.485281pt and 8.485281pt) -- +(45.000000:8.485281pt and 8.485281pt) -- +(135.000000:8.485281pt and 8.485281pt) -- +(225.000000:8.485281pt and 8.485281pt) -- cycle;
\clip (12.000000, -0.000000) +(-45.000000:8.485281pt and 8.485281pt) -- +(45.000000:8.485281pt and 8.485281pt) -- +(135.000000:8.485281pt and 8.485281pt) -- +(225.000000:8.485281pt and 8.485281pt) -- cycle;
\draw (12.000000, -0.000000) node {$H$};
\end{scope}
\begin{scope}
\draw[fill=white] (36.000000, 30.000000) +(-45.000000:8.485281pt and 8.485281pt) -- +(45.000000:8.485281pt and 8.485281pt) -- +(135.000000:8.485281pt and 8.485281pt) -- +(225.000000:8.485281pt and 8.485281pt) -- cycle;
\clip (36.000000, 30.000000) +(-45.000000:8.485281pt and 8.485281pt) -- +(45.000000:8.485281pt and 8.485281pt) -- +(135.000000:8.485281pt and 8.485281pt) -- +(225.000000:8.485281pt and 8.485281pt) -- cycle;
\draw (36.000000, 30.000000) node {$T$};
\end{scope}
\begin{scope}
\draw[fill=white] (36.000000, 15.000000) +(-45.000000:8.485281pt and 8.485281pt) -- +(45.000000:8.485281pt and 8.485281pt) -- +(135.000000:8.485281pt and 8.485281pt) -- +(225.000000:8.485281pt and 8.485281pt) -- cycle;
\clip (36.000000, 15.000000) +(-45.000000:8.485281pt and 8.485281pt) -- +(45.000000:8.485281pt and 8.485281pt) -- +(135.000000:8.485281pt and 8.485281pt) -- +(225.000000:8.485281pt and 8.485281pt) -- cycle;
\draw (36.000000, 15.000000) node {$T$};
\end{scope}
\begin{scope}
\draw[fill=white] (36.000000, -0.000000) +(-45.000000:8.485281pt and 8.485281pt) -- +(45.000000:8.485281pt and 8.485281pt) -- +(135.000000:8.485281pt and 8.485281pt) -- +(225.000000:8.485281pt and 8.485281pt) -- cycle;
\clip (36.000000, -0.000000) +(-45.000000:8.485281pt and 8.485281pt) -- +(45.000000:8.485281pt and 8.485281pt) -- +(135.000000:8.485281pt and 8.485281pt) -- +(225.000000:8.485281pt and 8.485281pt) -- cycle;
\draw (36.000000, -0.000000) node {$S$};
\end{scope}
\draw (57.000000,15.000000) -- (57.000000,0.000000);
\begin{scope}
\draw[fill=white] (57.000000, 0.000000) circle(3.000000pt);
\clip (57.000000, 0.000000) circle(3.000000pt);
\draw (54.000000, 0.000000) -- (60.000000, 0.000000);
\draw (57.000000, -3.000000) -- (57.000000, 3.000000);
\end{scope}
\filldraw (57.000000, 15.000000) circle(1.500000pt);
\draw (75.000000,30.000000) -- (75.000000,15.000000);
\begin{scope}
\draw[fill=white] (75.000000, 15.000000) circle(3.000000pt);
\clip (75.000000, 15.000000) circle(3.000000pt);
\draw (72.000000, 15.000000) -- (78.000000, 15.000000);
\draw (75.000000, 12.000000) -- (75.000000, 18.000000);
\end{scope}
\filldraw (75.000000, 30.000000) circle(1.500000pt);
\draw (93.000000,30.000000) -- (93.000000,0.000000);
\begin{scope}
\draw[fill=white] (93.000000, 30.000000) circle(3.000000pt);
\clip (93.000000, 30.000000) circle(3.000000pt);
\draw (90.000000, 30.000000) -- (96.000000, 30.000000);
\draw (93.000000, 27.000000) -- (93.000000, 33.000000);
\end{scope}
\filldraw (93.000000, 0.000000) circle(1.500000pt);
\end{tikzpicture} \raisebox{8mm}{\Large$\sim\ket{\W_3}$}
\end{center}

Output state is :

$\ket{\psi_{(0,2,3,3,4,2,1,1)}}=X_{[021]}\ket{\psi_{(0,2,1,3,1,3,2,4)}}=X_{02}X_{10}X_{21}(\mathtt{T}\otimes\mathtt{T}\otimes\mathtt{S})\mathtt{H}^{\otimes 3}\ket{000}$.\medskip

\caption{A circuit of $\XG[3]$ producing a $\SLOCC$-equivalent to $\ket{\W_3}$.\label{W3circ}}
  \end{figure}

  \subsection{The group $\XG[4]$ and entanglement of a 4-qubit system\label{fourqubits}}
The situation of 4-qubits systems is more complex than for 3-qubits systems.
The corresponding Hilbert space $\HH^{\otimes 4}$ has infinitely many orbits under the action
of $\SLOCC=\mathrm{SL}_2(\C)^{\times 4}$. These orbits have been classified by Verstraete \textit{et al.} \cite{2002VDDV} into 9 families (6 families are described with
parameters). Among these 9 families only one is generic : any state in the more general situation belongs to the family

\begin{equation}\begin{array}{rcl}
	G_{abcd}&=&{a+d\over 2}\left(|0000\rangle + |1111\rangle\right)+{a-d\over 2}\left(|0011\rangle+|1100\rangle\right){}
	\\&&+{b+c\over 2}\left(|0101\rangle+|1010\rangle\right)+{b-c\over 2}\left(|0110\rangle+|1001\rangle\right),\end{array}
    \end{equation}

for  independent parameters $a,b, c,$ and $d$ \cite{2002VDDV,2014HLT}.

More precisely, a generic state of 4 qubits is $\SLOCC$-equivalent, up to permutations of the qubits,  to 192 Verstraete states of the  $G_{abcd}$ family \cite{2017HLT}.

To determine the Verstraete family to which a given state belongs, one can use an algorithm described in a previous paper \cite{2017HLT}.
As is the case of 3 qubits, this algorithm is based on the evaluation of some covariants polynomials (see Appendix \ref{cov4}). We do not recall the algorithm since it is not usefull to understand the present paper.\medskip

Let $\ket{\psi}=\sum\limits_{i,j,k,\ell\in\{0,1\}}\alpha_{ijk\ell}\ket{ijk\ell}$ be a 4-qubit state. The algebra of $\SLOCC$-invariant polynomials is freely generated by the four following polynomials \cite{2003LT}:

\begin{itemize}
	\item The smallest degree invariant
	\begin{equation}{}
	B:=\sum_{0\leq i_{1},i_{2},i_{3}\leq 1}(-1)^{i_{1}+i_{2}+i_{3}}\alpha_{0i_{1}i_{2}i_{3}}\alpha_{1(1-i_{1})(1-i_{2})(1-i_{3})},
	\end{equation}
	\item Two polynomials of degree $4$
	\begin{equation}
		L:=\left|\begin{array}{cccc}\alpha_{0000}&\alpha_{0010}&\alpha_{0001}&\alpha_{0011}\\
		\alpha_{1000}&\alpha_{1010}&\alpha_{1001}&\alpha_{1011}\\
		\alpha_{0100}&\alpha_{0110}&\alpha_{0101}&\alpha_{0111}\\
		\alpha_{1100}&\alpha_{1110}&\alpha_{1101}&\alpha_{1111}\end{array}
		\right|\end{equation}  and 
		\begin{equation} 
		M:=\left|\begin{array}{cccc} \alpha_{0000}&\alpha_{0001}&\alpha_{0100}&\alpha_{0101}\\
		 \alpha_{1000}&\alpha_{1001}&\alpha_{1100}&\alpha_{1101}\\
		  \alpha_{0010}&\alpha_{0011}&\alpha_{0110}&\alpha_{0111}\\
		   \alpha_{1010}&\alpha_{1011}&\alpha_{1110}&\alpha_{1111}\end{array}\right|.
	\end{equation}
	\item and a polynomial of degree $6$ defined by $D_{xy}=-\det(B_{xy})$ where $B_{xy}$ is the $3\times 3$ matrix satisfying
	\begin{equation}
		\left[x_{0}^{2},x_{0}x_{1},x_{1}^{2}\right]B_{xy}\left[\begin{array}{c} y_{0}^{2}\\y_{0}y_{1}\\y_{1}^{2}\end{array}\right]=\det\left({\partial^{2}\over\partial z_{i}\partial t_{j}}A\right)
	\end{equation}
\noindent with $A=\sum\limits_{i,j,k,\ell\in\{0,1\}}\alpha_{ijk\ell}x_{i}y_{j}z_{k}t_{\ell}$ being the quadrilinear binary form associated to the state $\ket{\psi}$.
\end{itemize}

  Using the generators $B,L,M$ and $D_{xy}$, one can build $\Delta$, an invariant polynomial of degree 24 that plays an important role in the quantitative and qualitative study of entanglement \cite{2002MW,2003Miyake}.
  As described in \cite{2014HLT} and \cite{2003LT}, $\Delta$ is the discriminant of any of the quartics
\begin{equation}
	Q_{1}=x^{4}-2Bx^{3}y+(B^{2}+2L+4M)x^{2}y^{2}+4(D_{xy}-B(M+\frac12L))xy^{3}+L^{2}y^{4},
\end{equation}\vspace{-2mm}
\begin{equation}
	Q_{2}=x^{4}-2Bx^{3}y+(B^{2}-4L-2M)x^{2}y^{2}+(4D_{xy}-2MB)xy^{3}+M^{2}y^{4},
\end{equation}
and
\begin{equation}
	Q_{3}=x^{4}-2Bx^{3y}+(B^{2}+2L-2M)x^{2}y^{2}-(2(L+M)B-4D_{xy})xy^{3}+N^{2}y^{4}.
\end{equation}
We recall that, for a quartic $Q=\alpha x^{4}-4\beta x^{3}y+6\gamma x^{2}y^{2}-4\delta xy^{3}+\omega y^{4}$, the discriminant can be computed as
\begin{equation}
  \Delta=I_{2}^{3}-27I_{3}^{2}\label{delta4}
\end{equation}
where  $I_{2}=\alpha\omega-4\beta\delta+3\gamma^{2}$ and 
$I_{3}=\alpha\gamma\omega-\alpha\delta^{2}-\omega\beta^{2}-\gamma^{3}+2\beta\gamma\delta$.

  It appears that $\Delta$ is the Cayley hyperdeterminant (in the sense of Gelfand \textit{et al.} \cite{1992GKL}) of $A$ \cite{2003LT}. 
  In \cite{2003Miyake}, Miyake showed that the more generic entanglement holds only for the states $\ket{\psi}$ such that $\Delta(\ket{\psi})\neq 0$. Moreover any generically entangled state is equivalent to a state of the  Verstraete $G_{abcd}$ family \cite[Appendix~A]{2003LT}. So one can consider $\Delta$ as a qualitative measure of entanglement :  an entangled state (\textit{i.e.} not factorized state) is generically entangled if $\Delta\neq 0$.

  In the case of 3 qubits, $\Delta(\ket{\GHZ_3})\neq 0$ (see Table \ref{orbits3q}). Using \eqref{delta3} one checks easily that $\Delta(\ket{\GHZ_3})= \frac14$. So the state $\ket{\GHZ_3}$ is generically entangled. In the 4 qubits case one computes $\Delta(\ket{\GHZ_4})=0$ using \eqref{delta4}. Hence, suprisingly, $\ket{\GHZ_4}$ is not generically entangled and this result can be generalized : in \cite{MBJGL2019} Appendix C we proved that,
  for any $k>3$, $\Delta(\ket{\GHZ_k})=0$.
  
In this context we ask ourselves if it is possible to find a 4-qubit $\cnot$  circuit that takes as input a completely factorized state and output a generically entangled state. The following statement answers the question.

\begin{theo}\label{generic4} The state
  $\ket{BL}:=X_{01}X_{20}X_{03}X_{10}(\mathtt{T}\otimes\mathtt{S}\otimes\mathtt{S}\otimes\mathtt{S})\mathtt{H}^{\otimes 4}\ket{0000}$ is generically entangled. It is produced by the circuit :
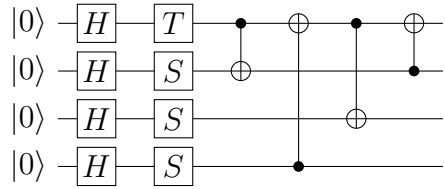
\begin{center}
  \begin{tikzpicture}[scale=1.200000,x=1pt,y=1pt]
\filldraw[color=white] (0.000000, -7.500000) rectangle (120.000000, 52.500000);
\draw[color=black] (0.000000,45.000000) -- (120.000000,45.000000);
\draw[color=black] (0.000000,45.000000) node[left] {$\ket{0}$};
\draw[color=black] (0.000000,30.000000) -- (120.000000,30.000000);
\draw[color=black] (0.000000,30.000000) node[left] {$\ket{0}$};
\draw[color=black] (0.000000,15.000000) -- (120.000000,15.000000);
\draw[color=black] (0.000000,15.000000) node[left] {$\ket{0}$};
\draw[color=black] (0.000000,0.000000) -- (120.000000,0.000000);
\draw[color=black] (0.000000,0.000000) node[left] {$\ket{0}$};
\begin{scope}
\draw[fill=white] (12.000000, 45.000000) +(-45.000000:8.485281pt and 8.485281pt) -- +(45.000000:8.485281pt and 8.485281pt) -- +(135.000000:8.485281pt and 8.485281pt) -- +(225.000000:8.485281pt and 8.485281pt) -- cycle;
\clip (12.000000, 45.000000) +(-45.000000:8.485281pt and 8.485281pt) -- +(45.000000:8.485281pt and 8.485281pt) -- +(135.000000:8.485281pt and 8.485281pt) -- +(225.000000:8.485281pt and 8.485281pt) -- cycle;
\draw (12.000000, 45.000000) node {$H$};
\end{scope}
\begin{scope}
\draw[fill=white] (12.000000, 30.000000) +(-45.000000:8.485281pt and 8.485281pt) -- +(45.000000:8.485281pt and 8.485281pt) -- +(135.000000:8.485281pt and 8.485281pt) -- +(225.000000:8.485281pt and 8.485281pt) -- cycle;
\clip (12.000000, 30.000000) +(-45.000000:8.485281pt and 8.485281pt) -- +(45.000000:8.485281pt and 8.485281pt) -- +(135.000000:8.485281pt and 8.485281pt) -- +(225.000000:8.485281pt and 8.485281pt) -- cycle;
\draw (12.000000, 30.000000) node {$H$};
\end{scope}
\begin{scope}
\draw[fill=white] (12.000000, 15.000000) +(-45.000000:8.485281pt and 8.485281pt) -- +(45.000000:8.485281pt and 8.485281pt) -- +(135.000000:8.485281pt and 8.485281pt) -- +(225.000000:8.485281pt and 8.485281pt) -- cycle;
\clip (12.000000, 15.000000) +(-45.000000:8.485281pt and 8.485281pt) -- +(45.000000:8.485281pt and 8.485281pt) -- +(135.000000:8.485281pt and 8.485281pt) -- +(225.000000:8.485281pt and 8.485281pt) -- cycle;
\draw (12.000000, 15.000000) node {$H$};
\end{scope}
\begin{scope}
\draw[fill=white] (12.000000, -0.000000) +(-45.000000:8.485281pt and 8.485281pt) -- +(45.000000:8.485281pt and 8.485281pt) -- +(135.000000:8.485281pt and 8.485281pt) -- +(225.000000:8.485281pt and 8.485281pt) -- cycle;
\clip (12.000000, -0.000000) +(-45.000000:8.485281pt and 8.485281pt) -- +(45.000000:8.485281pt and 8.485281pt) -- +(135.000000:8.485281pt and 8.485281pt) -- +(225.000000:8.485281pt and 8.485281pt) -- cycle;
\draw (12.000000, -0.000000) node {$H$};
\end{scope}
\begin{scope}
\draw[fill=white] (36.000000, 45.000000) +(-45.000000:8.485281pt and 8.485281pt) -- +(45.000000:8.485281pt and 8.485281pt) -- +(135.000000:8.485281pt and 8.485281pt) -- +(225.000000:8.485281pt and 8.485281pt) -- cycle;
\clip (36.000000, 45.000000) +(-45.000000:8.485281pt and 8.485281pt) -- +(45.000000:8.485281pt and 8.485281pt) -- +(135.000000:8.485281pt and 8.485281pt) -- +(225.000000:8.485281pt and 8.485281pt) -- cycle;
\draw (36.000000, 45.000000) node {$T$};
\end{scope}
\begin{scope}
\draw[fill=white] (36.000000, 30.000000) +(-45.000000:8.485281pt and 8.485281pt) -- +(45.000000:8.485281pt and 8.485281pt) -- +(135.000000:8.485281pt and 8.485281pt) -- +(225.000000:8.485281pt and 8.485281pt) -- cycle;
\clip (36.000000, 30.000000) +(-45.000000:8.485281pt and 8.485281pt) -- +(45.000000:8.485281pt and 8.485281pt) -- +(135.000000:8.485281pt and 8.485281pt) -- +(225.000000:8.485281pt and 8.485281pt) -- cycle;
\draw (36.000000, 30.000000) node {$S$};
\end{scope}
\begin{scope}
\draw[fill=white] (36.000000, 15.000000) +(-45.000000:8.485281pt and 8.485281pt) -- +(45.000000:8.485281pt and 8.485281pt) -- +(135.000000:8.485281pt and 8.485281pt) -- +(225.000000:8.485281pt and 8.485281pt) -- cycle;
\clip (36.000000, 15.000000) +(-45.000000:8.485281pt and 8.485281pt) -- +(45.000000:8.485281pt and 8.485281pt) -- +(135.000000:8.485281pt and 8.485281pt) -- +(225.000000:8.485281pt and 8.485281pt) -- cycle;
\draw (36.000000, 15.000000) node {$S$};
\end{scope}
\begin{scope}
\draw[fill=white] (36.000000, -0.000000) +(-45.000000:8.485281pt and 8.485281pt) -- +(45.000000:8.485281pt and 8.485281pt) -- +(135.000000:8.485281pt and 8.485281pt) -- +(225.000000:8.485281pt and 8.485281pt) -- cycle;
\clip (36.000000, -0.000000) +(-45.000000:8.485281pt and 8.485281pt) -- +(45.000000:8.485281pt and 8.485281pt) -- +(135.000000:8.485281pt and 8.485281pt) -- +(225.000000:8.485281pt and 8.485281pt) -- cycle;
\draw (36.000000, -0.000000) node {$S$};
\end{scope}
\draw (57.000000,45.000000) -- (57.000000,30.000000);
\begin{scope}
\draw[fill=white] (57.000000, 30.000000) circle(3.000000pt);
\clip (57.000000, 30.000000) circle(3.000000pt);
\draw (54.000000, 30.000000) -- (60.000000, 30.000000);
\draw (57.000000, 27.000000) -- (57.000000, 33.000000);
\end{scope}
\filldraw (57.000000, 45.000000) circle(1.500000pt);
\draw (75.000000,45.000000) -- (75.000000,0.000000);
\begin{scope}
\draw[fill=white] (75.000000, 45.000000) circle(3.000000pt);
\clip (75.000000, 45.000000) circle(3.000000pt);
\draw (72.000000, 45.000000) -- (78.000000, 45.000000);
\draw (75.000000, 42.000000) -- (75.000000, 48.000000);
\end{scope}
\filldraw (75.000000, 0.000000) circle(1.500000pt);
\draw (93.000000,45.000000) -- (93.000000,15.000000);
\begin{scope}
\draw[fill=white] (93.000000, 15.000000) circle(3.000000pt);
\clip (93.000000, 15.000000) circle(3.000000pt);
\draw (90.000000, 15.000000) -- (96.000000, 15.000000);
\draw (93.000000, 12.000000) -- (93.000000, 18.000000);
\end{scope}
\filldraw (93.000000, 45.000000) circle(1.500000pt);
\draw (111.000000,45.000000) -- (111.000000,30.000000);
\begin{scope}
\draw[fill=white] (111.000000, 45.000000) circle(3.000000pt);
\clip (111.000000, 45.000000) circle(3.000000pt);
\draw (108.000000, 45.000000) -- (114.000000, 45.000000);
\draw (111.000000, 42.000000) -- (111.000000, 48.000000);
\end{scope}
\filldraw (111.000000, 30.000000) circle(1.500000pt);
\end{tikzpicture}
\end{center}
  \end{theo}

  \begin{proof}
Using formula \eqref{delta4} one computes $\Delta(\ket{BL})=-\dfrac{1}{2^{24}}\simeq -5,96\times 10^{-8}$.
    \end{proof}
    According to Miyake \cite{2002MW}, the value of $|\Delta|$ can be considered as a measure of entanglement. So a state with the most amount
of generic entanglement can be defined as a state that maximize $|\Delta|$. 
    In the 4-qubits case the maximal value of $|\Delta|$ is $\frac{1}{2^83^9}\simeq 1,98\times 10^{-7}$ \cite{2013CD} and the mean is around $1.32\times10^{-9}$ (value based on $10000$ random states \cite{2017Alsina}).
    So the amount of generic entanglement in the state $\ket{BL}$ is not maximal, although much higher than the mean.
    
    A few states are known to maximize $|\Delta|$ : $\ket{L}$
    \cite{2000HS,2012GW,2013CD}, $\ket{HD}$ \cite{2017Alsina} and   $\ket{M_{2222}}$ (Jaffali, phd thesis).
    Figure \ref{delta_max} gives the definition of these states. It seems interesting to know if it is possible, starting from a fully factorized state, to reach one of these 3 states by a 4-qubit $\cnot$ circuit. To answer this question we proceed as follows. We start from a fully factorized state
\begin{equation}
  \ket{F(u)}:=(a_0\ket{0}+a_1\ket{1})\otimes(b_0\ket{0}+b_1\ket{1})\otimes(c_0\ket{0}+c_1\ket{1})\otimes(d_0\ket{0}+d_1\ket{1})
\end{equation}
where $u:=[a_0,a_1,b_0,b_1,c_0,c_1,d_0,d_1]$ is a vector of  complex numbers. Then, for each circuit $C$ in $\XG[4]$, we compute $C\ket{F(u)}$ and we check that the equation
$C\ket{F(u)}=\ket{\psi}$ where $\ket{\psi}\in\{ \ket{L},\ket{HD},\ket{M_{2222}}\}$ has no solution. This can be done in a few seconds using Maple 2020 (X86 64 LINUX). The corresponding script can be downloaded at
\href{https://github.com/marcbataille/cnot-circuits/blob/master/entanglement/search_state.mpl}{\texttt{https://github.com/marcbataille/cnot-circuits}}. Hence the answer to our question is negative.

   Note that the question whether a state different from $\ket{L}$, $\ket{HD}$, or $\ket{M_{2222}}$ that has maximal generic entanglement can be produced under the same constraints
    (\textit{i.e.} factorized state plus 4-qubit $\cnot$ circuit) remains open, however. 

    \begin{figure}[h]
     \begin{equation}
        \ket{L}=\frac{1}{\sqrt3}(\ket{u_0}+\omega\ket{u_1}+\omega^2\ket{u_2}
      \end{equation}
      with: $\ket{u_0}=\frac{1}{2}(\ket{0000} + \ket{0011} + \ket{1100} + \ket{1111})$
      
      $\phantom{with: }\ket{u_1}=\frac{1}{2}(\ket{0000} - \ket{0011} - \ket{1100} + \ket{1111})$ 
      
      $\phantom{with: }\ket{u_2}=\frac{1}{2}(\ket{0101} + \ket{0110} + \ket{1001} + \ket{1010})$

      $\phantom{with: }\omega=\ee^{\frac{2\ii\pi}{3}}$
      \begin{equation}
        \ket{HD} =\frac{1}{\sqrt6}(\ket{0001} + \ket{0010} + \ket{0100} + \ket{1000} + \sqrt{2}\ket{1111}). 
      \end{equation}
      \begin{equation}
\ket{M_{2222}}=\frac{1}{\sqrt6}\ket{v_1}+\frac{\sqrt6}{4}\ket{v_1}+\frac{1}{\sqrt2}\ket{v_3}
      \end{equation}
      with: $\ket{v_1}=\frac{1}{\sqrt{6}}(\ket{0000} + \ket{0101} - \ket{0110} - \ket{1001} + \ket{1010} + \ket{1111})$
      
      $\phantom{with: }\ket{v_2}=\frac{1}{\sqrt{2}}(\ket{0011} + \ket{1100})$
      
      $\phantom{with: }\ket{v_3}=\frac{1}{\sqrt{2}}(-\ket{0001} + \ket{0010} - \ket{0100} + \ket{0111} + \ket{1000} - \ket{1011} + \ket{1101} - \ket{1110})$
      \caption{4-qubits states for which $|\Delta|$ is maximal. \label{delta_max}}
    \end{figure}
    \medskip

    We examine now whether it is possible to obtain a $\SLOCC$-equivalent to $\ket{\W_4}=\frac12(\ket{0001}+\ket{0010}+\ket{0100}+\ket{1000})$ when a $\cnot$ circuit acts on a fully factorized state. 
    The $\SLOCC$-orbit of $\ket{\W_4}$ belongs to the null cone, which is the algebraic variety defined by the vanishing of all invariants (\textit{i.e.} $B(\ket{\psi})=L(\ket{\psi})=M(\ket{\psi})=D_{xy}(\ket{\psi})=0$).
    The null cone contains 31 $\SLOCC$-orbits and the orbit of $\ket{\W_4}$ is charaterized, inside the null cone, by the evaluation of a vector of 8 polynomial covariants 
    $A, P_B ,P_C^1 ,P_C^2 ,P_D^1,P_D^2,P_F,P_L$ whose definitions have been relegated to Appendix \ref{cov4} (see \cite[Section~III]{2014HLT} for more details).
    More precisely, one has the following criterion :

    \begin{prop}\label{W4criterion}
Let $V_1:=[B,L,M,D_{xy}]$ and $V_2:=[A, P_B ,P_C^1 ,P_C^2 ,P_D^1,P_D^2,P_F,P_L]$ then $\ket{\psi}$ is in the $\SLOCC$-orbit of $\ket{\W_4}$ if and only if  $V1[\ket{\psi}]=[0,0,0,0]$
    and $V2[\ket{\psi}]=[1,1,1,1,0,0,0,0]$. 
      \end{prop}
The use of this criterion makes it possible to answer our initial question :

    \begin{theo} The $\SLOCC$-orbit of $\ket{\W_4}$ cannot be reached when $\XG[4]$ acts on a fully factorized state $\ket{F(u)}$.
    \end{theo}
    \begin{proof}
      We prove that, for any circuit $C\in\XG[4]$, the state $\ket{\psi(u)}=C\ket{F(u)}$ cannot be in the $\SLOCC$-orbit of $\ket{\W_4}$. The algorithm is the following.
      For each $C$ in $\XG[4]$ we solve the system $B(\ket{\psi(u)})=L(\ket{\psi(u)})=M(\ket{\psi(u)})=D_{xy}(\ket{\psi(u)})=0$. The solutions are parametrized vectors $\ket{\psi_1(u)},\dots,\ket{\psi_n(u)}$. Then for each solution $\ket{\psi_i(u)}$ we solve the system $P_D^1(\ket{\psi_i(u)})=P_D^2(\ket{\psi_i(u)}=P_F(\ket{\psi_i(u)})=P_L(\ket{\psi_i(u)}=0$ and
      for each solution $\ket{\psi_{ij}(u)}$ we compute $P_C^2$ : we check that the polynomial $P_C^2(\ket{\psi_{ij}(u)})$ is null for any $u,C,i,j$.
      We deduce from Proposition \ref{W4criterion} that $\ket{\psi(u)}$ is not $\SLOCC$-equivalent to $\ket{\W_4}$.
      The Maple script that implements this algorithm can be downloaded at
      \href{https://github.com/marcbataille/cnot-circuits/blob/master/entanglement/findW4.mpl}{\texttt{https://github.com/marcbataille/cnot-circuits}} and needs a few hours to be executed.
      \end{proof}

\section{Conclusion and perspectives}
The omnipresence and great significance of $\cnot$ gates in Quantum Computation was our main motivation to better understand the quantum circuits build with these gates. First we described the link between $\cnot$ circuits of $n$ qubits and a classical group,  namely $\GL=\SL$. From there we deduced some simplification rules and applied our results to optimization and reduction problems. In Section \ref{general} we proposed some polynomial heuristics to reduce circuits in the general case  and in Section \ref{subgroups} we described a few algorithms to optimize circuits in some special cases. Finally we studied some issues about entanglement and proposed simple constructions to produce some usefull entangled states. Optimization and entanglement are indeed two central topics in QIT since they are related to some important issues on the way to a reliable and functional quantum machine : optimization of circuits for scalable quantum computing and production of entanglement as a physical resource in quantum communication protocols.
We hope the results contained in this paper will contribute to a better understanding of $\cnot$ circuits. We believe that the subject is rich and deserves certainly further investigations. In what follows we try to sketch some directions that future works on
this subject could take.\medskip 

Regarding to the optimization problem  of $\cnot$ circuits, it seems to us that the diversity of situations and methods described in Section \ref{subgroups} tends to show that a polynomial optimization algorithm for the  general case, if it ever exists, will be hard to find out. In our opinion, a more realistic and feasible approach could be a mix of heuristics for the general case (as the Gauss-Jordan algorithm described in Section \ref{general} or the algorithm by Patel \textit{et.al.} described in \cite{2004PMH}) combined with a rich atlas of various methods to optimize or to reduce circuits in special cases (as the atlas we started to built in Section \ref{subgroups}). We will continue to investigate technics of reduction in future works. \medskip

The study of the emergence of entanglement in $\cnot$ circuits done in Section \ref{entanglement} highlights
the utility of these circuits as a practical tool to create entangled states, especially in the case of 3 or 4 qubits systems. When the number of qubits is greater than 4, it would be interesting to know whether it is possible to create generic entanglement as in the 4 qubits case (Theorem \ref{generic4}).
Unfortunately, from 5 qubits the hyperdeterminant is too huge to be computed in a suitable form. However its nullity can be tested thanks to its interpretation in terms of solution of a system of equations \cite[p.~445]{1992GKL} : if $A=\displaystyle\sum_{0\leq i,j,k,l,n\leq 1}\alpha_{ijkln}x_{i}y_{j}z_{k}t_{l}s_{n}$ is the ground form associated to the five qubits state $|\phi\rangle=\displaystyle\sum_{0\leq i,j,k,l,n\leq 1}\alpha_{ijkln}|ijkln\rangle$, the condition $\Delta(|\phi\rangle)=0$ means that the system
 \begin{equation}
	 S_{\phi}:=\{A={d\over dx_{0}}A={d\over dx_{1}}A={d\over dy_{0}}A={d\over dy_{1}}A=\cdots={d\over ds_{0}}A={d\over ds_{1}}A=0\} \end{equation}
 has a solution  $\hat x_{0},\hat x_{1},\hat y_{0},\hat y_{1},\dots,\hat s_{0},\hat s_{1}$ in the variables $x_{0},x_{1},y_{0},y_{1},\dots,s_{0},s_{1}$ such that $(\hat x_{0},\hat x_{1}),(\hat y_{0},\hat y_{1}),\dots,(\hat s_{0},\hat s_{1})\neq (0,0)$. Such a solution is called non trivial.
 Hence a possible approach to show that a state $\ket{\psi}$ is generically entangled when $n>4$ is to prove that the corresponding system has no solution apart from the trivial solutions.

 \section*{Acknowledgment}

 We would like to thank Bruno Schmitt (Ecole Polytechnique F\'ed\'erale de Lausanne) for pointing out that the conjecture $\mathrm{MaxT}(n)=3(n-1)$ for any $n$ was wrong and for mentioning the paper by Patel \textit{et.al.} \cite{2004PMH}.

\bibliographystyle{plain}
\bibliography{cNOT_v3}

\appendix

\section{Some covariant polynomials associated to $4$ qubit systems\label{cov4}}
In this section, we shall explain how to compute the polynomials which are used to determine the $\SLOCC$-orbit of $\ket{\W_4}$ inside the null cone (Section \ref{fourqubits}). We shall first recall the definition of the transvection of two multi-binary forms on the binary variables $x^{(1)}=(x^{(1)}_0,x^{(1)}_1), \dots, x^{(p)}=(x^{(p)}_0,x^{(p)}_1)$
\begin{equation}
(f,g)_{i_1,\dots,i_p}={\mathrm tr}\;\Omega^{i_1}_{x^{(1)}}\dots \Omega_{x^{(p)}}^{i_p}f(x'^{(1)},\dots,x'^{(p)})g(x''^{(1)},\dots,x''^{(p)}),
\end{equation}
where  $\Omega$ is the Cayley operator
\[
\Omega_x=\left|\begin{array}{cc}\partial\over \partial x'_0& \partial\over \partial x''_0
\\ \partial\over \partial x'_1& \partial\over \partial x''_1\end{array}\right|
\]
and $\rm tr$ sends each variables $x', x''$ on $x$ (erases $'$ and $''$). In \cite{2012HLT}, we give a list of generators of the algebra of covariant polynomials for $4$ qubits systems which are obtained by transvection from the ground form
\[{}
A=\sum_{i,j,k,\ell}\alpha_{i,j,k,l}x_{i}y_{j}z_{k}t_{\ell}.
\]
Here we give formulas for some of the polynomials which are used in the paper.
$\begin{array}{cc}&\\
\begin{array}{|c|c|}
\hline \mbox{Symbol}&\mbox{Transvectant}\\\hline
 B_{2200}&\frac12(A,A)^{0011}\\
B_{2020}&\frac12(A,A)^{0101}\\
B_{2002}&\frac12(A,A)^{0110}\\
B_{0220}&\frac12(A,A)^{1001}\\
B_{0202}&\frac12(A,A)^{1010}\\
B_{0022}&\frac12(A,A)^{1100}\\\hline
\end{array}\nonumber&\begin{array}{|c|c|}
\hline \mbox{Symbol}&\mbox{Transvectant}\\
\hline C^1_{1111}&(A,B_{2200})^{1100}+(A,B_{0022})^{0011}
\\\hline
 C_{3111}&\frac13\left((A,B_{2200})^{0100}+(A,B_{2020})^{0010}+(A,B_{2002})^{0001}\right)\\
 C_{1311}&\frac13\left((A,B_{2200})^{1000}+(A,B_{0220})^{0010}+(A,B_{0202})^{0001}\right)\\
 C_{1131}&\frac13\left((A,B_{2020})^{1000}+(A,B_{0220})^{0100}+(A,B_{0022})^{0001}\right)\\
 C_{1113}&\frac13\left((A,B_{2002})^{1000}+(A,B_{0202})^{0100}+(A,B_{0022})^{0010}\right)\\
 \hline
\end{array}\nonumber \end{array}
$

$\begin{array}{cc}\begin{array}{c}\begin{array}{|c|c|}\hline
\mbox{Symbol}&\mbox{Transvectant}\\\hline
D_{2200}&(A,C_{1111}^1)^{0011}\\
D_{2020}&(A,C^1_{1111})^{0101}\\
D_{2002}&(A,C^1_{1111})^{0110}\\
D_{0220}&(A,C^1_{1111})^{1001}\\
D_{0202}&(A,C_{1111}^1)^{1010}\\
D_{0022}&(A,C_{1111})^{1100}\\\hline
D_{4000}&(A,C_{3111})^{0111}\\
D_{0400}&(A,C_{1311})^{1011}\\
D_{0040}&(A,C_{1131})^{1101}\\
D_{0004}&(A,C_{1113})^{1110}\\\hline\end{array}\\ \\
\begin{array}{|c|c|}\hline
\mbox{Symbol}&\mbox{Transvectant}\\\hline
E^1_{3111}&(A,D_{2200})^{0100}+(A,D_{2020})^{0010}+(A,D_{2002})^{0001}\\
E^1_{1311}&(A,D_{2200})^{1000}+(A,D_{0220})^{0010}+(A,D_{0202})^{0001}\\
E^1_{1131}&(A,D_{2020})^{1000}+(A,D_{0220})^{0100}+(A,D_{0022})^{0001}\\
E^1_{1113}&(A,D_{2002})^{1000}+(A,D_{0202})^{0100}+(A,D_{0022})^{0010}\\\hline
\end{array}\end{array}&
	\begin{array}{|c|c|}\hline
\mbox{Symbol}&\mbox{Transvectant}\\\hline
F_{4200}&(A,E^1_{3111})^{0011}\\
F_{4020}&(A,E^1_{3111})^{0101}\\
F_{4002}&(A,E^1_{3111})^{0110}\\
F_{0420}&(A,E^1_{1311})^{1001}\\
F_{0402}&(A,E^1_{1311})^{1010}\\
F_{0042}&(A,E^1_{1131})^{1100}\\
F_{2400}&(A,E^1_{1311})^{0011}\\
F_{2040}&(A,E^1_{1131})^{0101}\\
F_{2004}&(A,E^1_{1113})^{0110}\\
F_{0240}&(A,E^1_{1131})^{1001}\\
F_{0204}&(A,E^1_{1113})^{1010}\\
F_{0024}&(A,E^1_{1113})^{1100}\\\hline
\end{array}\end{array}\nonumber
$\medskip

$\begin{array}{cc}
\begin{array}{|c|c|}\hline
\mbox{Symbol}&\mbox{Transvectant}\\\hline
G_{5111}&(A,F_{4002})^{0001}+(A,F_{4020})^{0010}+(A,F_{4200})^{0100}\\
G_{1511}&(A,F_{0402})^{0001}+(A,F_{0420})^{0010}+(A,F_{2400})^{1000}\\
G_{1151}&(A,F_{0042})^{0001}+(A,F_{0240})^{0100}+(A,F_{2040})^{1000}\\
G_{1115}&(A,F_{0204})^{0100}+(A,F_{0024})^{0010}+(A,F_{2004})^{1000}\\\hline
\end{array}\end{array}\nonumber\nonumber
$\medskip

$\begin{array}{|c|c|}\hline
\mbox{Symbol}&\mbox{Transvectant}\\\hline
H_{4200}&(A,G_{5111})^{1011}\\
H_{4020}&(A,G_{5111})^{1101}\\
H_{4002}&(A,G_{5111})^{1110}\\
H_{0420}&(A,G_{1511})^{1101}\\
H_{0402}&(A,G_{1511})^{1110}\\
H_{0042}&(A,G_{1151})^{1110}\\
H_{2400}&(A,G_{1511}^1)^{0111}\\
H_{2040}&(A,G_{1151})^{0111}\\
H_{2004}&(A,G_{1115}^1)^{0111}\\
H_{0240}&(A,G_{1151})^{1011}\\
H_{0204}&(A,G_{1115})^{1011}\\
H_{0024}&(A,G_{1115}^1)^{1101}\\\hline
\end{array}
\hspace{5mm}	\begin{array}{|c|c|}\hline
\mbox{Symbol}&\mbox{Transvectant}\\\hline
I_{5111}^1&(A,H_{4020})^{0010}+(A,H_{4200})^{0100}+(A,H_{4002})^{0001}\\
I_{1511}^1&(A,H_{0420})^{0010}+(A,H_{2400})^{1000}+(A,H_{4002})^{0001}\\
I_{1151}^1&(A,H_{0240})^{0100}+(A,H_{2040})^{1000}+(A,H_{0042})^{0001}\\
I_{1115}^1&(A,H_{0204})^{0100}+(A,H_{2004})^{1000}+(A,H_{0024})^{0010}
\\\hline
\end{array}
$
$\begin{array}{cc}&\\
\begin{array}{|c|c|}\hline
\mbox{Symbol}&\mbox{Transvectant}\\\hline
J_{4200}&(A,I_{5111}^{1})^{1011}\\
J_{4020}&(A,I_{5111}^{1})^{1101}\\
J_{4002}&(A,I_{5111}^{1})^{1110}\\
J_{0420}&(A,I_{1511}^{1})^{1101}\\
J_{0402}&(A,I_{1511}^{1})^{1110}\\
J_{0042}&(A,I_{1151}^{1})^{1110}\\
J_{2400}&(A,I_{1511}^{1})^{0111}\\
J_{2040}&(A,I_{1151}^{1})^{0111}\\
J_{2004}&(A,I_{1115}^{1})^{0111}\\
J_{0240}&(A,I_{1151}^{1})^{1011}\\
J_{0204}&(A,I_{1115}^{1})^{1011}\\
J_{0024}&(A,I_{1115}^{1})^{1101}\\\hline
\end{array}&\begin{array}{c}
\begin{array}{|c|c|}\hline
\mbox{Symbol}&\mbox{Transvectant}\\\hline
K_{5111}&=(A,J_{4200})^{0100}-(A,J_{4020})^{0010}+(A,J_{4002})^{0001}\\
K_{1511}&=(A,J_{2400})^{1000}-(A,J_{0420})^{0010}+(A,J_{0402})^{0001}\\
K_{1151}&=(A,J_{2040})^{1000}-(A,J_{0240})^{0100}+(A,J_{0042})^{0001}\\
K_{1115}&=(A,J_{2004})^{1000}-(A,J_{0204})^{0110}+(A,J_{0024})^{0010}\\\hline
\end{array}
\\ \vspace{2mm}
\begin{array}{|c|c|}\hline
\mbox{Symbol}&\mbox{Transvectant}\\\hline
L_{6000}&=(A,K_{5111})^{0111}\\
L_{0600}&=(A,K_{1511})^{1011}\\
L_{0060}&=(A,K_{1151})^{1101}\\
L_{0006}&=(A,K_{1115})^{1110}\\\hline\end{array}\end{array}\end{array}\nonumber$
\medskip

We use the following polynomials in order to characterize the $\ket{\W_4}$ $\SLOCC$-orbit:
\[
P_B:=B_{2200} + B_{2020} + B_{2002} + B_{0220} + B_{0202} + B_{0022},
\]
\[
 P_C^1:=C_{3111} + C_{1311} + C_{1131} + C_{1113},
\]
\[
P_C^2:=C_{3111}C_{1311}C_{1131}C_{1113},
\]
\[
P_D^1:=D_{4000} + D_{0400} + D_{0040} + D_{0004},
\]
\[
P_D^2:=D_{2200} + D_{2020} + D_{2002} + D_{0220} + D_{0202} + D_{0022},
\]
\[
P_F:=F_{2220}^1 + F_{2202}^1 + F_{2022}^1 + F_{0222}^1,
\]
\[
P_L:=L_{6000} + L_{0600} + L_{0060} + L_{0006}.
\]
\end{document}